%% file: main_arxiv.tex
\documentclass[11pt]{article}
\input{preamble}
\usepackage{amsmath, amssymb, amsthm}
\usepackage[numbers]{natbib}
\usepackage{cleveref}
\usepackage{algorithm}
\usepackage{algpseudocode}
\usepackage{graphicx}
\usepackage{datetime}
\newtheorem{theorem}{Theorem}
\newtheorem{lemma}[theorem]{Lemma}

\newtheorem{definition}[theorem]{Definition}

\providecommand{\proofname}{Proof}
\newcommand{\jmlrBlackBox}{\rule{1.5ex}{1.5ex}}
\newcommand{\jmlrQED}{\hfill\jmlrBlackBox\par\bigskip}
\renewenvironment{proof}[1][]{
    \ifx\newenvironment#1\newenvironment
      \par\noindent{\bfseries\upshape \proofname\ }
    \else
      \def\temp{#1}%
      \def\noname{noname}%
      \ifx\temp\noname
      \else
        \par\noindent{\bfseries\upshape #1\ }
      \fi
    \fi
    \ignorespaces
  }{%
    \jmlrQED
  }

\crefname{lemma}{lemma}{lemmas}
\crefname{proposition}{proposition}{propositions}
\crefname{remark}{remark}{remarks}
\crefname{corollary}{corollary}{corollaries}
\crefname{definition}{definition}{definitions}
\crefname{conjecture}{conjecture}{conjectures}
\crefname{axiom}{axiom}{axioms}

\newdateformat{monthyeardate}{%
  \monthname[\THEMONTH] \THEYEAR}

\begin{document}
\title{Risk-Sensitive Online Algorithms}

\author{
  Nicolas Christianson\thanks{California Institute of Technology. Email: \href{mailto:nchristianson@caltech.edu}{\nolinkurl{nchristianson@caltech.edu}}}
  \and
  Bo Sun\thanks{University of Waterloo. Email: \href{mailto:bo.sun@uwaterloo.ca}{\nolinkurl{bo.sun@uwaterloo.ca}}}
  \and
  Steven Low\thanks{California Institute of Technology. Email: \href{mailto:slow@caltech.edu}{\nolinkurl{slow@caltech.edu}}}
  \and
  Adam Wierman\thanks{California Institute of Technology. Email: \href{mailto:adamw@caltech.edu}{\nolinkurl{adamw@caltech.edu}}}
}

\date{\monthyeardate\today\thanks{Accepted for presentation at the Conference on Learning Theory (COLT) 2024.}}

\maketitle

\begin{abstract}%
  We study the design of \emph{risk-sensitive online algorithms}, in which risk measures are used in the competitive analysis of randomized online algorithms. We introduce the $\dcvar$-competitive ratio ($\dcr$) using the conditional value-at-risk of an algorithm's cost, which measures the expectation of the $(1-\delta)$-fraction of worst outcomes against the offline optimal cost, and use this measure to study three online optimization problems: continuous-time ski rental, discrete-time ski rental, and one-max search. The structure of the optimal $\dcr$ and algorithm varies significantly between problems: we prove that the optimal $\dcr$ for continuous-time ski rental is $2-2^{-\Theta(\frac{1}{1-\delta})}$, obtained by an algorithm described by a delay differential equation. In contrast, in discrete-time ski rental with buying cost $B$, there is an abrupt phase transition at $\delta = 1 - \Theta(\frac{1}{\log B})$, after which the classic deterministic strategy is optimal. Similarly, one-max search exhibits a phase transition at $\delta = \frac{1}{2}$, after which the classic deterministic strategy is optimal; we also obtain an algorithm that is asymptotically optimal as $\delta \todown 0$ that arises as the solution to a delay differential equation.
\end{abstract}

\clearpage


\section{Introduction}

Randomness can improve decision-making performance in many online problems; for instance, randomization improves the competitive ratio of online ski rental from $2$ to $\frac{e}{e-1}$ \citep{karlinCompetitiveRandomizedAlgorithms1994}, of metrical task systems (MTS) from linear to polylogarithmic in number of states \citep{borodinOptimalOnlineAlgorithm1992,bubeckMetricalTaskSystems2021}, and of online search from polynomial to logarithmic in the fluctuation ratio \citep{el-yanivOptimalSearchOneWay2001,lorenzOptimalAlgorithmsKSearch2009}. However, this improved performance can only be obtained on average over multiple problem instances, as a randomized algorithm may vary wildly in its performance on any particular run.
While this may not pose a concern for decision-making agents facing a large number of problem instances, such variability may be undesirable if an agent has only a small number of instances to solve, or if they are sensitive to risks of a particular magnitude or likelihood.

Numerous fields, including economics, finance, and decision science, have fielded research on risk aversion and alternative \emph{risk measures} that enable modifying decision-making objectives to accommodate these risk preferences (e.g., \cite{markowitzPortfolioSelectionEfficient1959,prattRiskAversionSmall1964,jiaStandardMeasureRisk1996,artznerCoherentMeasuresRisk1999,szegoMeasuresRisk2002}). One of the most well-studied risk measures in recent years, due to its nice properties (as a \emph{coherent} risk measure) and computational tractability, is the \emph{conditional value-at-risk} ($\dcvar$), which measures the expectation of a random loss/reward on its $(1-\delta)$-fraction of worst outcomes \citep{rockafellarOptimizationConditionalValueatrisk2000,rockafellarConditionalValueatriskGeneral2002,acerbiCoherenceExpectedShortfall2002}. $\dcvar$ and other risk measures have been applied to problems spanning finance and insurance \citep{,krokhmalPortfolioOptimizationConditional2001,chiOptimalReinsuranceVaR2011}, energy systems \citep{ndrioPricingConditionalValue2021,madavanRiskBasedHostingCapacity2024}, and robotic control \citep{hakobyanRiskAwareMotionPlanning2019,ahmadiRiskAverseControlCVaR2022}, and have been studied as an objective in place of the expectation in MDPs \citep{chowAlgorithmsCVaROptimization2014,chowRiskSensitiveRobustDecisionMaking2015,keramatiBeingOptimisticBe2020}, bandits \citep{tamkinDistributionallyAwareExplorationCVaR2019,baudryOptimalThompsonSampling2021}, and online learning \citep{even-darRiskSensitiveOnlineLearning2006,uzielGrowthOptimalPortfolioSelection2018,somaOnlineRiskaverseSubmodular2023}.

Despite the significant extent of literature on risk-sensitive algorithms for online learning with the conditional value-at-risk, there has been no work on the design and analysis of \emph{competitive} algorithms for online optimization problems like ski rental, online search, knapsack, function chasing, or MTS with risk-sensitive objectives. These types of online optimization problems have deep connections with online learning \citep{blumOnlineLearningMetrical1997,buchbinderUnifiedAlgorithmsOnline2012,danielyCompetitiveRatioVs2019}, but also substantial qualitative differences due varied problem structures and the competitive analysis framework. Coupled with their practical applications (e.g., \cite{karlinDynamicTCPAcknowledgement2001,zhangCombinatorialSkiRental2020,antoniadisLearningAugmentedDynamicPower2021,sunCompetitiveAlgorithmsOnline2020}), we are thus motivated to ask: how can we design competitive online algorithms when we care about the $\dcvar$ of the cost/reward, and what are the optimal competitive ratios for different problems?

In this work, we begin to work toward answering this question, studying risk sensitivity in the design of competitive online algorithms for online optimization. In particular, we focus on two of the prototypical problems in online optimization: \emph{ski rental}, which, as a special case of MTS, encapsulates the fundamental ``rent vs. buy'' tradeoff inherent in online optimization with switching costs \citep{bansal2CompetitiveAlgorithmOnline2015,antoniadisLearningAugmentedDynamicPower2021}, and \emph{one-max search}, which exhibits a complementary ``accept vs. wait'' tradeoff fundamental to constrained online optimization \citep{geulenRegretMinimizationOnline2010,linCompetitiveOnlineOptimization2019}. While both of these problems are simple to pose, they both reflect crucial components of the difficulty of more complicated online optimization problems, and thus serve as ideal analytic testbeds for investigating the design of risk-sensitive algorithms in online optimization.

\subsection{Contributions}


In this work, we define a novel version of the competitive ratio that penalizes a randomized algorithm's cost via the conditional value at risk ($\dcvar$), which we call the $\dcvar$-competitive ratio ($\dcr$). We then study the design of algorithms for several online problems with the $\dcr$ objective. We make contributions along three fronts:

\noindent\textbf{(1) Optimal Risk-Sensitive Online Algorithms} We find the \emph{optimal} $\dcvar$-competitive algorithm for continuous-time ski rental with any $\delta$ and characterize its $\dcr$ as $2 - 2^{-\Theta\left(\frac{1}{1-\delta}\right)}$. For discrete-time ski rental, we analytically characterize the optimal $\dcvar$-competitive algorithm when $\delta = \calO(\frac{1}{B})$, where $B$ is the buying cost, and we prove that there is a phase transition at $\delta = 1 - \Theta(\frac{1}{\log B})$, after which the optimal $\dcr$ coincides with the deterministic optimal $2-\frac{1}{B}$. Finally, we propose an algorithm for one-max search whose $\dcr$ is asymptotically optimal for small $\delta$, and we prove that one-max search exhibits a phase transition at $\delta = \frac{1}{2}$, after which the optimal $\dcr$ coincides with the deterministic optimal $\sqrt{\theta}$, where $\theta$ is the so-called ``fluctuation ratio'' of the problem.

\noindent\textbf{(2) Techniques} For continuous-time ski rental and one-max search, we show that the conditional value-at-risk of an algorithm's cost can be written as an integral expression of its inverse cumulative distribution function. This parametrization is useful both for proving analytic bounds on algorithms' $\dcr$, and as a source for optimal algorithms for these problems: it is through this formulation that we obtain the delay differential equation describing the optimal algorithm for continuous-time ski rental, and similarly how we obtain our algorithm for one-max search, which is asymptotically optimal when $\delta$ is small. For both versions of ski rental, our results rely on structural characterizations of the optimal algorithm which, while evocative of similar results from the ski rental literature, require significantly more care due to the complicated behavior of the conditional value-at-risk.

\noindent\textbf{(3) Insights} We gain several new insights from our results. The phase transitions in the discrete-time ski rental and one-max search problems, where $\delta$ sufficiently large implies that the optimal $\dcr$ is the deterministic optimal competitive ratio, suggests that there is a sharp limit to the benefit that randomization can yield in certain risk-sensitive online problems. Moreover, the qualitative difference between the continuous- and discrete-time ski rental problems -- namely, the fact that the latter has a phase transition while the former does not -- indicates that continuous and discrete problems may, in general, behave differently when risk sensitivity is introduced.

\subsection{Related Work}

\paragraph{Risk-aware online algorithms}

As mentioned earlier, while numerous problems in MDPs, bandits, and online learning have been studied with the conditional value-at-risk and other risk measures penalizing the objective, we are not aware of existing work in the literature designing competitive online algorithms for online optimization with such objectives. Even-Dar et. al \cite{even-darRiskSensitiveOnlineLearning2006} consider the related problem of online learning with expert advice and rewards depending on the Sharpe ratio and mean-variance risk measures; they prove lower bounds precluding the possibility of obtaining sublinear regret in this setting as well as upper bounds for several relaxed objectives. While their work also considers a notion of competitive ratio against the best fixed expert in its lower bounds, their upper bounds focus on regret-style results, and their problem setting is markedly different from those we consider.
A related problem is the demonstration of \emph{high-probability guarantees} on the competitive ratio of randomized online algorithms, which was studied by Komm et al.~\cite{kommRandomizedOnlineComputation2022} for a general class of online problems. However, their work is concerned with proving that existing algorithms have performance close to some nominal value with high probability, rather than designing new algorithms that are provably optimal given an agent's particular risk preferences and the distribution over algorithm performance.

Closest to our current work is the recent paper of Dinitz et al.~\cite{dinitzControllingTailRisk2024} on risk-constrained algorithms for ski rental, where the objective remains to minimize the competitive ratio defined in terms of expected cost, but algorithms must satisfy additional constraints on the likelihood that their cost will exceed a specified value. This amounts to imposing constraints on the \emph{value-at-risk ($\var$)}, or quantiles, of the algorithm's competitive ratio; in contrast, we focus on algorithms that are optimal for a \emph{risk-sensitive objective} involving the conditional value at risk, which considers not just the likelihood of exceeding a certain value, but the expectation over the resulting tail of the distribution. $\var$ and $\cvar$ have been exhaustively compared in the financial literature (e.g., \cite{sarykalinValueatRiskVsConditional2008,chiOptimalReinsuranceVaR2011}), and $\cvar$, which is a so-called \emph{coherent} risk measure, often exhibits more favorable robustness and handling of tail events than $\var$, which is not coherent \citep{artznerCoherentMeasuresRisk1999}. Indeed, $\var$ is very sensitive to problem structure and parameter selection, leading to the interesting non-continuous behavior in solution structure observed in \cite{dinitzControllingTailRisk2024}. $\cvar$ does not beget such sensitivity, but influences the solution structure in its own unique way: for continuous-time ski rental and one-max search, we obtain algorithms that result from the solution of delay differential equations.

\paragraph{Beyond worst-case analysis of algorithms} The strengthening and weakening of the adversary as $\delta$ is varied in the $\dcvar$-competitive ratio is similar in spirit to beyond worst-case analysis, where the adversary is weakened or additional information is provided to enable improved bounds over the pessimistic and unrealistic adversarial setting. There is a significant breadth of work in the literature applying these ideas to online optimization and other online problems. Some notable directions on this subject are smoothed analysis \citep{schaferTopologyMattersSmoothed2004,haghtalabSmoothedAnalysisAdaptive2022}, in which an adversary's decision is tempered with stochastic noise; algorithms with advice \citep{bockenhauerOnlineAlgorithmsAdvice2017,emekOnlineComputationAdvice2011}, in which an algorithm receives a small number of accurate bits of information about the problem instance in advance; and algorithms with predictions \citep{lykourisCompetitiveCachingMachine2018,purohitImprovingOnlineAlgorithms2018,sunParetoOptimalLearningAugmentedAlgorithms2021,christiansonOptimalRobustnessconsistencyTradeoffs2023,lechowiczOnlineConversionSwitching2024}, in which algorithms are augmented with potentially unreliable predictions about the problem instance, and algorithms seek to exploit these predictions when they are accurate while maintaining worst-case guarantees when they are not.

\subsection{Notation}
Throughout, capital letters (e.g., $X$) refer to random variables on $\R$, which we interchangeably refer to via their measures (e.g., $\mu \in \Borel(\R)$) or their cumulative distribution functions (e.g., $F_X(x) = \mu(-\infty, x]$). Given a random variable $X$ with support bounded in the interval $[a, b]$, we define its inverse CDF as $F_X^{-1}(p) = \inf\{x \in [a, b] : F_X(x) \geq p\}$; note that, given the bounded support, this definition agrees with the standard definition (where the infimum is taken over all of $\R$) for all $p \in (0, 1]$, with the only disagreement being at $p = 0$, where $F_X^{-1}(0) = \essinf X$, whereas the standard definition yields $-\infty$; this variant is well-established in the literature (e.g., \cite[Definition 1.16]{wittingMathematischeStatistik1985}). We use this definition to ensure finiteness of $F_X^{-1}$ on $[0, 1]$, and the bounded support of $X$ will be clear by context whenever the inverse CDF is discussed. $\R_+$ denotes the nonnegative reals and $\Rpp$ the strictly positive reals, and $\N$ denotes the natural numbers. The notation $[\cdot]^+$ refers to the $\max\{\cdot, 0\}$ function, and for any $N \in \N$, we write $[N] = \{1, \ldots, N\}$ and denote by $\Delta_N$ the $N$-dimensional probability simplex. For a vector $\xvec \in \R^n$, we denote its $i$th entry $x_i$. The function $W_k(x)$ refers to the $k$th branch of the Lambert $W$ function, which is defined as a solution to $W_k(x)e^{W_k(x)} = x$ (see, e.g., \cite{corlessLambertFunction1996}).

\section{Background \& Preliminaries \label{section:background}}

In this section, we introduce risk measures and the conditional value-at-risk, and give overviews of the three online problems we study in this work.

\subsection{Risk Measures and the Conditional Value-at-Risk}

A \emph{risk measure} is a mapping from the set of $\R$-valued random variables to $\R$ that gives a deterministic valuation of the \emph{risk} associated with a particular random loss. As risk preferences can vary by decision-making agent and application, many different risk measures have been introduced and studied in the literature (see, e.g., \cite[Chapter 6]{shapiroLecturesStochasticProgramming2009} for several examples). A prominent class of measures that has emerged in practice due to its favorable properties is the set of \emph{coherent risk measures} \citep{artznerCoherentMeasuresRisk1999}. Perhaps one of the most well-studied coherent risk measures in recent years is the \emph{conditional value-at-risk} ($\cvar$): the $\cvar$ at probability level $\delta$ of a random variable $X$, written $\cvar_\delta[X]$, is the expectation of $X$ on the $\delta$-tail of its distribution, i.e., its $(1-\delta)$-fraction of worst outcomes. It can be defined in several ways:

\begin{definition}[Conditional Value-at-Risk]
Let $X$ be a real-valued random variable with CDF $F_X$. If $X$ has a density $f$, then for $\delta \in [0, 1)$ the conditional value-at risk at level $\delta$ of $X$ is defined as the expectation of $X$, conditional on its outcome lying in the $\delta$-tail of its distribution \citep{rockafellarOptimizationConditionalValueatrisk2000}:
$$\cvar_\delta[X] = \E[X | X \geq F_X^{-1}(\delta)].$$
For a general random loss $X$ with probability measure $\mu$, $\cvar_\delta[X]$ can be defined in several equivalent ways \citep{rockafellarConditionalValueatriskGeneral2002,acerbiCoherenceExpectedShortfall2002,duchiLearningModelsUniform2021}:
\begin{equation} \label{eq:cvar_definition}
    \cvar_\delta[X] = \inf_{t \in \R} \left\{t + \frac{1}{1-\delta} \E[X - t]^+\right\} = \frac{1}{1-\delta}\int_\delta^1 F_X^{-1}(p)\,\der p = \sup_{\nu \in \mathscr{Q}} \E_{Y \sim \nu}[Y],
\end{equation}
where in the final expression, $\mathscr{Q}$ is an uncertainty set of probability measures defined as
$$\mathscr{Q} = \{\nu : \mu = \beta\nu + (1-\beta)\rho\text{ for some measure $\rho$ and $\beta \in [1-\delta, 1]$}\}.$$
\end{definition}

The first expression in \eqref{eq:cvar_definition} is a variational form of $\cvar_\delta$, and is useful for tractable formulations of risk-sensitive optimization problems. The latter two expressions highlight the intuition that $\cvar_\delta[X]$ computes the expected loss of $X$ on the worst $(1-\delta)$-fraction of outcomes in its distribution, or, in the parlance of \cite{duchiLearningModelsUniform2021} which we sometimes adopt, on the ``worst $(1-\delta)$-sized subpopulation''.

From the above definition it is clear that $\cvar_0[X] = \E[X]$ and $\lim_{\delta \toup 1} \cvar_\delta[X] \to \esssup X$, the largest value that $X$ can take \citep{mafusalovEstimationAsymptoticsBuffered2018}; we thus define $\cvar_1[X] \coloneqq \esssup X$, so that $\cvar_\delta$ is defined for all $\delta \in [0, 1]$.

\subsection{Online Algorithms and Competitive Analysis}

In the study of online algorithms, algorithm performance is typically measured via the \emph{competitive ratio}, or the worst case ratio in (expected) cost between an algorithm and the offline optimal strategy that knows all uncertainty in advance.

\begin{definition}[Competitive ratio] \label{defn:comp_rat}
    Consider an online problem with uncertainty drawn adversarially from a set of instances $\calI$. Let $\alg$ be a deterministic online algorithm for the problem, and let $\opt$ be the offline optimal algorithm. $\alg$'s \textbf{competitive ratio ($\compratbold$)} is the worst-case ratio in cost between $\alg$ and $\opt$ over all problem instances:
    $$\comprat(\alg) \coloneqq \sup_{I \in \calI} \frac{\alg(I)}{\opt(I)}.$$
    If $\alg$ has competitive ratio $C$, it is also called \textbf{$\mathbf{C}$-competitive}. If $\alg$ is a randomized algorithm, then the competitive ratio is defined with its expected cost:
    $$\comprat(\alg) \coloneqq \sup_{I \in \calI} \frac{\E[\alg(I)]}{\opt(I)},$$
    where the expectation is taken over $\alg$'s randomness.
\end{definition}

In our work, we introduce a new version of the competitive ratio for randomized algorithms that goes beyond expected performance: instead, we penalize a randomized algorithm via the ratio between the conditional value-at-risk of its cost and the offline optimal's cost, terming this metric the \emph{$\dcvar$-competitive ratio} (abbreviated $\dcr$).

\begin{definition}[$\dcvar$-Competitive Ratio]
    Let $\alg$ be a randomized algorithm, and let $\opt$ be the offline optimal algorithm. The \textbf{$\dcvarbold$-Competitive Ratio ($\dcrbold$)} is defined as the worst-case ratio between the $\dcvar$ of $\alg$'s cost and the offline optimal cost:
    $$\dcr(\alg) \coloneqq \sup_{I \in \calI} \frac{\cvard[\alg(I)]}{\opt(I)},$$
    where the $\cvard$ is taken over $\alg$'s randomness.
\end{definition}

It is immediately clear that any deterministic algorithm has $\dcr = \comprat$ for all $\delta \in [0, 1]$, while for randomized algorithms these metrics will generally differ for $\delta > 0$. Note that, given the definition of $\dcvar$ as focusing on the \emph{worst} $(1-\delta)$-fraction of a distribution, the $\dcr$ may also be interpreted as a metric that gives the adversary additional power to shift the distribution of the algorithm's randomness. Under this interpretation, the $\dcr$ may be viewed as an interpolation between the classic randomized case where the adversary has no power over $\alg$'s randomness ($\delta = 0$), and the case where the adversary has full control over $\alg$'s randomness and determinism is optimal ($\delta = 1$). This model can also be seen as a complement to the oblivious adversary, which knows $\alg$ but cannot see the realization of its randomness, and the adaptive adversary, which sees all random outcomes; in the $\dcr$ case, while the adversary does not see $\alg$'s random outcome directly, it has the ability to control this outcome in a way limited by the $\cvar$.


\noindent

\subsection{Online Problems Studied \label{section:problems_studied}}
We now provide a brief introduction for each of the three problems we study in this work.
\paragraph{Continuous-Time Ski Rental}

In the \emph{continuous-time ski rental (CSR)} problem, a player faces a ski season of unknown and adversarially-chosen duration $s \in \Rpp$, and must choose how long to rent skis before purchasing them. In particular, the player pays cost equal to the duration of renting, and cost $B$ for purchasing the skis. Deterministic algorithms for ski rental are wholly determined by the day $x \in \Rpp$ on which the player stops renting and purchases the skis: an algorithm that rents until day $x$ and then purchases pays cost $s \cdot \indic_{x > s} + (x + B) \cdot \indic_{x \leq s}$. Randomized algorithms can be described by a random variable $X$ over purchase days, in which case the algorithm pays (random) cost $s \cdot \indic_{X > s} + (X + B) \cdot \indic_{X \leq s}$. Given knowledge of the total number of skiing days $s$, the offline optimal strategy is to rent for the entire season if $s < B$, incurring cost $s$, and to buy immediately otherwise, yielding cost $B$. Defining $\alpha_\delta^{\text{CSR}, \mu}$ as the $\dcr$ of a strategy $X \sim \mu$, we have
$$\alpha_\delta^{\text{CSR}, \mu} \coloneqq \sup_{s \in \Rpp}\alpha_\delta^{\text{CSR}, \mu}(s) \coloneqq \sup_{s \in \Rpp} \frac{\cvard[s \cdot \indic_{X > s} + (X + B) \cdot \indic_{X \leq s}]}{\min\{s, B\}},$$
where $\alpha_\delta^{\text{CSR}, \mu}(s)$ denotes the competitive ratio of the strategy $\mu$ when the adversary's decision is $s$. We denote by $\alpha_\delta^{\text{CSR}, *}$ the smallest $\dcr$ of any strategy. We will omit the ``CSR'' in the superscript when it is clear through context that we are discussing the continuous-time ski rental problem.

We will assume without loss of generality that $B = 1$. It is well known that $\alpha_1^{\text{CSR}, *} = 2$, which is achieved by purchasing skis deterministically at time $1$, and $\alpha_0^{\text{CSR}, *} = \frac{e}{e-1}$, which is achieved by a probability density supported on the interval $[0, 1]$ \citep{karlinCompetitiveSnoopyCaching1988,karlinCompetitiveRandomizedAlgorithms1994}. In the following lemma, which is proved in Appendix~\ref{appendix:lemma:cont_ski_rental_01_support}, we show that when considering $\dcr$ as a performance metric with general $\delta \in [0, 1]$, we may similarly restrict our focus to probability measures with support on $[0, 1]$.

\begin{lemma} \label{lemma:cont_ski_rental_01_support}
    Let $\mu_1$ be a distribution on $\R_+$. There is a distribution $\mu_2$ with support in $[0, 1]$ such that, for any $\delta \in [0, 1]$, $\mu_2$ has no worse $\dcr$ than $\mu_1$: $\alpha_\delta^{\text{CSR}, \mu_2} \leq \alpha_\delta^{\text{CSR}, \mu_1}$.
\end{lemma}

An important consequence of the preceding lemma is that we can restrict the adversary's decisions to $s \in (0, 1]$, since choosing $s > 1$ will not change the $\dcr$ for any random strategy supported on $[0, 1]$. Thus for $\mu$ supported in $[0, 1]$, we have $\alpha_\delta^{\text{CSR}, \mu} = \sup_{s \in (0, 1]}\alpha_\delta^{\text{CSR}, \mu}(s)$.

\paragraph{Discrete-Time Ski Rental} In the \emph{discrete-time ski rental (DSR)} problem, a player faces a ski season of unknown and adversarially-chosen duration $s \in \N$ and must choose an integer number of days to rent skis before purchasing them; renting for a day costs $1$, and purchasing skis has an integer cost $B \geq 2$. The cost structure is essentially identical to the continuous-time case, except the algorithm's and adversary's decisions are restricted to lie in $\N$: if a player buys skis at the start of day $x \in \N$ and the true season duration is $s \in \N$, their cost will be $s \cdot \indic_{x > s} + (B + x - 1) \cdot \indic_{x \leq s}$. Thus for a random strategy $X \sim \mu$ with support on $\N$, the $\dcr$ is defined as follows:
$$\alpha_\delta^{\text{DSR}(B), \mu} \coloneqq \sup_{s \in \N}\alpha_\delta^{\text{DSR}(B), \mu}(s) \coloneqq \sup_{s \in \N} \frac{\cvard[s \cdot \indic_{X > s} + (B + X - 1) \cdot \indic_{X \leq s}]}{\min\{s, B\}}.$$
As in the continuous-time setting, we denote by $\alpha_\delta^{\text{DSR}(B), *}$ the smallest $\dcr$ of any strategy, and will omit the ``DSR'' from the superscript when it is clear from context, instead writing just $\alpha_\delta^{B, \mu}$. It is well known that $\alpha_1^{B, *} = 2 - \frac{1}{B}$, achieved by deterministically purchasing skis at the start of day $B$, and $\alpha_0^{B, *} = \frac{1}{1 - \left(1 - B^{-1}\right)^{B}}$, which approaches $\alpha_0^{\text{CSR}, *} = \frac{e}{e-1}$ as $B \to \infty$. Following identical reasoning as in Lemma~\ref{lemma:cont_ski_rental_01_support} for the continuous-time setting, we may without loss of generality restrict our focus to strategies $\mu$ with support on $[B]$, and likewise to adversary decisions in $[B]$. Finally, note that the discrete problem is easier than the continuous-time problem, i.e., $\alpha_\delta^{\text{DSR}(B), *} \leq \alpha_\delta^{\text{CSR}, *}$ for all $\delta \in [0, 1]$ and $B \in \N$; this is because we can embed DSR into the continuous setting by restricting the continuous-time adversary to choose season durations $\{\frac{1}{B}, \ldots, \frac{B-1}{B}, 1\}$ and reducing the player's buying cost by $\frac{1}{B}$.

\paragraph{One-Max Search}
In the \emph{one-max search (OMS)} problem, a player faces a sequence of prices $v_t \in [L, U]$ arriving online, with $U \geq L > 0$ known upper and lower bounds on the price sequence; we define the \emph{fluctuation ratio} $\theta = \frac{U}{L}$ as the ratio between these bounds. The player's goal is to sell an indivisible item for the greatest possible price: after observing a price $v_t$, the player can choose to either accept the price and earn profit $v_t$, or to wait and observe the next price. The duration $T \in \N$ of the sequence is \emph{a priori} unknown to the player, and if $T$ elapses and the player has not yet sold the item, they sell it for the smallest possible price $L$ in a compulsory trade. In the deterministic setting, the player aims to minimize their competitive ratio, defined as the worst-case ratio between the price accepted by the player and the optimal price $\vmax = \max_t v_t$:
$$\comprat(\alg) \coloneqq \sup_{(v_1, \ldots, v_T) \in [L, U]^T} \frac{\opt(v_1, \ldots, v_T)}{\alg(v_1, \ldots, v_T)} = \sup_{(v_1, \ldots, v_T) \in [L, U]^T} \frac{\vmax}{\alg(v_1, \ldots, v_T)},$$
with an expectation around $\alg$ in the denominator if the algorithm is randomized.
Note that this definition of competitive ratio differs from that in Definition~\ref{defn:comp_rat} because this is a reward maximization, rather than a loss minimization, problem. Likewise, when discussing the conditional value-at-risk and $\dcr$ in this setting, we will use the reward formulation, which is the expected reward on the worst (i.e., smallest) $(1-\delta)$-fraction of outcomes in the reward distribution \citep{acerbiCoherenceExpectedShortfall2002}:
\begin{equation} \label{eq:integral_form_cvar_reward}
    \dcvar[X] = \sup_{t \in \R}\left\{t - \frac{1}{1-\delta}\E[t - X]^+\right\} = \frac{1}{1-\delta}\int_0^{1-\delta} F_X^{-1}(p)\,\der p.
\end{equation}
While these definitions of $\dcvar$ and $\comprat$ differ from those employed in discussion of the ski rental problem, we will generally not distinguish which version we are using throughout this paper, as it will be clear from context which problem (and hence which version) we are concerned with.

The one-max search problem was first studied in \cite{el-yanivOptimalSearchOneWay2001}, which found that the optimal deterministic competitive ratio is $\sqrt{\theta}$, achieved by a ``reservation price'' or ``threshold'' \citep{sunParetoOptimalLearningAugmentedAlgorithms2021} algorithm that accepts the first price above $\sqrt{LU}$. Randomization improves the competitive ratio exponentially: the optimal randomized competitive ratio is ${1 + W_0\left(\frac{\theta - 1}{e}\right)} = \Theta(\log \theta)$, where $W_0$ is the principal branch of the Lambert $W$ function \citep{el-yanivOptimalSearchOneWay2001,lorenzOptimalAlgorithmsKSearch2009}. In this work, we restrict our focus to the class of \emph{random threshold algorithms} without loss of generality;\footnote{This restriction is made without loss of generality, in the sense that any randomized algorithm for OMS with $\dcr$ $\alpha$ can be approximated by a random threshold policy with $\dcr$ $\alpha + \epsilon$, with $\epsilon$ arbitrarily small; see Appendix~\ref{appendix:random_threshold_restriction} for a full explanation.} such algorithms fix a distribution $\mu$ with support on $[L, U]$, draw a threshold $X \sim \mu$ at random, and accept the first price above $X$, earning profit $L \cdot \indic_{X > \vmax} + X \cdot \indic_{X \leq \vmax}$.\footnote{When the player uses a random threshold algorithm, we may assume that they earn profit exactly $X$ whenever $X \leq \vmax$, since if the adversary (who is unaware of $X$) plays a sequence of prices that increases by $\epsilon$ at every time until reaching $\vmax$, then the player will accept the first price above $X$, which will be at most $X + \epsilon$; sending $\epsilon \to 0$, the player's profit is exactly $X$.}
Thus the $\dcr$ of a threshold algorithm is defined:
\begin{equation} \label{eq:dcr_oms}
    \alpha_\delta^{\mathrm{OMS}(\theta), \mu} \coloneqq \sup_{v \in [L, U]} \alpha_\delta^{\mathrm{OMS}(\theta), \mu}(v) \coloneqq \sup_{v \in [L, U]} \frac{v}{\dcvar[L \cdot \indic_{X > v} + X \cdot \indic_{X \leq v}]},
\end{equation}
where we denote by $\alpha_\delta^{\mathrm{OMS}(\theta), \mu}(v)$ the $\dcr$ of one-max search with fluctuation ratio $\theta$ restricted to price sequences with maximal price $v$, which is wholly determined by the distribution of the random threshold $X$. As in the ski rental problems, we denote by $\alpha_\delta^{\mathrm{OMS}(\theta), *}$ the optimal $\dcr$ for the problem, and we omit ``OMS'' from the superscript when the problem is clear from context.

\section{\hspace{-0.8mm}\texorpdfstring{$\dcvarbold$}{$\dcvar$}-Competitive Continuous-Time Ski Rental: \\ Optimal Algorithm and Lower Bound} \label{section:cont_ski_rental}

As noted in the previous section, the optimal deterministic competitive ratio for continuous-time ski rental is $\alpha_1^* = 2$, and the optimal randomized competitive ratio is $\alpha_0^* = \frac{e}{e-1}$. This immediately motivates the question of what the optimal $\dcr$ is, for arbitrary $\delta \in (0, 1)$: how does $\alpha_\delta^*$ grow as $\delta \toup 1$? And does $\alpha_\delta^*$ strictly improve on the deterministic worst case of 2 whenever $\delta < 1$?

The classical approach for obtaining the optimal randomized algorithm for continuous-time ski rental is to assume that the optimal purchase distribution has a probability density $p$ supported on $[0, 1]$, use this to express the expected cost of the algorithm given any adversary decision, and write out the inequalities that must be satisfied for the algorithm to be $\alpha$-competitive for some constant $\alpha$ \citep{karlinCompetitiveRandomizedAlgorithms1994}:
$$\int_0^s (t + 1)p(t)\,\der t + s\int_s^1 p(t)\,\der t \leq \alpha s \quad\text{for all $s \in [0, 1]$}.$$
The optimal $p$ is found by setting these inequalities to equalities, differentiating with respect to $s$, solving the resulting differential equations, and choosing $\alpha$ to ensure $p$ integrates to $1$. If we attempt to apply this methodology to the problem with the $\dcr$ objective, we are met with two challenges: first, while the assumption that the optimal strategy has a density and the trick of setting the above inequalities to equalities works in the expected cost setting, there is no guarantee that these assumptions can be imposed without loss of generality when the expectation is replaced with $\dcvar$. Second, and more formidably, even if we can restrict to densities, the limits of integration in the $\dcvar$ case will depend on the particular quantile structure induced by $p$. If $X \sim p$ and $F_X(s) \geq 1-\delta$, using the definition of $\dcvar$ as the expected cost on the worst $(1-\delta)$-sized subpopulation of the loss, one can compute
$$\dcvar[s \cdot \indic_{X > s} + (X + 1) \cdot \indic_{X \leq s}] = \int_{F_X^{-1}\left(F_X(s) - (1-\delta)\right)}^s (t+1)p(t)\,\der t,$$
whose lower limit of integration depends on $p$'s quantile structure in a nontrivial way (i.e., it is the smallest point with CDF value equal to $F_X(s) - (1-\delta)$), significantly complicating the formulation of any differential equation we could construct using this expression.

Not all is lost, however: if we instead take inspiration from the formulation of $\dcvar$ in terms of the inverse CDF of the loss distribution, it is possible to formulate the $\dcvar$ of the loss of an \emph{arbitrary} strategy $X$ (i.e., not necessarily one with a density) in terms of the inverse CDF of $X$. We state this result formally in the following lemma, which is proved in Appendix~\ref{appendix:lemma:cvar_cost_integral_general}.

\begin{lemma} \label{lemma:cvar_cost_integral_general}
    Let $X$ be a random variable supported in $[0, 1]$, and fix an adversary decision $s \in (0, 1]$. Then the $\cvar_\delta$ of the cost incurred by the algorithm playing $X$ is
    $$\cvar_\delta[s \cdot \indic_{X > s} + (X + 1) \cdot \indic_{X \leq s}] = \begin{cases} \frac{1}{1-\delta}\!\left[(1 \!- \!\delta \!- \!F_X(s)) s + \!\int_0^{F_X(s)} (1 \!+\! F_X^{-1}(t))\,\der t\right] & \text{if $F_X(s) \leq 1 \!-\! \delta$} \\ \frac{1}{1-\delta}\int_{F_X(s) - (1 - \delta)}^{F_X(s)} (1 + F_X^{-1}(t))\,\der t &\text{otherwise.}\end{cases}$$
\end{lemma}

While the integral representation of the algorithm's cost given in Lemma~\ref{lemma:cvar_cost_integral_general} depends on both the CDF $F_X$ and the inverse CDF $F_X^{-1}$, it is possible to remove the CDF when it is continuous and strictly increasing on $[0, 1]$; in this case, for any $s \in [0, 1]$, we may define a corresponding $y = F_X(s)$ and replace $F_X(s)$ with $y$ and $s$ with $F_X^{-1}(y)$ in Lemma~\ref{lemma:cvar_cost_integral_general}'s representation. We will show later that the optimal strategy indeed has such a continuous and strictly increasing $F_X$ (see Lemmas~\ref{lemma:cont_time_optimal_strictly_increasing} and \ref{lemma:cont_time_optimal_continuous} in Appendix~\ref{appendix:theorem:optimal_cont_time_ski_rental}).

As a first application of the representation for the $\dcvar$-cost in Lemma~\ref{lemma:cvar_cost_integral_general}, we construct in the following theorem a family of densities parametrized by $\delta$ whose $\dcr$ we can compute analytically, giving an upper bound on $\alpha_\delta^*$, and in particular showing that $\alpha_\delta^* < 2$ for all $\delta \in [0, 1)$.


\begin{theorem} \label{theorem:cvar_skirental_firstalg}
    Let $p_\delta(x)$ be a probability density defined on the unit interval $[0, 1]$ as
    $$p_\delta(x) = \frac{(1-\delta)(1-e^{\frac{c}{1-\delta}})}{c(e^{\frac{c}{1-\delta}}(x-1)-x)},$$
    with constant $c = -\frac{1 + 2 W_{-1}\left(\nicefrac{-1}{2\sqrt{e}}\right)}{2} \approx 1.25643$, where $W_{-1}$ is the $-1$ branch of the Lambert $W$ function. Then the strategy that buys on day $X \sim p_\delta$ achieves competitive ratio
    $$\alpha_\delta^{p_\delta} = 2 - \frac{1}{e^{\frac{c}{1-\delta}}-1}.$$
    In particular, $\alpha_\delta^{p_\delta} < 2$ for all $\delta \in [0, 1)$.
\end{theorem}

We present a proof of this theorem in Appendix~\ref{appendix:theorem:cvar_skirental_firstalg}. Our approach is to compute the inverse CDF corresponding to the proposed density and reformulate the inequalities defining $\alpha_\delta^{p_\delta}$-competitiveness using Lemma~\ref{lemma:cvar_cost_integral_general}; the rest of the work is concerned with computing $\alpha_\delta^{p_\delta}$.

Intuitively, the strategy $p_\delta$ in Theorem~\ref{theorem:cvar_skirental_firstalg} behaves like one might expect a good algorithm for ski rental with the $\dcr$ metric should: it assigns less probability mass to earlier times and more to later times, and as $\delta$ increases, it shifts mass from earlier times to later times. However, the algorithm cannot be optimal, since $\alpha_0^{p_0} = 2 - \frac{1}{e^c - 1} \approx 1.60$, which is larger -- though only slightly -- than the randomized optimal $\frac{e}{e-1} \approx 1.58$. This motivates the question: is it possible to leverage the representation in Lemma~\ref{lemma:cvar_cost_integral_general} to obtain the optimal algorithm for continuous-time ski rental with the $\dcr$ objective?
In the following theorem, we answer this question in the affirmative: in particular, the optimal algorithm's inverse CDF is the solution to a delay differential equation defined on the interval $[0, 1]$.

\begin{theorem} \label{theorem:optimal_cont_time_ski_rental}
    For any $\delta \in [0, 1)$, let $\phi : [0, 1] \to [0, 1]$ be the solution to the delay differential equation
    \begin{align*}
        \phi'(t) &= \frac{1}{\alpha(1-\delta)}\left[\phi(t) - \phi(t - (1-\delta))\right] & \text{for $t \in [1-\delta, 1]$},
    \end{align*}
    with initial condition $\phi(t) = \log\left(1 + \frac{t}{(\alpha-1)(1-\delta)}\right)$ on $t \in [0, 1-\delta]$. Then when $\alpha = \alpha_\delta^*$, $\phi$ is the inverse CDF of the \textbf{unique optimal} strategy for continuous-time ski rental with the $\dcr$ metric.
\end{theorem}

We prove this theorem in Appendix~\ref{appendix:theorem:optimal_cont_time_ski_rental}. The crux of the proof is a pair of structural lemmas (Lemmas~\ref{lemma:s_1_tight} and \ref{lemma:any_s_tight}) which establish that, for any $\delta \in [0, 1)$, the optimal algorithm $\mu^*$ is indifferent to the adversary's decision, i.e., $\alpha_\delta^{\mu^*}(s) = \alpha_\delta^*$ for all $s \in (0, 1]$. This is analogous to the trick of ``setting the inequalities to equalities'' in the classical version of ski rental \citep{karlinCompetitiveRandomizedAlgorithms1994}, but requires a great deal more care in the continuous-time $\dcvar$ setting to make rigorous. In addition, this result depends on the fact that $\alpha_\delta^* < 2$, which we showed in Theorem~\ref{theorem:cvar_skirental_firstalg}. With this property established, we can apply Lemma~\ref{lemma:cvar_cost_integral_general} to pose a family of integral equations constraining the optimal inverse CDF, which can be transformed to obtain the delay differential equation in Theorem~\ref{theorem:optimal_cont_time_ski_rental}.

Note, however, that the delay differential equation yielding the optimal inverse CDF depends on the optimal $\dcr$ $\alpha_\delta^*$, which we have no analytic form for. Fortunately, the solution $\phi$ to the delay differential equation in Theorem~\ref{theorem:optimal_cont_time_ski_rental} has the property that $\phi(t)$ is strictly decreasing in $\alpha$ for each $t \in (0, 1]$ (see Appendix~\ref{appendix:opt_soln_delay_diffeq_strictly_decreasing_alpha}); since the optimal inverse CDF must have $\phi(1) = 1$ (see Lemma~\ref{lemma:cont_time_optimal_continuous} in the appendix), $\alpha_\delta^*$ is equivalently defined as the unique choice of $\alpha$ for which the solution to the above delay differential equation satisfies $\phi(1) = 1$. We may thus determine $\alpha_\delta^*$ via binary search: given some $\alpha$, we solve the delay differential equation numerically and evaluate $\phi(1)$; if $\phi(1) > 1$, then we decrease $\alpha$, and if $\phi(1) < 1$, we increase $\alpha$. We plot the optimal $\dcr$ obtained via this binary search methodology (with delay differential equations solved numerically in Mathematica) alongside the upper bound from Theorem~\ref{theorem:cvar_skirental_firstalg} in Figure~\ref{fig:cont_time_crs}.

While Theorem~\ref{theorem:optimal_cont_time_ski_rental} gives us a method for computing the optimal strategy and $\dcr$ for continuous-time ski rental, it does not give an analytic form of this solution or metric. An analytic form of $\phi(t)$ can be obtained when $\delta \leq \frac{1}{2}$, though its form is complicated and does not facilitate analysis of the optimal $\dcr$ in this regime (see Appendix~\ref{appendix:cont_time_analytic_solution_small_delta}). We thus conclude this section by providing a lower bound on $\alpha_\delta^*$, which we prove in Appendix~\ref{appendix:theorem:cont_time_ski_lower_bound}.

\begin{theorem} \label{theorem:cont_time_ski_lower_bound}
    For any $\delta \in [0, 1)$, the optimal $\dcr$ $\alpha_\delta^*$ has the lower bound $$\alpha_\delta^* \geq \max\left\{\frac{e}{e-1}, 2-\frac{1}{2^{\left\lfloor\frac{1}{1-\delta}\right\rfloor - 1}}\right\}.$$
\end{theorem}
We plot this lower bound in Figure~\ref{fig:cont_time_crs} alongside the upper bound from Theorem~\ref{theorem:cvar_skirental_firstalg} and the optimal competitive ratio. While this lower bound is vacuous for $\delta < \frac{2}{3}$, in which case it is exactly the expected cost lower bound of $\frac{e}{e-1}$, it has the same asymptotic form as the upper bound in Theorem~\ref{theorem:cvar_skirental_firstalg} as $\delta$ approaches $1$. Thus, Theorems~\ref{theorem:cvar_skirental_firstalg} and \ref{theorem:cont_time_ski_lower_bound} together give us that $\alpha_\delta^* = 2 - \frac{1}{2^{\Theta\left(\frac{1}{1-\delta}\right)}}$, as $\delta \toup 1$.

\begin{figure}
    \centering
    \includegraphics{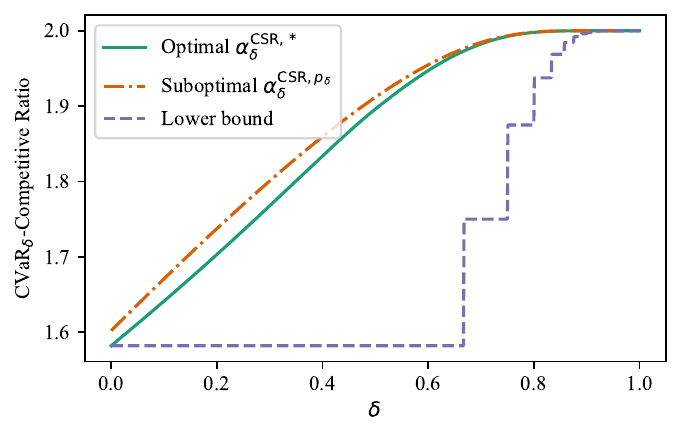}
    \caption{$\dcvar$-competitive ratios from Theorems~\ref{theorem:cvar_skirental_firstalg} (Suboptimal) and \ref{theorem:optimal_cont_time_ski_rental} (Optimal)  and lower bound from Theorem~\ref{theorem:cont_time_ski_lower_bound} for continuous-time ski rental.}
    \label{fig:cont_time_crs}
\end{figure}

\section{\texorpdfstring{$\dcvarbold$}{$\dcvar$}-Competitive Discrete-Time Ski Rental: \\ Phase Transition and Analytic Optimal Algorithm \label{section:discrete_ski_rental}}

Having characterized the optimal algorithm and $\dcr$ for continuous-time ski rental in the previous section, we now turn to the discrete-time version of the problem, and ask: are there any qualitative differences between the optimal algorithm or $\dcr$ for the discrete problem and the continuous-time problem? And are there any regimes of $\delta$ for which we can analytically characterize the optimal algorithm? It turns out that the answer to both of these questions is yes; we begin by showing, in the following theorem, that the optimal $\dcr$ for the discrete-time ski rental problem exhibits a phase transition at $\delta = 1-\Theta(\frac{1}{\log B})$ such that, for $\delta$ beyond this transition, the optimal algorithm is exactly the deterministic algorithm that buys at time $B$.

\begin{theorem} \label{theorem:phase_transition}
    Let $\alpha_\delta^{B, *}$ be the optimal $\dcr$ for discrete-time ski rental with buying cost $B \in \N$. Then $\alpha_\delta^{B, *}$ exhibits a \textbf{phase transition} at $\delta = 1-\Theta(\frac{1}{\log B})$, whereby before this transition, $\alpha_\delta^{B, *}$ strictly improves on the deterministic optimal $\dcr$ of $2 - \frac{1}{B}$, whereas after this transition, $\alpha_\delta^{B, *} = 2-\frac{1}{B}$. Specifically: \vspace{0.5em}
    \begin{enumerate}[itemsep=0.5em,label=(\roman*)]
        \item For all $\delta < 1-\frac{c}{\log(B+1)}$, the optimal $\dcr$ is strictly bounded above by the deterministic optimal $\comprat$: $\alpha_\delta^{B, *} < 2 - \frac{1}{B}$ (where $c = -\frac{1 + 2 W_{-1}\left(\nicefrac{-1}{2\sqrt{e}}\right)}{2} \approx 1.25643$ as in Theorem~\ref{theorem:cvar_skirental_firstalg}).

        \item For all $\delta \geq 1 - \frac{1}{2\lfloor \log_2 B \rfloor + 1}$, the optimal $\dcr$ is exactly the deterministic optimal $\comprat$: ${\alpha_\delta^{B, *} = 2 - \frac{1}{B}}$. Thus, the optimal algorithm for this regime purchases deterministically at time $B$.
    \end{enumerate}
\end{theorem}

We prove this result in Appendix~\ref{appendix:theorem:phase_transition}; the proof of part (i) of follows essentially immediately from our analytic upper bound for the continuous setting (Theorem~\ref{theorem:cvar_skirental_firstalg}), and the proof of part (ii) is adapted from that for the continuous-time lower bound (Theorem~\ref{theorem:cont_time_ski_lower_bound}) in order to handle the discrete nature of the problem. Note that this phase transition behavior is in sharp contrast to the behavior of the optimal $\dcr$ and algorithm in the continuous time setting: whereas in continuous time, $\alpha_\delta^{\mathrm{CSR}, *}$ strictly improves on the deterministic optimal for all $\delta < 1$, in discrete time, $\alpha_\delta^{\mathrm{DSR}(B), *}$ is equal to the deterministic optimal for a non-degenerate interval of $\delta$, implying a limit to the benefit of randomization in the risk-sensitive setting.
In addition, this phase transition result gives an analytic solution for the algorithm with optimal $\dcr$ when $\delta$ is sufficiently large; a natural, complementary question is whether it is possible to obtain an analytic solution for the optimal algorithm with smaller $\delta$. We prove in the next theorem that such a solution can be obtained when $\delta = \calO(\frac{1}{B})$.

\begin{theorem} \label{theorem:discrete_analytic_result}
    Suppose $\delta \leq \left(\frac{B-1}{B}\right)^{B-1}\frac{(1-(1-1/B)^B)^{-1}}{B} = \calO(\frac{1}{B})$. Then the optimal $\dcr$ $\alpha_\delta^{B, *}$ and strategy $\pvec^{B, \delta, *}$ for discrete-time ski rental with buying cost $B$ are
    $$\alpha_\delta^{B, *} = \frac{C-\delta}{1-\delta} \qquad\text{and}\qquad p_i^{B, \delta, *} = \frac{C}{B}\left(1-\frac{1}{B}\right)^{B-i} \quad\text{for all $i \in [B]$},$$
    where $C = \frac{1}{1-(1-1/B)^B}$ is the optimal competitive ratio for the $\delta = 0$ case. In particular, $\pvec^{B, \delta, *}$ is constant as a function of $\delta$, and is identical to the optimal algorithm for the expected cost setting.
\end{theorem}

We prove this result in Appendix~\ref{appendix:theorem:discrete_analytic_result}; the proof follows a similar strategy to the proof characterizing the optimal strategy in the continuous-time setting, and in particular involves the proof of several technical lemmas that, similar to Lemmas~\ref{lemma:s_1_tight} and \ref{lemma:any_s_tight} in the continuous-time setting, characterize the optimal algorithm $\pvec^{B, \delta, *}$ via the adversary's indifference to its chosen ski season duration. As a consequence of this theorem, we can analytically obtain the optimal algorithm for discrete-time ski rental with the $\dcr$ objective whenever $\delta = \calO(\frac{1}{B})$, and the corresponding $\dcr$ is a rational function of $\delta$. We anticipate that extensions of this result may be possible for larger $\delta$, but in general the optimal $\dcr$ will be a piecewise function of $\delta$ whose pieces, including the number of pieces and the intervals they are defined on, will depend on $B$, so we leave the problem of characterizing the $\dcr$ for all $\delta$ and general $B$ to future work. However, if computational results suffice, an adapted form of the binary search approach employed in \cite[Appendix E]{dinitzControllingTailRisk2024} can be used in tandem with a linear programming formulation of the $\cvar$ in order to approximate the optimal solution for any $\delta$ with $\alpha_\delta^{B, *} < 2-\frac{1}{B}$.

\section{\texorpdfstring{$\dcvarbold$}{$\dcvar$}-Competitive One-Max Search: \\Asymptotically Optimal Algorithm and Phase Transition \label{section:one_max_search}}

We now turn our focus to the one-max search problem. As noted in Section~\ref{section:background}, existing results for this problem in the deterministic and randomized settings have established that the optimal deterministic competitive ratio is $\alpha_1^{\theta, *} = \sqrt{\theta}$ and the optimal randomized competitive ratio is $\alpha_0^{\theta, *} = 1 + W_0\left(\frac{\theta - 1}{e}\right) = \Theta(\log \theta)$, where $\theta = \frac{U}{L}$ is the fluctuation ratio. We seek to obtain an upper bound on the $\dcr$ for more general $\delta$; to this end, we prove a lemma that, in an analogous fashion to Lemma~\ref{lemma:cvar_cost_integral_general} for the continuous-time ski rental problem, leverages the integral form of the conditional value-at risk to let us express the $\dcvar$-reward of a particular randomized threshold algorithm $X$ in terms of the inverse CDF of $X$.

\begin{lemma} \label{lemma:one_max_search_cost_integral_representation}
    Let $X$ be a random variable supported in $[L, U]$, and fix an adversary choice of the maximal price $v \in [L, U]$. Then the $\dcvar$ of the profit earned by the algorithm playing the random threshold $X$ is
    \begin{align*}
        \cvar_\delta\left[L \cdot \indic_{X > v} + X \cdot \indic_{X \leq v}\right] &= \begin{cases}L & \text{if $F_X(v) \leq \delta$} \\ \frac{1}{1-\delta}\left[(1 - F_X(v))L + \int_0^{F_X(v) - \delta} F_X^{-1}(t)\,\der t\right] &\text{otherwise.} \end{cases}
    \end{align*}
\end{lemma}
We prove this lemma in Appendix~\ref{appendix:lemma:one_max_search_cost_integral_representation}. While the representation of the $\dcvar$ of profit in this case differs substantially from the cost representation for ski rental in Lemma~\ref{lemma:cvar_cost_integral_general}, it nonetheless also has a relatively simple parametrization in terms of the inverse CDF of the decision $X$, which will facilitate algorithm design. This is due, in part, to the piecewise linear structure exhibited by the cost/profit in these problems, and we anticipate that extending our results to online problems with more general classes of piecewise linear costs and rewards may be a fruitful avenue for future work.

While the representation of the $\dcvar$-reward in Lemma~\ref{lemma:one_max_search_cost_integral_representation} depends on both the CDF and the inverse CDF of $X$, we can eliminate the CDF so long as the maximal price $v \in F_X^{-1}([0, 1])$. Using this fact, we prove the following theorem, proposing an algorithm and establishing an upper bound on the $\dcr$ for all $\delta \in [0, 1]$. We prove the result in Appendix~\ref{appendix:theorem:one_max_search_upper_bound}.

\begin{theorem} \label{theorem:one_max_search_upper_bound}
Let $\delta \in [0, 1]$, and let $\phi : [0, 1] \to [L, U]$ be the solution to the following delay differential equation:
\begin{equation} \label{eq:one_max_search_delay_diffeq}
    \phi'(t) = \frac{\alpha_\delta^{\theta}}{1-\delta}\left[\phi(t-\delta) - L\right] \qquad\text{for $t \in [\delta, 1]$},
\end{equation}
with initial condition $\phi(t) = \alpha_\delta^{\theta} L$ on $t \in [0, \delta]$, where $\alpha_\delta^{\theta}$ is chosen such that $\phi(1) = U$ when $\delta < 1$, and $\alpha_\delta^\theta \coloneqq \sqrt{\theta}$ when $\delta = 1$. Then $\phi$ is the inverse CDF of a random threshold algorithm for one-max search with $\dcr$ $\alpha_\delta^\theta$.
Moreover, $\alpha_\delta^\theta$ is bounded above by the unique positive solution $\overline{r}(\delta)$ to the equation
\begin{equation}\label{theorem:one_max_dcr_ub_root}
    (\overline{r}(\delta) - 1)\left(1 + \frac{\overline{r}(\delta)}{\overline{n}(\delta)}\right)^{\overline{n}(\delta)} = \theta - 1,
\end{equation}
where $\overline{n}(\delta) = \max\left\{1, \left\lfloor \left(\lfloor \delta^{-1}\rfloor - 1\right)/2 \right\rfloor\right\}$, with the $\delta = 0$ case defined by taking $\delta \todown 0$. In particular,
\begin{equation} \label{theorem:one_max_dcr_ub}
    \alpha_\delta^{\theta} \leq \begin{cases} 1 + W_0\left(\frac{\theta - 1}{e}\right) + \calO(\delta) & \text{as $\delta \todown 0$} \\ \sqrt{\theta} &\text{when $\delta > \frac{1}{5}$,} \end{cases}
\end{equation}
with the equality $\alpha_\delta^{\theta} = \sqrt{\theta}$ when $\delta \geq \frac{1}{2}$, where the asymptotic notation omits dependence on $\theta$.
\end{theorem}

We make three brief remarks concerning this result. First, note that the proposed algorithm merely gives an upper bound on the $\dcr$ of one-max search and might not be optimal, although its $\dcr$ matches the optimal randomized and deterministic algorithms in the $\delta = 0$ and $1$ cases. Second, when $\delta \in [0, 1)$, it is possible to analytically solve the delay differential equation \eqref{eq:one_max_search_delay_diffeq} by integrating step-by-step (see Appendix~\ref{appendix:theorem:one_max_search_upper_bound}):
\begin{equation} \label{theorem:one_max_analytic_phi}
    \phi(t) = L + (\alpha_\delta^\theta - 1)L\sum_{j=0}^\infty \frac{(\alpha_\delta^\theta)^j ([t-j\delta]^+)^j}{(1-\delta)^j j!}.
\end{equation}
When $\delta = 0$, \eqref{theorem:one_max_analytic_phi} simplifies to $\phi(t) = L + (\alpha_\delta^\theta - 1)L e^{\alpha_\delta^\theta t}$, the optimal randomized algorithm \citep{sunCompetitiveAlgorithmsOnline2020}. On the other hand, when $\delta \in (0, 1)$, all terms with $j \geq \lceil \delta^{-1}\rceil$ disappear for $t \in [0, 1]$, so $\phi(t)$ is a continuous, piecewise polynomial function. In either case, $\phi(1)$ is strictly increasing in $\alpha_\delta^\theta > 0$, so the $\dcr$ $\alpha_\delta^\theta$ can be obtained numerically by solving $\phi(1) = U$ via standard root-finding methods.

Finally, when $\delta \geq \frac{1}{2}$, Theorem~\ref{theorem:one_max_search_upper_bound} asserts that $\alpha_\delta^\theta = \sqrt{\theta}$, which is identical to the optimal deterministic competitive ratio. This raises the question: can \emph{any} algorithm improve upon the deterministic bound when $\delta \geq \frac{1}{2}$, or is this behavior reflective of a phase transition at $\delta = \frac{1}{2}$ such that randomness cannot improve performance when $\delta$ is greater than this level? In the following result, which we prove in Appendix~\ref{appendix:theorem:oms_lower_bound} by leveraging connections with the $k$-max search problem \citep{lorenzOptimalAlgorithmsKSearch2009}, we provide a lower bound establishing that the latter case is true, and that moreover, the algorithm in Theorem~\ref{theorem:one_max_search_upper_bound} is asymptotically optimal for small $\delta$.

\begin{theorem} \label{theorem:oms_lower_bound}
    Fix $\delta \in [0, 1]$, let $\alpha_\delta^{\theta, *}$ be the optimal $\dcr$ for one-max search, and define $\underline{r}(\delta)$ to be the unique positive solution to the equation
    \begin{equation} \label{eq:one_max_dcr_root_lb}
        (\underline{r}(\delta)-1)\left(1+\frac{\underline{r}(\delta)}{\underline{n}(\delta)}\right)^{\underline{n}(\delta)} = \theta - 1,
    \end{equation}
    where $\underline{n}(\delta) = \max\left\{1, \left\lceil \delta^{-1} \right\rceil - 1\right\}$, with the $\delta = 0$ case defined by taking $\delta \todown 0$. Then $\alpha_\delta^{\theta, *} \geq \underline{r}(\delta)$; in particular,
    \begin{equation} \label{eq:one_max_dcr_lb}
        \alpha_\delta^{\theta, *} \geq \begin{cases} 1 + W_0\left(\frac{\theta - 1}{e}\right) + \Omega(\delta) & \text{as $\delta \todown 0$} \\ \sqrt{\theta} & \text{when $\delta \geq \frac{1}{2}$,} \end{cases}
    \end{equation}
    where the asymptotic notation omits dependence on $\theta$.
\end{theorem}

Thus, in contrast to the continuous-time ski rental problem, which exhibited no phase transition in its competitive ratio, and the discrete-time ski rental problem, which had a phase transition that shrank as $B \to \infty$, Theorem~\ref{theorem:oms_lower_bound} establishes that the one-max search problem has a phase transition at $\delta = \frac{1}{2}$ that remains present even as $\theta \to \infty$. As such, there is a significant limit to the power of randomization in risk-sensitive one-max search. In addition, note that the form of the implicit lower bound \eqref{eq:one_max_dcr_root_lb} matches that of the upper bound \eqref{theorem:one_max_dcr_ub_root}, aside from the definitions of the functions $\underline{n}(\delta)$ and $\overline{n}(\delta)$. This suggests that our upper and lower bounds are tight up to the choice of the function $n(\delta) = \Theta(\delta^{-1})$. In particular, this tightness is made clear in the analytic bounds \eqref{theorem:one_max_dcr_ub} and \eqref{eq:one_max_dcr_lb} in the $\delta \todown 0$ limit, which indicate that our algorithm is asymptotically optimal when $\delta$ is small. We plot the numerically obtained $\dcr$ $\alpha_\delta^\theta$ together with the upper and lower bounds \eqref{theorem:one_max_dcr_ub_root} and \eqref{eq:one_max_dcr_root_lb} in Figure~\ref{fig:cont_time_crs_one_max} when $L = 1$ and $U = 100$, which confirms the near-tightness of the bounds and the phase transition at $\delta = \frac{1}{2}$.

\begin{figure}
    \centering
    \includegraphics{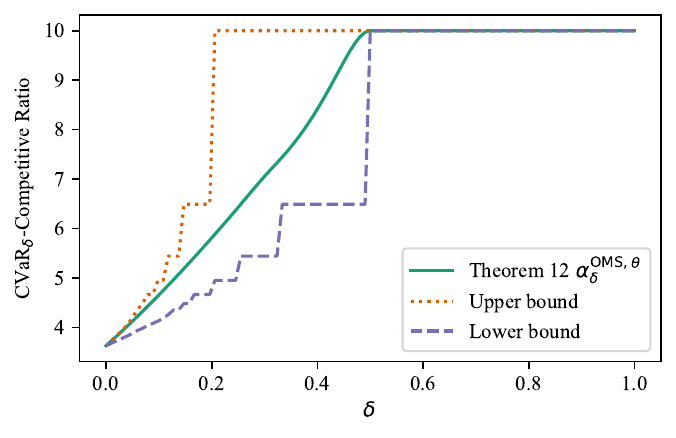}
    \caption{$\dcvar$-competitive ratio of the algorithm in Theorem~\ref{theorem:one_max_search_upper_bound} along with the upper bound \eqref{theorem:one_max_dcr_ub_root} and lower bound \eqref{eq:one_max_dcr_root_lb} for one-max search.}
    \label{fig:cont_time_crs_one_max}
\end{figure}

\section{Conclusion}
In this work, we considered the problem of designing \emph{risk-sensitive} online algorithms, with performance evaluated via a competitive ratio metric -- the $\dcr$ -- defined using the conditional value-at-risk of an algorithm's cost. We considered the continuous- and discrete-time ski rental problems as well as the one-max search problem, obtaining optimal (and suboptimal) algorithms, lower bounds, and analytic characterizations of phase transitions in the optimal $\dcr$ for discrete-time ski rental and one-max search. Our work motivates many interesting new directions, including (a) obtaining an exact or asymptotic analytic form of the optimal $\dcr$ for discrete-time ski rental and one-max search across all $\delta$, (b) the design and analysis of risk-sensitive algorithms for online problems with more general classes of cost and reward functions, or more complex problems such as metrical task systems, (c) exploring the use of alternative risk measures in place of the conditional value-at-risk, and (d) exploring potential connections between risk-sensitive online algorithms and robustness to distribution shift in learning-augmented online algorithms, drawing motivation from the framing of $\dcvar$ in terms of distribution shift.

\section{Acknowledgments}

The authors acknowledge support from an NSF Graduate Research Fellowship (DGE-2139433), NSF Grants CNS-2146814, CPS-2136197, CNS-2106403, and NGSDI-2105648, the Resnick Sustainability Institute, and a Caltech S2I Grant.

\bibliographystyle{plainnat}
\bibliography{main_arxiv}

\newpage

\appendix

\section{Additional Details for Section~\ref{section:background}}
\subsection{Proof of Lemma~\ref{lemma:cont_ski_rental_01_support} \label{appendix:lemma:cont_ski_rental_01_support}}
\begin{proof}[noname]
    Let $X_1 \sim \mu_1$, and define another random variable $X_2 \sim \mu_2$ with support on $[0, 1]$ as
    $$X_2 = \begin{cases}
        X_1 & \text{if $X_1 \leq 1$} \\
        1 & \text{otherwise.}
    \end{cases}$$
    Suppose $s \in [0, 1)$. Clearly $\indic_{X_1 > s} = \indic_{X_2 > s}$ and $\indic_{X_1 \leq s} = \indic_{X_2 \leq s}$, and since $X_2 = X_1$ when $X_1 \leq 1$, $X_1 \cdot \indic_{X_1 \leq s} = X_2 \cdot \indic_{X_2 \leq s}$ for $s < 1$. Thus
    \begin{align*}
        \alpha_\delta^{\mu_1}(s) &= \frac{\cvard[s \cdot \indic_{X_1 > s} + (X_1 + 1) \cdot \indic_{X_1 \leq s}]}{\min\{s, 1\}} \\
        &= \frac{\cvard[s \cdot \indic_{X_2 > s} + (X_2 + 1) \cdot \indic_{X_2 \leq s}]}{\min\{s, 1\}} \\
        &= \alpha_\delta^{\mu_2}(s). \tageq\label{eq:support_01_case1}
    \end{align*}

    Now, consider the case $s \in [1, +\infty]$. By construction, $X_2 \leq 1$, so $\indic_{X_2 > s} = 0$ and $\indic_{X_2 \leq s} = 1$, and thus
    \begin{align*}
        \alpha_\delta^{\mu_2}(s) &= \frac{\cvard[s \cdot \indic_{X_2 > s} + (X_2 + 1) \cdot \indic_{X_2 \leq s}]}{\min\{s, 1\}} \tageq\label{eq:support_01_translation_invariance} \\
        &= 1 + \cvard[X_2] \tageq\label{eq:support_01_case2_1}
    \end{align*}
    where \eqref{eq:support_01_translation_invariance} uses translation invariance of $\cvar$. On the other hand, we have
    \begin{align*}
        \alpha_\delta^{\mu_1}(s) &= \frac{\cvard[s \cdot \indic_{X_1 > s} + (X_1 + 1) \cdot \indic_{X_1 \leq s}]}{\min\{s, 1\}} \\
        &\geq \cvard[1 \cdot \indic_{X_1 > s} + (X_1 + 1) \cdot \indic_{X_1 \leq s}] \tageq\label{eq:support_01_monotonicity_1} \\
        &= 1 + \cvard[X_1 \cdot \indic_{X_1 \leq s}] \tageq\label{eq:support_01_translation_invariance_2} \\
        &\geq 1 + \cvard[X_2] \tageq\label{eq:support_01_monotonicity_2},
    \end{align*}
    where \eqref{eq:support_01_monotonicity_1} follows by monotonicity of $\cvar$, \eqref{eq:support_01_translation_invariance_2} follows by translation invariance, and \eqref{eq:support_01_monotonicity_2} follows by monotonicity when $s = +\infty$ (in which case $\indic_{X_1 \leq s} = 1$). Combining \eqref{eq:support_01_case1}, \eqref{eq:support_01_case2_1}, and \eqref{eq:support_01_monotonicity_2}, we obtain $\alpha_\delta^{\mu_2} \leq \alpha_\delta^{\mu_1}$, as claimed.
\end{proof}

\subsection{On the Restriction to Random Threshold Policies \label{appendix:random_threshold_restriction}}

In this section, we briefly justify the claim that the restriction to random threshold policies for one-max search is made without loss of generality. Let $\alg$ be an arbitrary randomized algorithm for one-max search. Following the argument in the proof of \cite[Theorem 1]{el-yanivOptimalSearchOneWay2001}, the lack of memory restrictions in this problem implies, by Kuhn's Theorem, that $\alg$ is, without loss of generality, a mixed strategy, or a probability distribution over deterministic algorithms \citep{aumannMixedBehaviorStrategies1964}. For some $k \in \N$, let $\epsilon = \frac{U - L}{k}$, and define a restricted set of adversary price sequences $\calI_\epsilon$ as
\begin{align*}
    \calI_\epsilon = \left\{\vvec : \vvec = (L, L+\epsilon, \ldots, L+n\epsilon), n \in \{0, \ldots, k\}\right\},
\end{align*}
i.e., $\calI_\epsilon$ is the set of all price sequences that begin at $L$ and increase by $\epsilon$ at each time. If the adversary is restricted to choosing price sequences in $\calI_\epsilon$, then any deterministic algorithm is equivalent in behavior to some deterministic threshold algorithm. To see why, note that $\calI_\epsilon$ comprises $k+1$ price sequences, each of a unique length in $[k+1]$; we will call $\vvec_n$ the sequence of length $n \in [k+1]$. Moreover, $\vvec_n$ constitutes the first $n$ entries of $\vvec_{n+1}$. As such, the behavior of a deterministic algorithm $\overline{\alg}$ on $\vvec_n$ will be identical to its behavior on the first $n$ prices revealed in $\vvec_{n+1}$; in particular, if $\overline{\alg}$ sells at time $j < n$ in $\vvec_n$, it will do the same in $\vvec_{n+1}$ and earn the same profit $L + (j-1)\epsilon$. As a result, $\overline{\alg}$'s behavior on $\calI_\epsilon$ is wholly determined by the price at which it chooses to sell, which will be consistent across price sequences in this set; in other words, $\overline{\alg}$ is equivalent to a deterministic threshold algorithm, with threshold chosen amongst the $k+1$ choices $\{v : v = L + n\epsilon, n \in \{0, \ldots, k\}\}$.\footnote{If $\overline{\alg}$ never sells on any of the sequences in $\calI_\epsilon$, we can choose a corresponding deterministic threshold of $U$, which obtains performance at least as good.} Thus, on $\calI_\epsilon$, the mixed strategy $\alg$ is equivalent to a distribution over such threshold algorithms, i.e., a random threshold algorithm $X \sim \mu$ with support on $\{v : v = L + n\epsilon, n \in \{0, \ldots, k\}\}$. Formally, we have
\begin{align*}
    \alpha_\delta^{\alg} &= \sup_{\substack{\vvec \in [L, U]^T, \\ T \in \N}} \frac{v_{\max}}{\dcvar[\alg(\vvec)]} \\
    &\geq \max_{\vvec \in \calI_\epsilon} \frac{v_{\max}}{\dcvar[\alg(\vvec)]} \\
    &\geq \max_{\vvec \in \calI_\epsilon} \frac{v_{\max}}{\dcvar[L \cdot \indic_{X > v_{\max}} + X \cdot \indic_{X \leq v_{\max}}]} \\
    &= \max_{\substack{v = L + n\epsilon, \\ n \in \{0, \ldots, k\}}} \frac{v}{\dcvar[L \cdot \indic_{X > v} + X \cdot \indic_{X \leq v}]} \tageq\label{eq:oms_cvar_max_Ieps}
\end{align*}
where $\alg(\vvec)$ denotes the (random) profit of $\alg$ on the price sequence $\vvec$, $v_{\max} \coloneqq \max_j v_j$, and \eqref{eq:oms_cvar_max_Ieps} holds by the construction of $\calI_\epsilon$. Then since $X \sim \mu$ is a random threshold algorithm, its $\dcr$ (with unrestricted adversary) is defined as in \eqref{eq:dcr_oms}:

\begin{align*}
    \alpha_\delta^\mu &= \sup_{v \in [L, U]} \frac{v}{\dcvar[L \cdot \indic_{X > v} + X \cdot \indic_{X \leq v}]} \\
    &\leq \max_{n \in \{0, \ldots, k\}} \sup_{v \in \left[L+n\epsilon,\, L+(n+1)\epsilon\right)} \frac{v}{\dcvar[L \cdot \indic_{X > v} + X \cdot \indic_{X \leq v}]} \\
    &\leq \max_{n \in \{0, \ldots, k\}} \frac{L + (n+1)\epsilon}{\dcvar[L \cdot \indic_{X > L+n\epsilon} + X \cdot \indic_{X \leq L+n\epsilon}]} \tageq\label{ineq:oms_cvar_discrete_support_bound} \\
    &\leq \max_{\substack{v = L + n\epsilon, \\ n \in \{0, \ldots, k\}}} \frac{v}{\dcvar[L \cdot \indic_{X > v} + X \cdot \indic_{X \leq v}]} + \frac{\epsilon}{L} \tageq\label{ineq:oms_cvar_eps_ub} \\
    &\leq \alpha_\delta^{\alg} + \frac{\epsilon}{L} \tageq\label{ineq:oms_cvar_threshold_final_ub}
\end{align*}
where the inequality \eqref{ineq:oms_cvar_discrete_support_bound} holds due to $X$ having support restricted to $\{v : v = L + n\epsilon, n \in \{0, \ldots, k\}\}$, which implies that $\dcvar[L \cdot \indic_{X > v} + X \cdot \indic_{X \leq v}]$ is equal to the $v = L + n\epsilon$ case for all $v \in \left[L+n\epsilon,\, L+(n+1)\epsilon\right)$, \eqref{ineq:oms_cvar_eps_ub} follows by the fact that the algorithm's profit is lower bounded by $L$, and \eqref{ineq:oms_cvar_threshold_final_ub} follows by the inequality in \eqref{eq:oms_cvar_max_Ieps}. Thus, by selecting $k$ arbitrarily large (i.e., $\epsilon$ arbitrarily small), the random threshold algorithm $X$ can be made to have $\dcr$ arbitrarily close to the original randomized algorithm $\alg$.

\section{Proofs and Additional Results for Section~\ref{section:cont_ski_rental}}

\subsection{Proof of Lemma~\ref{lemma:cvar_cost_integral_general} \label{appendix:lemma:cvar_cost_integral_general}}

Before proving the result, we first prove a general lemma that allows for writing an algorithm's $\cvar_\delta$-cost given a particular adversary's decision $s \in [0, 1]$ in terms of the inverse CDF of the algorithm's decision.

\begin{lemma} \label{lemma:cvar_cost_inv_cdf_decision}
    Let $X$ be a random variable supported in $[0, 1]$, and fix an adversary's decision $s \in [0, 1]$. Then the inverse CDF of the ski rental cost given by the random variable $C(X, s) = s \cdot \indic_{X > s} + (X + 1) \cdot \indic_{X \leq s}$ is
    $$F_{C(X, s)}^{-1}(p) = \begin{cases} s & \text{if $p \leq 1 - F_X(s)$} \\ 1 + F_X^{-1}(p + F_X(s) - 1) & \text{otherwise,}\end{cases}$$
    for $p \in [0, 1]$.
\end{lemma}
\begin{proof}
    Observe that the cost $C(X, s)$ takes value $s$ when $X > s$, and is equal to $X + 1$ otherwise, which is always at least $s$; as such, we can easily compute its CDF:
    $$F_{C(X, s)}(x) = \begin{cases}0 & \text{if $x < s$} \\ 1 - F_X(s) & \text{if $x \in [s, 1)$} \\ 1 - F_X(s) + F_X(x - 1) & \text{if $x \in [1, 1 + s]$} \end{cases}$$
    Note that $C(X, s)$ is supported in $[s, 1+s]$; thus, we define its inverse CDF as
    \begin{align*}
        F_{C(X, s)}^{-1}(p) &= \inf\{x \in [s, 1+s] : F_{C(X, s)}(x) \geq p\} \\
        &= \begin{cases} s & \text{if $p = 0$} \\ s & \text{if $p \in (0, 1 - F_X(s)]$} \\ \inf\{x \in [1, 1+s] : 1 - F_X(s) + F_X(x-1) \geq p\} & \text{otherwise} \end{cases} \\
        &= \begin{cases} s & \text{if $p \leq 1 - F_X(s)$} \\ 1 + \inf\{x \in [0, s] : F_X(x) \geq p + F_X(s) - 1\} &\text{otherwise}\end{cases} \\
        &= \begin{cases} s & \text{if $p \leq 1 - F_X(s)$} \\ 1 + F_X^{-1}(p + F_X(s) - 1) &\text{otherwise,}\end{cases}
    \end{align*}
    just as claimed.
\end{proof}

Lemma~\ref{lemma:cvar_cost_integral_general} now follows as a near-immediate consequence of the preceding lemma.

\begin{proof}[Proof of Lemma~\ref{lemma:cvar_cost_integral_general}]
    Define $C(X, s) = s \cdot \indic_{X > s} + (X + 1) \cdot \indic_{X \leq s}$ as the algorithm's cost given a strategy $X$ and adversary's decision $s$, just as in Lemma~\ref{lemma:cvar_cost_inv_cdf_decision}. By the second definition of $\cvar_\delta$ in \eqref{eq:cvar_definition} expressing it as an integral of the inverse CDF, we may write $\cvar_\delta[C(X, s)]$ as:
    \begin{align*}
        \cvar_\delta[C(X, s)] &= \frac{1}{1-\delta}\int_\delta^1 F_{C(X, s)}^{-1}(t)\,\der t
    \end{align*}
    We break into two cases. If $\delta > 1 - F_X(s)$, then by Lemma~\ref{lemma:cvar_cost_inv_cdf_decision}, $F_{C(X, s)}^{-1}(t) = 1 + F_X^{-1}(t + F_X(s) - 1)$ on the entire domain of integration, so we have
    \begin{align*}
        \cvar_\delta[C(X, s)] &= \frac{1}{1-\delta}\int_\delta^1 1 + F_X^{-1}(t + F_X(s) - 1)\,\der t \\
        &= \frac{1}{1-\delta}\int_{F_X(s) - (1 - \delta)}^{F_X(s)} 1 + F_X^{-1}(t )\,\der t.
    \end{align*}
    On the other hand, if $\delta \leq 1 - F_X(s)$, then by Lemma~\ref{lemma:cvar_cost_inv_cdf_decision}, $F_{C(X, s)}^{-1}(t) = s$ on $[\delta, 1 - F_X(s)]$ and $F_{C(X, s)}^{-1}(t) = 1 + F_X^{-1}(t + F_X(s) - 1)$ on $[1-F_X(s), 1]$. Thus,
    \begin{align*}
        \cvar_\delta[C(X, s)] &= \frac{1}{1-\delta}\left(\int_\delta^{1 - F_X(s)}s\,\der t + \int_{1-F_X(s)}^{1} 1 + F_X^{-1}(t + F_X(s) - 1)\,\der t\right) \\
        &= \frac{1}{1-\delta}\left[(1 - \delta - F_X(s))s + \int_0^{F_X(s)}1 + F_X^{-1}(t)\,\der t\right].
    \end{align*}
\end{proof}

\subsection{Proof of Theorem~\ref{theorem:cvar_skirental_firstalg} \label{appendix:theorem:cvar_skirental_firstalg}}
\begin{proof}[noname]
    First, note that the CDF of the strategy $X$ on $[0, 1]$ is
    $$F_X(x) = \int_0^x p_\delta(y) \,\der y = -\frac{1-\delta}{c}\log\left[1 + \left(e^{-\frac{c}{1-\delta}} - 1\right)x\right],$$
    which is strictly increasing (and hence one-to-one) on $[0, 1]$, with $F_X(0) = 0$ and $F_X(1) = 1$. The corresponding inverse CDF is
    $$F_X^{-1}(y) = \frac{1 - e^{-\frac{cy}{1-\delta}}}{1 - e^{-\frac{c}{1-\delta}}},$$
    for $y \in [0, 1]$. This is strictly increasing in $y$, so $F_X^{-1}$ is one-to-one, and any adversary decision $s \in [0, 1]$ corresponds to some $y \in [0, 1]$ such that $s = F_X^{-1}(y)$.

    Now suppose the adversary's decision is $s = F_X^{-1}(y)$ for $y \leq 1-\delta$. Then by Lemma~\ref{lemma:cvar_cost_integral_general}, the the $\dcr$ of the algorithm's cost in this case is
    \begin{align*}
        \frac{\cvard[s \cdot \indic_{X > s} + (X + 1) \cdot \indic_{X \leq s}]}{s}
    &= \frac{1}{F_X^{-1}(y)}\dfrac{1}{1-\delta}\left[(1-\delta - y)F_X^{-1}(y) + \int_0^y 1 + F_X^{-1}(t)\,\der t\right] \\
    &= \dfrac{1}{1-\delta}\left[(1-\delta - y) + \frac{1 - e^{-\frac{c}{1-\delta}}}{1 - e^{-\frac{cy}{1-\delta}}}\left(y + \int_0^y \frac{1 - e^{-\frac{ct}{1-\delta}}}{1 - e^{-\frac{c}{1-\delta}}}\,\der t\right)\right] \\
    &= \dfrac{1}{1-\delta}\left[(1-\delta - y) + \frac{1 - e^{-\frac{c}{1-\delta}}}{1 - e^{-\frac{cy}{1-\delta}}}\left(y + \frac{1-cy - e^{-\frac{cy}{1-\delta}}(1-\delta) - \delta}{c(e^{-\frac{c}{1-\delta}}-1)}\right)\right] \\
    &= \dfrac{1}{1-\delta}\left[(1-\delta - y) + \frac{2 - e^{-\frac{c}{1-\delta}}}{1 - e^{-\frac{cy}{1-\delta}}}y - \frac{1-\delta}{c}\right] \\
    &= 1 + \frac{y}{1-\delta}\left(\frac{2 - e^{-\frac{c}{1-\delta}}}{1 - e^{-\frac{cy}{1-\delta}}} - 1\right) - \frac{1}{c} \tageq\label{eq:dcr_y_leq_1_min_del}\\
    &\leq \frac{2 - e^{-\frac{c}{1-\delta}}}{1 - e^{-c}} - \frac{1}{c} \tageq\label{eq:dcr_y_leq_1_min_del_final_ineq}
    \end{align*}
    where the final inequality \eqref{eq:dcr_y_leq_1_min_del_final_ineq} follows from the straightforward observation that \eqref{eq:dcr_y_leq_1_min_del} is increasing in $y$, so is maximized in this case at $y = 1-\delta$ (recall we have assumed $y \leq 1-\delta$).

    Now, consider the alternative case that $y > 1-\delta$. By Lemma~\ref{lemma:cvar_cost_integral_general}, the the $\dcr$ of the algorithm's cost in this case is
    \begin{align*}
        \frac{\cvard[s \cdot \indic_{X > s} + (X + 1) \cdot \indic_{X \leq s}]}{s}
    &= \frac{1}{F_X^{-1}(y)}\dfrac{1}{1-\delta}\int_{y - (1-\delta)}^y 1 + F_X^{-1}(t)\,\der t \\
    &= \frac{1 - e^{-\frac{c}{1-\delta}}}{1 - e^{-\frac{cy}{1-\delta}}}\left(1 + \frac{e^{-\frac{cy}{1-\delta}}(1 - e^c) + c}{c(1 - e^{-\frac{c}{1-\delta}})}\right) \\
    &= \frac{2 - e^{-\frac{c}{1-\delta}}}{1 - e^{-\frac{cy}{1-\delta}}} + \frac{e^{-\frac{cy}{1-\delta}}(1 - e^c)}{c(1 - e^{-\frac{cy}{1-\delta}})} \tageq\label{eq:dcr_y_geq_1_min_del} \\
    &\leq \frac{2 - e^{-\frac{c}{1-\delta}}}{1 - e^{-\frac{c}{1-\delta}}} + \frac{e^{-\frac{c}{1-\delta}}(1 - e^c)}{c(1 - e^{-\frac{c}{1-\delta}})} \tageq\label{eq:dcr_y_geq_1_min_del_final_ineq}
    \end{align*}
    where the final inequality \eqref{eq:dcr_y_geq_1_min_del_final_ineq} follows from the fact that \eqref{eq:dcr_y_geq_1_min_del} is increasing in $y$, and thus is maximized for $y = 1$ (recall that $y \in (1-\delta, 1]$ in this case). To see that this is the case, observe that
    \begin{align*}
        \frac{\der}{\der y} \left(\frac{2 - e^{-\frac{c}{1-\delta}}}{1 - e^{-\frac{cy}{1-\delta}}} + \frac{e^{-\frac{cy}{1-\delta}}(1 - e^c)}{c(1 - e^{-\frac{cy}{1-\delta}})}\right) &= -\frac{e^{\frac{c(y-1)}{1-\delta}}(-c + e^{\frac{c}{1-\delta}}(1+2c-e^c))}{(1-e^{\frac{cy}{1-\delta}})^2(1-\delta)} \\
        &= \frac{e^{\frac{c(y-1)}{1-\delta}}c}{(1-e^{\frac{cy}{1-\delta}})^2(1-\delta)} \tageq\label{eq:dcr_y_geq_1_min_del_derivative_lambertw}\\
        &> 0 \quad\text{for all $y \in (1-\delta, 1]$,}
    \end{align*}
    where \eqref{eq:dcr_y_geq_1_min_del_derivative_lambertw} follows by the assumption in the theorem statement that $c = -\frac{1 + 2 W_{-1}\left(\nicefrac{-1}{2\sqrt{e}}\right)}{2}$, since if we substitute this definition of $c$ into $1 + 2c-e^c$, we obtain
    \begin{align*}
        1 + 2c-e^c &= -2 W_{-1}\left(\nicefrac{-1}{2\sqrt{e}}\right) - e^{-\frac{1 + 2 W_{-1}\left(\nicefrac{-1}{2\sqrt{e}}\right)}{2}} \\
        &= -e^{-\frac{1}{2}}\left(2\sqrt{e} \cdot W_{-1}\left(\nicefrac{-1}{2\sqrt{e}}\right) + e^{-W_{-1}\left(\nicefrac{-1}{2\sqrt{e}}\right)}\right) \\
        &= 0, \tageq\label{eq:applying_lambert_eq}
    \end{align*}
    since the Lambert $W$ function is defined to satisfy $W_k(z) \cdot e^{W_k(z)} = z$.

    Combining the two cases \eqref{eq:dcr_y_leq_1_min_del_final_ineq} and \eqref{eq:dcr_y_geq_1_min_del_final_ineq}, we have that the $\dcr$ of the algorithm that buys on a random day with density $p_\delta$ is
    \begin{equation} \label{eq:ansatz_cr_double_bound}
        \alpha_\delta^{p_\delta} = \max\left\{\frac{2 - e^{-\frac{c}{1-\delta}}}{1 - e^{-c}} - \frac{1}{c}, \frac{2 - e^{-\frac{c}{1-\delta}}}{1 - e^{-\frac{c}{1-\delta}}} + \frac{e^{-\frac{c}{1-\delta}}(1 - e^c)}{c(1 - e^{-\frac{c}{1-\delta}})}\right\}.
    \end{equation}
    We will now show that for our chosen constant $c$, the latter entry in the maximum is larger for all $\delta \in [0, 1)$. Define a function $f$ as the difference of \eqref{eq:dcr_y_geq_1_min_del_final_ineq} and \eqref{eq:dcr_y_leq_1_min_del_final_ineq}:
    $$f(\delta; c) = \frac{2 - e^{-\frac{c}{1-\delta}}}{1 - e^{-\frac{c}{1-\delta}}} + \frac{e^{-\frac{c}{1-\delta}}(1 - e^c)}{c(1 - e^{-\frac{c}{1-\delta}})}\ - \left(\frac{2 - e^{-\frac{c}{1-\delta}}}{1 - e^{-c}} - \frac{1}{c}\right).$$
    Our goal is to show that $f(\delta; c) \geq 0$ for all $\delta \in [0, 1)$. First, observe that $f(0) = 0$. Moreover, since $\lim_{\delta \toup 1} e^{-\frac{c}{1-\delta}} = 0$, we have
    \begin{align*}
        \lim_{\delta \toup 1} f(\delta;c) &= 2 - \frac{2}{1-e^{-c}} + \frac{1}{c} \\
        &= -\frac{2e^{-c}}{1-e^{-c}} + \frac{1}{c} \\
        &= -\frac{2}{e^c - 1} + \frac{1}{c} \\
        &= 2\left(\frac{1}{1 - e^c} + \frac{1}{2c}\right) \\
        &= 0,
    \end{align*}
    where the final equality follows from rearranging the equality $1 + 2c - e^c = 0$ shown in \eqref{eq:applying_lambert_eq}, which follows from our choice of $c$. Thus the function $f(\delta; c)$ is zero at the endpoints of the interval $[0, 1)$. Since $f'$ is continuously differentiable, if we can show that $f'(0; c) > 0$ and that $f'(\delta; c) = 0$ exactly once on the interval $[0, 1)$, these together will imply the desired property that $f(\delta; c) \geq 0$ for all $\delta \in [0, 1)$.\footnote{To see that this is the case, suppose instead that $f(\delta'; c) < 0$ for some $\delta' \in (0, 1)$, and note that strict positivity of the initial derivative $f'(0; c) > 0$, continuity of $f'$, and the limit $\lim_{\delta \toup 1} f(\delta;c) = 0$ imply that $f'(\delta; c)$ must be zero at least twice on the interval, contradicting the supposition that $f'(\delta; c) = 0$ exactly once.}

    Computing the derivative of $f$ at $\delta = 0$, we find
    \begin{align*}
        f'(\delta; c)\Big|_{\delta = 0} &= \frac{e^{-\frac{c}{1-\delta}}\left(\left(1 - e^c\right)^2 - c\left(-1 + 2e^c - 2e^{-\frac{c\delta}{1-\delta}} + e^{-\frac{c(1+\delta)}{1-\delta}}\right)\right)}{(e^c - 1)(1 - e^{-\frac{c}{1-\delta}})^2(1-\delta)^2}\Bigg|_{\delta = 0} \tageq\label{eq:f_prime_defn}\\
        &= \frac{c + e^c(e^c - 1 - 2c)}{(e^c - 1)^2} \\
        &= \frac{c}{(e^c - 1)^2} \tageq\label{eq:f_prime_lambert}\\
        &> 0,
    \end{align*}
    where \eqref{eq:f_prime_lambert} follows from \eqref{eq:applying_lambert_eq}. Moreover, inspecting the form of \eqref{eq:f_prime_defn}, it is clear that $e^{-\frac{c}{1-\delta}}$ and the denominator $(e^c - 1)(1 - e^{-\frac{c}{1-\delta}})^2(1-\delta)^2$ are both strictly positive (recall, in particular, that $c > 1$). As such, the sign of $f'(\delta; c)$ is exactly the sign of $\left(1 - e^c\right)^2 - c\left(-1 + 2e^c - 2e^{-\frac{c\delta}{1-\delta}} + e^{-\frac{c(1+\delta)}{1-\delta}}\right)$, so to determine the zeros of $f'(\delta; c)$, we may instead determine the zeros of the expression
    $$\left(1 - e^c\right)^2 - c\left(-1 + 2e^c - 2e^{-\frac{c\delta}{1-\delta}} + e^{-\frac{c(1+\delta)}{1-\delta}}\right).$$ To this end, we compute another derivative:
    \begin{align*}
        \frac{\der}{\der \delta}\left[\left(1 - e^c\right)^2 - c\left(-1 + 2e^c - 2e^{-\frac{c\delta}{1-\delta}} + e^{-\frac{c(1+\delta)}{1-\delta}}\right)\right] &= \frac{2c^2e^{-\frac{c\delta}{1-\delta}}(e^{-\frac{c}{1-\delta}}-1)}{(1-\delta)^2} \\
        &< 0
    \end{align*}
    for all $\delta \in [0, 1)$, since $e^{-\frac{c}{1-\delta}} < 1$. As
    \begin{align*}
        \left(1 - e^c\right)^2 - c\left(-1 + 2e^c - 2e^{-\frac{c\delta}{1-\delta}} + e^{-\frac{c(1+\delta)}{1-\delta}}\right)\Bigg|_{\delta = 0} &= 4c^2 - c(-3 + 2e^c + e^{-c}) \\
        &= 4c^2 - c(-1 + 4c + e^{-c})\\
        &= -\frac{\left(1 + 2\cdot W_{-1}\left(\nicefrac{-1}{2\sqrt{e}}\right)\right)^{2}}{4\cdot W_{-1}\left(\nicefrac{-1}{2\sqrt{e}}\right)} \approx 0.899 > 0
    \end{align*}
    and
    \begin{align*}
        \lim_{\delta \toup 1} \left(1 - e^c\right)^2 - c\left(-1 + 2e^c - 2e^{-\frac{c\delta}{1-\delta}} + e^{-\frac{c(1+\delta)}{1-\delta}}\right) &= c - 2ce^c + (e^c - 1)^2 \\
        &= \frac{1}{2} + W_{-1}\left(\nicefrac{-1}{2\sqrt{e}}\right) \approx -1.256 < 0,
    \end{align*}
    it follows that $\left(1 - e^c\right)^2 - c\left(-1 + 2e^c - 2e^{-\frac{c\delta}{1-\delta}} + e^{-\frac{c(1+\delta)}{1-\delta}}\right)$, and thus $f'(\delta; c)$, has exactly one zero on $[0, 1)$. As argued previously, this implies that $f(\delta; c) \geq 0$ for all $\delta \in [0, 1)$, and hence the second entry on the right-hand side of \eqref{eq:ansatz_cr_double_bound} is always larger:
    $$\alpha_\delta^{p_\delta} = \frac{2 - e^{-\frac{c}{1-\delta}}}{1 - e^{-\frac{c}{1-\delta}}} + \frac{e^{-\frac{c}{1-\delta}}(1 - e^c)}{c(1 - e^{-\frac{c}{1-\delta}})}.$$

    Simplifying this formula via \eqref{eq:applying_lambert_eq}, we have
    $$\alpha_\delta^{p_\delta} = 2 - \frac{1}{e^{\frac{c}{1-\delta}}-1},$$
    from which it is readily observed that $\alpha_\delta^{p_\delta} < 2$ for all $\delta \in [0, 1)$; moreover, $\lim_{\delta \toup 1} 2 - \frac{1}{e^{\frac{c}{1-\delta}}-1} = 2$, so the above expression for $\alpha_\delta^{p_\delta}$ is valid, and indeed optimal, in the case of $\delta = 1$ (in this case, we interpret the algorithm as placing full probability mass on purchasing at time 1). On the other hand, $\alpha_0^{p_0} = 2 - \frac{1}{e^c - 1} \approx 1.60 > \alpha_0^* \approx 1.58$, so this algorithm is not optimal for all $\delta$, though it provides a very close approximation of the optimal competitive ratio in the case of $\delta = 0$.

\end{proof}

\subsection{Proof of Theorem~\ref{theorem:optimal_cont_time_ski_rental} \label{appendix:theorem:optimal_cont_time_ski_rental}}

It is known that the optimal ski-rental algorithm $\mu^*$ is \emph{indifferent} to the adversary's decision $s \in (0, 1]$ when $\delta = 0$ (in the expected cost case), i.e., $\alpha_0^{\mu^*}(s) = \frac{e}{e-1}$ for all $s \in (0, 1]$ \cite{karlinCompetitiveRandomizedAlgorithms1994}. A similar tightness property was proved in \cite{dinitzControllingTailRisk2024} in the setting of discrete-time ski rental with $\var$ constraints. In the following, we show that this tightness property also holds for any $\delta \in (0, 1)$ for continuous-time ski rental: if $\mu_\delta^*$ is optimal for the $\dcr$, then $\alpha_\delta^{\mu^*}(s) = \alpha_\delta^{\mu^*}$ for all $s \in (0, 1]$. Following the high-level strategy of \cite{dinitzControllingTailRisk2024}, we prove this result in two steps: first, we prove that $\alpha_\delta^{\mu^*}(1) = \alpha_\delta^{\mu^*}$. Then, we prove that for any algorithm $\mu$, if $\alpha_\delta^{\mu}(s) < \alpha_\delta^{\mu}$ for some $s \in (0, 1]$, we can construct an algorithm $\hat{\mu}$ with a competitive ratio that is no worse than $\mu$, yet which has $\alpha_\delta^{\hat{\mu}}(1) < \alpha_\delta^{\hat{\mu}}$, thus implying $\mu$ is not optimal. We begin with a lemma establishing that any optimal algorithm cannot have a probability mass more than $(1-\delta)(\alpha_\delta^* - 1)$ on any single point.

\begin{lemma} \label{lemma:no_large_point_masses}
    Let $\delta \in [0, 1)$, and let $\mu^*$ be an algorithm with optimal $\dcr$ for continuous-time ski rental. Then $\mu^*$ cannot assign any point a probability mass greater than $(1-\delta)(\alpha_\delta^* - 1)$.
\end{lemma}
\begin{proof}
    By Lemma~\ref{lemma:cont_ski_rental_01_support} we can assume that $\mu^*$ has support in $[0, 1]$; now suppose for the sake of contradiction that $\mu^*(x) > (1-\delta)(\alpha_\delta^* - 1)$ for some $x \in [0, 1]$. We can easily construct a lower bound on the $\dcr$ as follows:
    \begin{align*}
        \alpha_\delta^{\mu^*}(x) &> \frac{(1-\delta)(\alpha_\delta^* - 1)(1+x) + (1-\delta-(1-\delta)(\alpha_\delta^* - 1))x}{(1-\delta)x} \\
        &= \frac{(1-\delta)(\alpha_\delta^* - 1) + (1-\delta)x}{(1-\delta)x} \\
        &= \frac{\alpha_\delta^* - 1}{x} + 1 \\
        &\geq \alpha_\delta^*
    \end{align*}
    where the final inequality follows since $\frac{\alpha_\delta^* - 1}{x}$ is decreasing in $x$ and $\frac{\alpha_\delta^* - 1}{x} + 1\Big|_{x = 1} = \alpha_\delta^*$. Since $\mu^*$ was assumed optimal, this strict inequality clearly yields a contradiction.

    We conclude by briefly noting that $(1-\delta)(\alpha_\delta^* - 1) < 1-\delta$, since by Theorem~\ref{theorem:cvar_skirental_firstalg}, $\alpha_\delta^* < 2$ for any $\delta \in [0, 1)$.
\end{proof}

Now, we prove that the competitive ratio must be tight when $s = 1$.

\begin{lemma} \label{lemma:s_1_tight}
    Let $\delta \in [0, 1)$, and let $\mu$ be an algorithm with optimal $\dcr$ for continuous-time ski rental, so $\alpha_\delta^{\mu} = \alpha_\delta^*$. Then $\alpha_\delta^{\mu}(1) = \alpha_\delta^{\mu}$.
\end{lemma}
\begin{proof}
    Suppose otherwise, and let $\alpha_\delta^{\mu}(1) = \alpha_\delta^{\mu} - \epsilon$ for some $\epsilon > 0$. Note that by Lemma~\ref{lemma:no_large_point_masses}, $\mu$ cannot place a probability mass greater than $(1-\delta)(\alpha_\delta^* - 1)$ on any single point, which is strictly less than $(1-\delta)$ by the fact that $\alpha_\delta^* < 2$ for $\delta \in [0, 1)$ (Theorem~\ref{theorem:cvar_skirental_firstalg}). In particular, this implies that $\mu$ does not assign all its probability to the decision $x = 1$, and decreasing the probability mass assigned to a particular decision $x$ and moving it to an earlier decision will strictly decrease the $\dcr$ at that decision.

    Now, consider another algorithm $\hat{\mu}$ with measure defined as:
    $$\hat{\mu} = (1-\gamma)\mu + \gamma\delta_1,$$
    with $\gamma > 0$ a small constant and where $\delta_1$ is a unit point mass at $x = 1$. For any $s \in (0, 1)$, one of the following two cases must hold:
    \begin{enumerate}[(a)]
        \item Suppose $\alpha_\delta^{\mu}(s) = s^{-1} \cdot \cvar_\delta[s \cdot \indic_{X > s} + (X + 1) \cdot \indic_{X \leq s}] = 1$. This means that the worst $(1-\delta)$-sized subpopulation of the loss distribution (i.e., the distribution of the random variable $\indic_{X > s} + s^{-1}(X + 1) \cdot \indic_{X \leq s}$) is contained in the event $X > s$, so it must be that $X$ takes values at most $s$ with probability zero, i.e., $\mu[0, s] = 0$. Likewise, we must have $\hat{\mu}[0, s] = 0$, so $\alpha_\delta^{\hat{\mu}}(s) = 1$ as well.

        \item Alternatively, let $\alpha_\delta^{\mu}(s) > 1$. This means that $\mu[0, s] = c > 0$. We break into two subcases:

        \noindent(i) If $c < 1-\delta$, then the worst $(1-\delta)$-sized subpopulation of the loss distribution must yield a loss of $1$ with probability $1-\delta - c > 0$; hence we may write
        \begin{align*}
            \alpha_\delta^{\mu}(s) &= \E[s^{-1}(X + 1) \cdot \indic_{X \leq s}] + 1-\delta - c \\
            &= c \cdot \E[s^{-1}(X + 1) \cdot \indic_{X \leq s}|X \leq S] + 1-\delta - c.
        \end{align*}
        Let $B \sim \mathrm{Bernoulli}(\gamma)$ be independent of $X$, and define $\hat{X}$ as a random variable that is equal to $X$ when $B = 0$ and is 1 otherwise; clearly $\hat{X}$ has distribution $\hat{\mu}$. Then since $\hat{\mu}[0, s] = (1-\gamma)\mu[0, s] = (1-\gamma)c$, we similarly obtain
        \begin{align*}
            \alpha_\delta^{\hat{\mu}}(s) &= \E[s^{-1}(\hat{X} + 1) \cdot \indic_{\hat{X} \leq s}] + 1-\delta - (1-\gamma)c \\
            &= \E[s^{-1}(X + 1)|\hat{X} \leq s]\Prob(\hat{X} \leq s) + 1-\delta - (1-\gamma)c \tageq\label{eq:cvar_lemma_X}\\
            &= \E[s^{-1}(X + 1)|X \leq s]\Prob(\hat{X} \leq s) + 1-\delta - (1-\gamma)c \tageq\label{eq:cvar_lemma_X_2}\\
            &= (1-\gamma)c \cdot \E[s^{-1}(X + 1)|X \leq s] + 1-\delta - (1-\gamma)c \\
            &< \alpha_\delta^{\mu}(s), \tageq\label{eq:cvar_lemma_strict_ineq}
        \end{align*}
        where \eqref{eq:cvar_lemma_X} holds since, by construction, $\hat{X} \leq s < 1$ implies $\hat{X} = X$; \eqref{eq:cvar_lemma_X_2} follows from the fact that the event $\hat{X} \leq s$ is exactly the joint event $\{X \leq S, B = 0\}$ and $X$ is independent of $B$; and \eqref{eq:cvar_lemma_strict_ineq} is a consequence of $\gamma > 0$ and $\E[s^{-1}(X + 1)|X \leq s] > 1$ for $s < 1$.

        \noindent(ii) If $c \geq 1-\delta$, then the worst $(1-\delta)$-sized subpopulation of the loss distribution is wholly induced by outcomes of $X$ lying in the interval $[0, s]$; calling $\nu$ this subpopulation distribution of $X$, we have $\alpha_\delta^{\mu}(s) = \E_{X \sim \nu}[s^{-1}(X + 1)]$. This subpopulation shrinks to size $(1-\delta)(1-\gamma)$ in the construction of $\hat{\mu}$, so in order to construct the worst-case $(1-\delta)$-sized loss subpopulation of $\hat{\mu}$, we must augment $\nu$ with an additional loss subpopulation (call it $\hat{\nu}$) with size $(1-\delta)\gamma$. It must hold that some nontrivial portion of the losses included in $\hat{\nu}$ are strictly less than $s \cdot \alpha_\delta^{\mu}(s)$, for if this were not the case, it would imply that $\mu$ contains a probability atom of size $\frac{1-\delta}{1-\gamma}$, violating optimality. Thus, we may calculate:
        \begin{align*}
            \alpha_\delta^{\hat{\mu}}(s) &= (1-\gamma) \cdot \E_{X \sim \nu}[s^{-1}(X + 1)] + \gamma \cdot \E_{X \sim \hat{\nu}}[\indic_{X > s} + s^{-1}(X + 1) \cdot \indic_{X \leq s}] \\
            &< (1-\gamma) \cdot \E_{X \sim \nu}[s^{-1}(X + 1)] + \gamma \cdot \alpha_\delta^{\mu}(s) \tageq\label{ineq:cvar_lemma_ii}\\
            &= \alpha_\delta^{\mu}(s),
        \end{align*}
        where \eqref{ineq:cvar_lemma_ii} follows from a nontrivial portion of the losses in $\hat{\nu}$ being strictly less than $\alpha_\delta^{\mu}(s)$.
    \end{enumerate}

    Finally, consider the case of $s = 1$. In this case, the loss is always $X + 1$, since $X \leq 1$ without loss of generality (Lemma~\ref{lemma:cont_ski_rental_01_support}); since this is strictly increasing in the outcome of $X$, the worst $(1-\delta)$-sized subpopulation of the loss distribution is exactly the $(1-\delta)$ tail of $X$, which we call $\nu$. This tail shrinks to size $(1-\delta)(1-\gamma)$ in the construction of $\hat{\mu}$, but an additional probability mass of weight $\gamma$ is added to the outcome $X = 1$, and as this outcome maximizes the loss, the worst $(1-\delta)$-sized subpopulation of the loss under $\hat{\mu}$ is easily seen to be $(1-\gamma)\nu + \gamma\delta_1$. Thus, if we choose $\gamma = \frac{\epsilon}{2}$,
    \begin{align*}
        \alpha_\delta^{\hat{\mu}}(1) &= (1-\gamma)\cdot \E_{X \sim \nu}[X + 1] + \gamma \cdot \E_{X \sim \delta_1}[X + 1] \\
        &= (1-\gamma)\alpha_\delta^{\mu}(1) + 2\gamma \\
        &= (1-\gamma)(\alpha_\delta^{\mu} - \epsilon) + 2\gamma \\
        &< (\alpha_\delta^{\mu} - \epsilon) + 2\gamma \\
        &= \alpha_\delta^{\mu}.
    \end{align*}

    It follows from the above cases that $\alpha_\delta^{\hat{\mu}} = \sup_{s \in (0, 1]} \alpha_\delta^{\hat{\mu}} < \alpha_\delta^{\mu}$: for $s \in (0, 1)$ satisfying case (a), $\alpha_\delta^{\hat{\mu}}(s) = 1 < \alpha_\delta^{\mu}$, and in case (b) and the case of $s = 1$, we have shown $\alpha_\delta^{\hat{\mu}}(s) < \alpha_\delta^{\mu}$. However, this implies that $\hat{\mu}$ has strictly better $\dcr$ than $\mu$, contradicting the optimality of $\mu$. As a result, we must have $\alpha_\delta^{\mu}(1) = \alpha_\delta^{\mu}$.
\end{proof}

We will now begin to prove the second structural result: that if $\alpha_\delta^{\mu}(s) < \alpha_\delta^\mu$ for some $s \in (0, 1)$, we can construct an algorithm $\hat{\mu}$ with $\dcr$ that is no worse than $\mu$ and for which $\alpha_\delta^{\hat{\mu}}(1) < \alpha_\delta^{\hat{\mu}}$, which by the previous lemma implies that $\mu$ is not optimal. We will first prove a series of technical lemmas that support this proof: the first tells us that, if there is ``slack'' in the $\dcr$ for a given adversary decision $x \in (0, 1]$, then this implies that there is slack of a comparable magnitude in a small interval $[x, x+\epsilon]$.

\begin{lemma} \label{lemma:dcr_slack_interval}
    Let $\delta \in [0, 1)$, and let $\mu$ be an algorithm for continuous-time ski rental with support in $[0, 1]$. Suppose there exists an $x \in (0, 1)$ for which $\alpha_\delta^\mu(x) < \alpha_\delta^\mu$, and define $2\gamma \coloneqq \alpha_\delta^\mu - \alpha_\delta^\mu(x) > 0$. Then there exists some $\epsilon > 0$ such that for any $y \in [x, x+\epsilon]$, $\alpha_\delta^\mu(y) \leq \alpha_\delta^\mu - \gamma$.
\end{lemma}
\begin{proof}
    When $\delta = 0$, $\cvar_\delta$ is exactly the expectation, so we have:
    \begin{align*}
        \alpha_0^\mu(x) &= \E[\indic_{X > x} + x^{-1}(X + 1)\cdot\indic_{X \leq x}] \\
        &= 1 - F_X(x) + x^{-1}\left(F_X(x) + \E[X\cdot\indic_{X \leq x}]\right) \\
        &= 1 - F_X(x) + x^{-1}\left(F_X(x) + \int_0^x 1 - F_{X\cdot\indic_{X \leq x}}(t)\,\der t\right) \\
        &= 1 - F_X(x) + x^{-1}\left(F_X(x) + \int_0^x F_X(x) - F_X(t)\,\der t\right), \tageq\label{eq:expected_cr_expression_cdf}
    \end{align*}
    which is easily seen to be right-continuous at $x$, by right-continuity of the CDF. Thus there must exist some $\epsilon$ ensuring $|\alpha_0^\mu(x) - \alpha_0^\mu(y)| \leq \gamma$ for all $y \in [x, x+\epsilon]$, which implies the desired property.

    On the other hand, suppose $\delta \in (0, 1)$. We can choose $\epsilon > 0$ sufficiently small such that $\rho \coloneqq F_X(x+\epsilon) - F_X(x) = \Prob(X \in (x, x+\epsilon]) \leq \delta$ (this is always possible due to right-continuity of $F_X$). Defining the algorithm's cost given an adversary decision $s$ as $C(X, s) = s\cdot\indic_{X > s} + (X + 1)\cdot\indic_{X \leq s}$, we have
    \begin{align*}
        \alpha_\delta^\mu(x+\epsilon) &= \frac{1}{(1-\delta)(x+\epsilon)}\int_\delta^1 F_{C(X, x+\epsilon)}^{-1}(t)\,\der t \\
        &\leq \frac{1}{(1-\delta)(x+\epsilon)}\left[\int_{\delta-\rho}^{1-\rho} F_{C(X, x+\epsilon)}^{-1}(t)\,\der t + \int_{1-\rho}^1 F_{C(X, x+\epsilon)}^{-1}(t)\,\der t\right]. \tageq\label{ineq:intermediate_two_integrals}
    \end{align*}
    Now, applying Lemma~\ref{lemma:cvar_cost_inv_cdf_decision}, we can bound each of the two integrals in \eqref{ineq:intermediate_two_integrals}. For the first integral, there are two cases according to the two cases in Lemma~\ref{lemma:cvar_cost_inv_cdf_decision}: if $\delta - \rho > 1 - F_X(x+\epsilon)$, then $F_{C(X, x+\epsilon)}^{-1}(t) = 1 + F_X^{-1}(t + F_X(x+\epsilon)-1)$ on the domain of integration, so we have
    \begin{align*}
        \int_{\delta-\rho}^{1-\rho} F_{C(X, x+\epsilon)}^{-1}(t)\,\der t &= \int_{\delta-\rho}^{1-\rho} 1 + F_X^{-1}(t + F_X(x+\epsilon)-1)\,\der t \\
        &= \int_{F_X(x+\epsilon) - (1-\delta) -\rho}^{F_X(x+\epsilon)-\rho} 1 + F_X^{-1}(t)\,\der t \\
        &= \int_{F_X(x) - (1-\delta)}^{F_X(x)} 1 + F_X^{-1}(t)\,\der t \\
        &= (1-\delta)\cvar_\delta[C(X, x)]
    \end{align*}
    where the final equality follows by Lemma~\ref{lemma:cvar_cost_integral_general}. Dividing both sides by $(1-\delta)(x+\epsilon)$, we obtain
    \begin{equation} \label{eq:equal_lemma_cr_bound_1}
        \frac{1}{(1-\delta)(x+\epsilon)}\int_{\delta-\rho}^{1-\rho} F_{C(X, x+\epsilon)}^{-1}(t)\,\der t = \frac{\cvar_\delta[C(X, x)]}{x+\epsilon} < \alpha_\delta^\mu(x)
    \end{equation}
    since $\epsilon > 0$ implies $\frac{x}{x+\epsilon} < 1$. Similarly, if $\delta - \rho \leq 1 - F_X(x+\epsilon)$, the value of $F_{C(X, x+\epsilon)}^{-1}(t)$ depends on which part of the domain of integration contains $t$:
    \begin{align*}
        \int_{\delta-\rho}^{1-\rho} F_{C(X, x+\epsilon)}^{-1}(t)\,\der t &= \int_{\delta - \rho}^{1 - F_X(x+\epsilon)}x+\epsilon\,\der t + \int_{1 - F_X(x+\epsilon)}^{1 - \rho} 1 + F_X^{-1}(t + F_X(x+\epsilon)-1)\,\der t \\
        &= (1 - \delta - F_X(x))(x + \epsilon) + \int_{0}^{F_X(x)} 1 + F_X^{-1}(t)\,\der t.
    \end{align*}
    Dividing both sides by $(1-\delta)(x+\epsilon)$, we obtain
    \begin{align*}
        \frac{1}{(1-\delta)(x+\epsilon)}\int_{\delta-\rho}^{1-\rho} F_{C(X, x+\epsilon)}^{-1}(t)\,\der t &= \frac{1}{1-\delta}\left[(1 - \delta - F_X(x)) + \frac{1}{x+\epsilon}\int_{0}^{F_X(x)} 1 + F_X^{-1}(t)\,\der t\right] \\
        &< \frac{1}{(1-\delta)x}\left[(1 - \delta - F_X(x))x + \int_{0}^{F_X(x)} 1 + F_X^{-1}(t)\,\der t\right] \\
        &= \alpha_\delta^\mu(x), \tageq\label{eq:equal_lemma_cr_bound_2}
    \end{align*}
    where the final step is a consequence of Lemma~\ref{lemma:cvar_cost_integral_general}; note this exactly matches the bound \eqref{eq:equal_lemma_cr_bound_1} in the first case. For the second integral in \eqref{ineq:intermediate_two_integrals}, since $\rho = F_X(x+\epsilon) - F_X(x)$, we have $1 - \rho = 1 - F_X(x+\epsilon) + F_X(x) \geq 1 - F_X(x+\epsilon)$, so by Lemma~\ref{lemma:cvar_cost_inv_cdf_decision} we may calculate
    \begin{align*}
        \int_{1-\rho}^1 F_{C(X, x+\epsilon)}^{-1}(t)\,\der t &= \int_{1-\rho}^1 1 + F_X^{-1}(t + F_X(x+\epsilon)-1)\,\der t \\
        &= \int_{F_X(x)}^{F_X(x+\epsilon)} 1 + F_X^{-1}(t)\,\der t \\
        &\leq (F_X(x+\epsilon) - F_X(x))(1 + F_X^{-1}(F_X(x+\epsilon))) \tageq\label{ineq:using_monotonicity_of_inverse_cdf} \\
        &\leq (F_X(x+\epsilon) - F_X(x))(1 + x + \epsilon) \tageq\label{ineq:Finv_F_identity}
    \end{align*}
    where the bound \eqref{ineq:using_monotonicity_of_inverse_cdf} follows by monotonicity of the inverse CDF, and \eqref{ineq:Finv_F_identity} is from the well-known bound $F_X^{-1}(F_X(y)) \leq y$ (e.g., \cite[Lemma 1.17f]{wittingMathematischeStatistik1985}).

    Inserting \eqref{eq:equal_lemma_cr_bound_1}, \eqref{eq:equal_lemma_cr_bound_2}, and \eqref{ineq:Finv_F_identity} into \eqref{ineq:intermediate_two_integrals}, we obtain the bound
    \begin{equation} \label{ineq:dcr_bound_x_plus_eps}
        \alpha_\delta^\mu(x+\epsilon) < \alpha_\delta^\mu(x) + \frac{1}{1-\delta}\left(1 + \frac{1}{x+\epsilon}\right)(F_X(x+\epsilon) - F_X(x)).
    \end{equation}
    Right-continuity of the CDF ensures that the right-hand side of \eqref{ineq:dcr_bound_x_plus_eps} can be made at most $\alpha_\delta^\mu(x) + \gamma$ by choosing $\epsilon > 0$ sufficiently small, thus establishing the result.
\end{proof}

The second technical lemma we will need to prove the structural result asserts that if $X$ comes from a distribution with a probability atom at $x \in (0, 1]$, then there is a non-degenerate interval of slack in the $\dcr$. The assumed bound on the size of the probability atom is made without loss of generality due to Lemma~\ref{lemma:no_large_point_masses}, since we solely focus on the optimal algorithm.

\begin{lemma} \label{lemma:atom_implies_slack}
    Fix $\delta \in [0, 1)$, and let $X \sim \mu$ be a random variable supported in $[0, 1]$ with a probability atom of mass $\eta \leq (\alpha_\delta^* - 1)(1-\delta)$ at some $x \in (0, 1]$. Then there is some $\epsilon > 0$ and $\gamma > 0$ such that $\alpha_\delta^\mu(y) \leq \alpha_\delta^\mu(x) - \gamma$ for all $y \in [x-\epsilon, x)$.
\end{lemma}
\begin{proof}
    Let us assume that there exists an $\epsilon > 0$ such that $F_X(x - \epsilon) > 1-\delta$; the alternative case follows from an essentially identical argument. Then by Lemma~\ref{lemma:cvar_cost_integral_general},
    \begin{align*}
        \alpha_\delta^\mu(x) - \alpha_\delta^\mu(x - \epsilon) &= \frac{1}{1-\delta}\left[\frac{1}{x}\int_{F_X(x)-(1-\delta)}^{F_X(x)} 1 + F_X^{-1}(t)\,\der t - \frac{1}{x-\epsilon}\int_{F_X(x-\epsilon)-(1-\delta)}^{F_X(x-\epsilon)} 1 + F_X^{-1}(t)\,\der t\right] \\
        &\overset{\text{as $\epsilon \todown 0$}}{=} \frac{1}{(1-\delta)x}\left[\int_{F_X(x)-\eta}^{F_X(x)} 1 + F_X^{-1}(t)\,\der t - \int_{F_X(x) - \eta - (1-\delta)}^{F_X(x) - (1-\delta)} 1 + F_X^{-1}(t)\,\der t\right] \tageq\label{eq:second_integral_limit} \\
        &> 0,
    \end{align*}
    where \eqref{eq:second_integral_limit} follows by taking the limit $\epsilon \todown 0$ and by the assumption $\eta \leq (\alpha_\delta^* - 1)(1-\delta) < 1-\delta$. The final inequality follows from $F_X^{-1}(t) = x$ on $t \in (F_X(x) - \eta, F_X(x)]$, along with the fact that $F_X^{-1}(t) < x$ on $t \leq F_X(x) - \eta$, which implies that the integrand of the second integral in \eqref{eq:second_integral_limit} is strictly less than the first integrand (which is $1+x$), as $\eta < 1-\delta$. This strict inequality holds in the limit, so there exists some $\gamma > 0$ for which $\lim_{\epsilon \todown 0} \alpha_\delta^\mu(x) - \alpha_\delta^\mu(x - \epsilon) = 2\gamma$; that this holds in the limit implies that there is some $\epsilon > 0$ such that $\alpha_\delta^\mu(y) \leq \alpha_\delta^\mu(x) - \gamma$ for all $y \in [x-\epsilon, x)$, as desired.
\end{proof}

We are now prepared to prove the structural result establishing that the optimal algorithm has $\dcr$ that is independent of the adversary's choice of season duration $x \in (0, 1)$.

\begin{lemma} \label{lemma:any_s_tight}
    Let $\delta \in [0, 1)$, and let $\mu$ be an algorithm with optimal $\dcr$ for continuous-time ski rental, so $\alpha_\delta^{\mu} = \alpha_\delta^*$. Then $\alpha_\delta^{\mu}(x) = \alpha_\delta^{\mu}$ for all $x \in (0, 1)$.
\end{lemma}
\begin{proof}
    Suppose for the sake of contradiction that $\alpha_\delta^\mu(x) < \alpha_\delta^\mu$ for some $x \in (0, 1)$. We will construct another algorithm $\hat{\mu}$ with $\alpha_\delta^{\hat{\mu}} \leq \alpha_\delta^\mu$ and $\alpha_\delta^{\hat{\mu}}(1) < \alpha_\delta^{\hat{\mu}}$, which by Lemma~\ref{lemma:s_1_tight} immediately implies that $\hat{\mu}$, and therefore $\mu$, is not optimal. In the following proof, we will say that $\mu$ has ``slack'' at $x$ when $\alpha_\delta^\mu(x) < \alpha_\delta^\mu$.

    Define $x = \sup\{y \in [0, 1] : \alpha_\delta^\mu(y) < \alpha_\delta^\mu\}$. If $x = 1$ and $\alpha_\delta^\mu(x) < \alpha_\delta^\mu$, clearly by Lemma~\ref{lemma:s_1_tight} we're done. On the other hand, if $x < 1$, we must have $\alpha_\delta^\mu(x) = \alpha_\delta^\mu$ by the supremum definition of $x$ and Lemma~\ref{lemma:dcr_slack_interval}. Thus, we may proceed assuming that $x \in (0, 1]$ and $\alpha_\delta^\mu(x) = \alpha_\delta^\mu$. Note that the interval $(x, 1]$, if it is nonempty, cannot contain any probability atoms, by Lemma~\ref{lemma:atom_implies_slack} and the definition of $x$.

    We will now proceed to prove the result in two parts. First, we will show that we can construct a distribution $\tilde{\mu}$ with no worse $\dcr$ than $\mu$ that has slack at $x$, i.e., $\alpha_\delta^{\tilde{\mu}}(x) < \alpha_\delta^{\tilde{\mu}}$. Then, we will show that we can construct new distributions iteratively propagating this slack toward $1$, eventually culminating with the desired $\hat{\mu}$ with $\alpha_\delta^{\hat{\mu}} \leq \alpha_\delta^\mu$ and $\alpha_\delta^{\hat{\mu}}(1) < \alpha_\delta^{\hat{\mu}}$.

    \paragraph{Part (1): Obtaining slack at $x$} We break into two cases depending on whether $\mu$ has a probability atom at $x$.
    \begin{enumerate}[(a)]
        \item Suppose $\mu$ has a probability atom at $x$ of size $\zeta$; by Lemma~\ref{lemma:no_large_point_masses}, $\zeta$ must be bounded as
        $$\zeta \leq (\alpha_\delta^* - 1)(1-\delta) < 1-\delta.$$
        Then by Lemma~\ref{lemma:atom_implies_slack}, there is an $\epsilon > 0$ and $\gamma > 0$ such that $\alpha_\delta^\mu(y) \leq \alpha_\delta^\mu - \gamma$ for all $y \in [x-\epsilon, x)$. Define a measure $\tilde{\mu}$ from $\mu$ by moving a small amount of mass $\eta > 0$ from $x$ to $x - \epsilon$; it can easily be seen that this will not change $\alpha_\delta^\mu(y)$ for $y < x-\epsilon$, it will strictly decrease $\alpha_\delta^\mu(x)$ due to shifting of its mass to an action with a smaller cost, and similarly it will either decrease or not affect $\alpha_\delta^\mu(y)$ for $y > x$. Note that Lemma~\ref{lemma:no_large_point_masses}'s bound on $\zeta$'s size $\zeta \leq (\alpha_\delta^* - 1)(1-\delta) < 1-\delta$ is crucial to obtain that $\alpha_\delta^\mu(x)$ strictly decreases, since if $\zeta$ were larger than $1-\delta$, decreasing its mass by a small amount might not change the $\dcr$ at $x$.

        On the interval $[x-\epsilon, x)$, the $\dcr$ will increase, but at most by $\frac{\eta}{1-\delta}\left(1 + \frac{1}{x-\epsilon}\right)$, which by choosing $\eta$ small can be kept sufficiently small that $\alpha_\delta^{\tilde{\mu}}(y) < \alpha_\delta^\mu$ remains true for all $y \in [x-\epsilon, x)$. To see this, note that this movement of mass increases $F_X$ by $\eta$ on the interval $[x-\epsilon, x)$; as a result, on the interval $(F_X(x-\epsilon), F_X(x-\epsilon)+\eta]$, $F_X^{-1}$ will decrease to $x-\epsilon$. That is, $\tilde{X} \sim \tilde{\mu}$ will have inverse CDF of the form
        \begin{equation} \label{eq:move_atom_inv_cdf}
            F_{\tilde{X}}^{-1}(p) = \begin{cases} F_X^{-1}(p) &\text{if $p \leq F_X(x-\epsilon)$} \\ x-\epsilon &\text{if $p \in (F_X(x-\epsilon), F_X(x-\epsilon) + \eta]$} \\
            F_X^{-1}(p-\eta) &\text{if $p \in (F_X(x-\epsilon) + \eta, F_X(x)]$} \\
            F_X^{-1}(p) &\text{otherwise}
            \end{cases}
        \end{equation}

        Assuming that $F_X(x) > 1-\delta$ (the alternative case proceeds similarly), we may compute, using Lemma~\ref{lemma:cvar_cost_integral_general}:
        \begin{align*}
            (1-\delta)(x-\epsilon) \cdot \alpha_\delta^{\tilde{\mu}}(x-\epsilon) &=
            \int_{F_{\tilde{X}}(x-\epsilon) - (1-\delta)}^{F_{\tilde{X}}(x-\epsilon)} 1 + F_{\tilde{X}}^{-1}(t)\,\der t \\
            &= \int_{F_X(x-\epsilon) + \eta - (1-\delta)}^{F_X(x-\epsilon) + \eta} 1 + F_{\tilde{X}}^{-1}(t)\,\der t \\
            &= \eta (1 + x - \epsilon) + \int_{F_X(x-\epsilon) + \eta - (1-\delta)}^{F_X(x-\epsilon)} 1 + F_X^{-1}(t)\,\der t \\
            &\leq \eta(1 + x - \epsilon) + (1-\delta)(x-\epsilon)\cdot\alpha_\delta^\mu(x-\epsilon).
        \end{align*}
        implying that $\alpha_\delta^{\tilde{\mu}}(x-\epsilon) \leq \alpha_\delta^\mu(x-\epsilon) + \frac{\eta}{1-\delta}\left(1 + \frac{1}{x-\epsilon}\right)$.
        Similarly, for any $y \in (x-\epsilon, x)$ with $F_X(y) - F_X(x) < 1-\delta$ (any others are not impacted by this change), we can obtain an analogous bound:
        \begin{align*}
            (1-\delta)y \cdot \alpha_\delta^{\tilde{\mu}}(y) \leq \eta(1 + x - \epsilon) + (1-\delta)y \cdot \alpha_\delta^{\mu}(y)
        \end{align*}
        implying that $\alpha_\delta^{\tilde{\mu}}(y) \leq \alpha_\delta^\mu(y) + \frac{\eta}{1-\delta}\frac{1 + x - \epsilon}{y} \leq \alpha_\delta^\mu(y) + \frac{\eta}{1-\delta}\left(1 + \frac{1}{x-\epsilon}\right)$. Note that these bounds are coarse, and intuitively capture the idea that if we introduce a new point mass of size $\eta$ within the support of the worst-case loss subpopulation distribution realizing the $\dcvar$, the worst that this added loss can do is increase the $\dcvar$ in proportion to the loss value that it adds, weighted by its probability normalized by $1-\delta$. Thus, it is clear that by selecting $\eta$ sufficiently small, we can guarantee that $\alpha_\delta^{\tilde{\mu}}(y) < \alpha_\delta^\mu(y)$ for all $y \in [x-\epsilon, x)$, while we have still strictly decreased $\alpha_\delta^\mu(x)$ by moving some of its mass to an earlier decision, thus introducing slack at $x$.

        \item Suppose that $\mu$ does not have a point mass at $x$; thus, $F_X$ is continuous at $x$. We break into two further cases:

        \begin{enumerate}[(i)]
            \item Suppose that there is some $y < x$ such that $\mu$ has a point mass at $y$ and
            $$\lim_{h \todown 0} F_X(x)-F_X(y-h) < 1-\delta.$$
            By the argument in part (a), we can construct a measure $\tilde{\mu}$ by moving some small amount $\eta > 0$ of mass from $y$ to $y - \epsilon$ for some $\epsilon > 0$, which will not impact $\alpha_\delta^\mu(z)$ for $z < y-\epsilon$, will increase it in a controlled manner (such that we can maintain slack by choosing $\eta$ sufficiently small) for $z \in [y-\epsilon, y)$, and will strictly decrease it for $z = y$. We can choose $\epsilon$ sufficiently small that $F_X(x) - F_X(y-\epsilon) \leq 1-\delta$, by the strict inequality in the limit assumed above. Moreover, this modification will also strictly decrease $\alpha_\delta^\mu(z)$ for $z \in [y, x]$. In particular, for $x$ we have, using the inverse CDF expression \eqref{eq:move_atom_inv_cdf} and Lemma~\ref{lemma:cvar_cost_integral_general},
            \begin{align*}
                (1-\delta)x\cdot \alpha_\delta^{\tilde{\mu}}(x) &= \int_{F_{\tilde{X}}(x)-(1-\delta)}^{F_{\tilde{X}}(x)}1 + F_{\tilde{X}}^{-1}(t)\,\der t \\
                &= \int_{F_{X}(x)-(1-\delta)}^{F_{X}(x)}1 + F_{\tilde{X}}^{-1}(t)\,\der t \\
                &= \int_{F_X(x) - (1-\delta)}^{F_X(y)}1 + F_{\tilde{X}}^{-1}(t)\,\der t + \int_{F_X(y)}^{F_X(x)}1 + F_{\tilde{X}}^{-1}(t)\,\der t \\
                &< \int_{F_X(x) - (1-\delta)}^{F_X(y)}1 + F_{X}^{-1}(t)\,\der t + \int_{F_X(y)}^{F_X(x)}1 + F_{X}^{-1}(t)\,\der t \\
                &= (1-\delta)x\cdot \alpha_\delta^\mu(x),
            \end{align*}
            where the strict inequality results from $F_X^{-1}(t)$ being strictly decreased on the domain \\$(F_X(y-\epsilon), F_X(y-\epsilon)+\eta]$
            and the choice of $\epsilon$ satisfying $F_X(x) - F_X(y-\epsilon) \leq 1-\delta$. In the above derivation we have assumed that $F_X(x) > 1-\delta$; the alternative case proceeds similarly.

            Thus, if the worst $(1-\delta)$ fraction of loss outcomes when the adversary chooses $x$ contains decisions (of positive probability) from a non-degenerate set $[y-\epsilon, y]$, and $y$ has a probability atom, we have that introducing slack at $y$ in turn introduces slack at $x$.

            \item Suppose that $\mu$ has no point mass at any $y$ satisfying the property that $y < x$ and \\$\lim_{h \todown 0} F_X(x)-F_X(y-h) < 1-\delta$. This implies that there must be a small interval to the left of $x$ on which $F_X$ is continuous, since $F_X$ must increase from $F_X(x) - (1-\delta)$ to $F_X(x)$ without any discontinuities. There thus must exist, by the supremum definition of $x$, some $y < x$ at which $F_X$ is continuous, $\alpha_\delta^\mu(y) < \alpha_\delta^\mu$, and $F_X(x) - F_X(y) < 1-\delta$. Continuity of the $\dcr$ on this subinterval, and the bound \eqref{ineq:dcr_bound_x_plus_eps} in particular, imply that there is some non-degenerate interval $[y, y+\epsilon]$ and $\gamma > 0$ such that $\alpha_\delta^\mu(z) \leq \alpha_\delta^\mu - \gamma$ for all $z \in [y, y+\epsilon]$. By the assumption that $\alpha_\delta^\mu(x) = \alpha_\delta^\mu$ while there are no point masses on the interval $[y, x]$, the bound \eqref{ineq:dcr_bound_x_plus_eps} in particular tells us that we can choose $\epsilon$ such that the half-open interval $(y, y+\epsilon]$ has strictly positive measure. Then suppose we move a small fraction $\eta$ of the probability mass on $(y, y+\epsilon]$ to $y$. By the same basic argument as employed previously, this will not affect the $\dcr$ $\alpha_\delta^\mu(z)$ for $z < y$, and if we choose $\eta$ small enough, it will increase $\alpha_\delta^\mu(z)$ for $z \in [y, y+\epsilon)$ in a controlled fashion so that we can keep $\alpha_\delta^{\tilde{\mu}}(z) < \alpha_\delta^\mu$. Moreover, by the assumption $F_X(x) - F_X(y) < 1-\delta$, $\alpha_\delta^\mu(x)$ will strictly decrease, just as it did in the previous subcase. Thus, we can introduce slack at $x$ while not increasing the $\dcr$.
        \end{enumerate}
    \end{enumerate}

    Having obtained a measure $\tilde{\mu}$ with a $\dcr$ not worse than $\mu$ and with slack at $x$, we now proceed to the second part.

    \paragraph{Part (2): Obtaining slack at $1$} If $x = 1$, we are done; otherwise, recall that by definition $(x, 1]$ cannot have any probability atoms, since this would introduce slack in the interval by Lemma~\ref{lemma:atom_implies_slack}, so $F_{\tilde{X}}$ is continuous on $[x, 1]$ (note that at $x$ itself, it may only be right-continuous). The argument employed to obtain slack at $1$ follows an iterated form of the approach in case (b.ii) from part 1: Because $\tilde{\mu}$ has slack at $x$, the bound \eqref{ineq:dcr_bound_x_plus_eps} implies that it has slack in an interval $[x, x+\epsilon]$ with the property that $\tilde{\mu}(x, x+\epsilon] > 0$. Then we may transfer a fraction $\eta$ of the probability mass on $(x, x+\epsilon]$ to $x$, and as long as $\eta$ is chosen sufficiently small, this will increase $\alpha_\delta^{\tilde{\mu}}(y)$ for $y \in [x, x+\epsilon)$ while maintaining slack, and it will strictly decrease $\alpha_\delta^{\tilde{\mu}}(z)$ for all $z \geq x+\epsilon$ such that $F_{\tilde{X}}(z) - F_{\tilde{X}}(x) \leq 1-\delta$ (note we can choose $\epsilon$ so that this set of $z$ is nonempty). Thus we can ``propagate'' the slack in the $\dcr$ to decisions whose CDF value is up to $1-\delta$ greater than that of the original slack point, $x$, without increasing the $\dcr$ of the algorithm. Iteratively applying this construction at most $\calO(\frac{1}{1-\delta})$ times, we eventually obtain an algorithm $\hat{\mu}$ with $\dcr$ no worse than $\mu$, and with slack at $1$: $\alpha_\delta^{\hat{\mu}}(1) < \alpha_\delta^{\hat{\mu}} \leq \alpha_\delta^\mu$.
\end{proof}

Using the structural characterization of the optimal algorithm in terms of its $\dcr$'s indifference to the adversary's choice of ski season duration, we may now prove that the optimal algorithm has a CDF that is both strictly increasing and continuous (i.e., one-to-one) on the interval $[0, 1]$.

\begin{lemma} \label{lemma:cont_time_optimal_strictly_increasing}
    Fix $\delta \in [0, 1)$, and let $X \sim \mu$ be a random variable yielding the optimal $\dcr$ for continuous-time ski rental, i.e., $\alpha_\delta^\mu = \alpha_\delta^*$. Then $F_X$ is strictly increasing on $[0, 1]$.
\end{lemma}
\begin{proof}
    When $\delta \in (0, 1)$, this is an immediate consequence of the strict inequality \eqref{ineq:dcr_bound_x_plus_eps} in the proof of Lemma~\ref{lemma:dcr_slack_interval}: if $F_X$ weren't strictly increasing, there would exist a non-degenerate interval $[a, b] \subseteq [0, 1]$ with $F_X(x) = c$ for all $x \in [a, b]$, which by \eqref{ineq:dcr_bound_x_plus_eps} would imply that $\alpha_\delta^\mu(a + \epsilon) < \alpha_\delta^\mu(a)$ for some small $\epsilon > 0$, contradicting (via Lemma~\ref{lemma:any_s_tight}) the optimality of $\mu$. Likewise, in the $\delta = 0$ case, $F_X$ not strictly increasing means $F_X(a + \epsilon) = F_X(a)$ for all sufficiently small $\epsilon$, so the expression \eqref{eq:expected_cr_expression_cdf} in the proof of Lemma~\ref{lemma:dcr_slack_interval} implies that $\alpha_0^\mu$ will strictly decrease for some small interval starting at $a$, again contradicting the optimality of $\mu$ by Lemma~\ref{lemma:any_s_tight}.
\end{proof}

\begin{lemma} \label{lemma:cont_time_optimal_continuous}
    Fix $\delta \in [0, 1)$, and let $X \sim \mu$ be a random variable yielding the optimal $\dcr$ for continuous-time ski rental, i.e., $\alpha_\delta^\mu = \alpha_\delta^*$. Then $F_X$ is continuous on $[0, 1]$, $F_X(0) = 0$, and $F_X(1) = 1$.
\end{lemma}
\begin{proof}
    The first two properties amount to proving that the optimal algorithm contains no probability atoms; the third is always satisfied, since we can assume without loss of generality that $X$ is supported on $[0, 1]$ (Lemma~\ref{lemma:cont_ski_rental_01_support}). Suppose for the sake of contradiction that $\mu$ has an atom at $x \in [0, 1]$. If $x = 0$, $\mu$ cannot be competitive, let alone optimal, yielding a contradiction and establishing $F_X(0) = 0$. Otherwise, if $x \in (0, 1]$, Lemma~\ref{lemma:atom_implies_slack} implies that there is slack in the $\dcr$ -- i.e., there is some $\epsilon > 0$ such that $\alpha_\delta^\mu(x-\epsilon) < \alpha_\delta^\mu(x)$, which by Lemmas~\ref{lemma:s_1_tight} and \ref{lemma:any_s_tight} contradicts the optimality of $\mu$. Thus $\mu$ has no atoms and $F_X$ is continuous on $[0, 1]$.
\end{proof}

We are now adequately equipped to prove the main result of this section, Theorem~\ref{theorem:optimal_cont_time_ski_rental}.

\begin{proof}[Proof of Theorem~\ref{theorem:optimal_cont_time_ski_rental}]
    Let $\phi : [0, 1] \to [0, 1]$ be the inverse CDF of the optimal strategy for continuous-time ski rental with the $\dcr$ metric. By Lemmas~\ref{lemma:cont_time_optimal_strictly_increasing} and \ref{lemma:cont_time_optimal_continuous}, $\phi$ is strictly increasing and continuous on $[0, 1]$, and hence one-to-one, with $\phi(0) = 0$ and $\phi(1) = 1$. By Lemma~\ref{lemma:cvar_cost_integral_general}, for any fixed adversary decision $s \in (0, 1]$, we may express the $\dcvar$ of the cost incurred by playing $X \sim \phi$ as
    $$\cvar_\delta[s \cdot \indic_{X > s} + (X + 1) \cdot \indic_{X \leq s}] = \begin{cases} \frac{1}{1-\delta}\left[(1 \!-\! \delta \!-\! \phi^{-1}(s)) s + \!\int_0^{\phi^{-1}(s)} (1 \!+\! \phi(p))\,\der p\right] & \text{if $\phi^{-1}(s) \leq 1 \!-\! \delta$} \\ \frac{1}{1-\delta}\int_{\phi^{-1}(s) - (1 - \delta)}^{\phi^{-1}(s)} (1 + \phi(p))\,\der p &\text{otherwise.}\end{cases}$$

    By Lemmas~\ref{lemma:s_1_tight} and \ref{lemma:any_s_tight}, the optimal algorithm has competitive ratio independent of the adversary's choice $s \in (0, 1]$ of the ski season duration -- that is, $\alpha_\delta^*(s) = \alpha_\delta^*$ for all $s \in (0, 1]$. As such, $\phi$ must satisfy the following equations:
    \begin{align*}
        \frac{1}{1-\delta}\left[(1 - \delta - \phi^{-1}(s)) s + \int_0^{\phi^{-1}(s)} (1 + \phi(p))\,\der p\right] &= \alpha_\delta^*\cdot s &\text{if $\phi^{-1}(s) \leq 1 - \delta$} \\
        \frac{1}{1-\delta}\int_{\phi^{-1}(s) - (1 - \delta)}^{\phi^{-1}(s)} (1 + \phi(p))\,\der p &= \alpha_\delta^* \cdot s & \text{otherwise}
    \end{align*}
    for all $s \in (0, 1]$. Because $\phi$ is one-to-one, the above equations holding for all $s \in (0, 1]$ is equivalent to their holding for all $t \in (0, 1]$ when $s \coloneqq \phi(t)$ (and hence $t = \phi^{-1}(s)$):
    \begin{align*}
        \frac{1}{1-\delta}\left[(1 - \delta - t) \phi(t) + \int_0^{t} (1 + \phi(p))\,\der p\right] &= \alpha_\delta^*\cdot \phi(t) &\text{if $t \leq 1 - \delta$} \tageq\label{eq:delay_diff_eq_ski_1}\\
        \frac{1}{1-\delta}\int_{t - (1 - \delta)}^{t} (1 + \phi(p))\,\der p &= \alpha_\delta^* \cdot \phi(t) & \text{otherwise.} \tageq\label{eq:delay_diff_eq_ski_2}
    \end{align*}
    Differentiating \eqref{eq:delay_diff_eq_ski_1} with respect to $t$, we find
    \begin{align*}
        \phi'(t) + \frac{1}{1-\delta}\left[-\phi(t) - t\phi'(t) + 1 + \phi(t)\right] &= \alpha_\delta^* \cdot \phi'(t) \\
        \implies \phi'(t) = \frac{1}{(\alpha_\delta^* - 1)(1-\delta) + t},
    \end{align*}
    and integrating this and applying the initial condition $\phi(0) = 0$, we obtain $\phi(t) = \log\left(1 + \frac{t}{(\alpha_\delta^*-1)(1-\delta)}\right)$ on $t \in [0, 1-\delta]$. Differentiating the second equation \eqref{eq:delay_diff_eq_ski_2} and rearranging, we obtain
    \begin{equation} \label{eq:delay_diff_eq_statement_proof}
        \phi'(t) = \frac{1}{\alpha_\delta^*(1-\delta)}\left[\phi(t) - \phi(t-(1-\delta))\right].
    \end{equation}
    Thus, as claimed, $\phi$ is the solution to the delay differential equation \eqref{eq:delay_diff_eq_statement_proof} subject to the initial condition $\phi(t) = \log\left(1 + \frac{t}{(\alpha_\delta^*-1)(1-\delta)}\right)$ on $t \in [0, 1-\delta]$. Uniqueness of $\phi$ as the optimal strategy follows from uniqueness results in the theory of delay differential equations, e.g., \cite[Theorem 3.1]{bellmanDifferentialDifferenceEquations1963}.
    
\end{proof}

\subsection{Optimal solution is strictly decreasing in $\alpha$ \label{appendix:opt_soln_delay_diffeq_strictly_decreasing_alpha}}

Here, we will prove that the solution $\phi(t)$ to the delay differential equation posed in Theorem~\ref{theorem:optimal_cont_time_ski_rental} is strictly decreasing in $\alpha$ for $t \in (0, 1]$. Let $\phi_\alpha(t)$ denote the solution at time $t$ for a given choice of $\alpha$. We will employ a form of induction on the continuum in our argument:

\begin{definition}[Induction on the continuum, \citep{hathawayUsingContinuityInduction2011}]
    Let $P(t) : \R \to \{0, 1\}$ be a truth-valued function, and let $[a, b]$ be a closed interval. Suppose the following three properties hold:
    \begin{enumerate}[(1)]
        \item $P(a) = 1$;

        \item For any $x \in [a, b)$, $P(t) = 1$ for all $t \in [a, x]$ implies $P(t) = 1$ in some non-degenerate interval $[x, x+\epsilon)$;
        \item For any $x \in (a, b]$, $P(t) = 1$ for all $t \in [a, x)$ implies $P(x) = 1$.
    \end{enumerate}
    Then $P(t) = 1$ for all $t \in [a, b]$.
\end{definition}

Fix $\alpha < \alpha'$; in our setting, $P(t)$ will be the truth function of the strict inequality $\phi_\alpha(t) > \phi_{\alpha'}(t)$, and the interval of interest will be $t \in [1-\delta, 1]$.

First, note that we clearly have $P(1-\delta) = 1$, since $\phi_\alpha(1-\delta) = \log\left(1+\frac{1}{\alpha-1}\right)$, which is strictly decreasing in $\alpha$; in fact, we have that the initial condition $\phi(t) = \log\left(1+\frac{t}{\alpha-1}\right)$ is strictly decreasing for all $t \in (0, 1-\delta]$. Thus property (1) is satisfied.

Second, note that the solutions $\phi_\alpha, \phi_{\alpha'}$ are both continuous, as they are the solutions to a delay differential equation. Continuity guarantees that if $\phi_\alpha(t) > \phi_{\alpha'}(t)$, then $\phi_\alpha(x) > \phi_{\alpha'}(x)$ for $x$ in some interval $[t, t+\epsilon]$ with $\epsilon > 0$. Thus property (2) holds.

Finally, suppose $\phi_\alpha(t) > \phi_{\alpha'}(t)$ for $t$ in some interval $[1-\delta, x)$. By the integral form \eqref{eq:delay_diff_eq_ski_2} of the delay differential equation, we have

$$\phi_\alpha(x) = \frac{1}{\alpha(1-\delta)} \int_{x-(1-\delta)}^x 1 + \phi_\alpha(t)\,\der t > \frac{1}{\alpha '(1-\delta)} \int_{x-(1-\delta)}^x 1 + \phi_{\alpha'}(t)\,\der t = \phi_{\alpha'}(x),$$
where the strict inequality follows from $\alpha < \alpha'$ and the assumption that $\phi_\alpha(t) > \phi_{\alpha'}(t)$ for $t$ in $[1-\delta, x)$ (this inequality also holds for the initial condition on the region $(0, 1-\delta]$). Thus property (3) holds, and we conclude that $\phi_\alpha(t) > \phi_{\alpha'}(t)$ for all $t \in (0, 1]$.

\subsection{Analytic solution when $\delta \leq \frac{1}{2}$ \label{appendix:cont_time_analytic_solution_small_delta}}

When $\delta \leq \frac{1}{2}$, the delay differential equation on the domain $[1-\delta, 1]$ can be written as an ordinary differential equation by substituting the initial condition in for $\phi(t - (1-\delta))$:
$$\phi'(t) = \frac{1}{\alpha(1-\delta)}\left[\phi(t) - \log\left(1 + \frac{t-(1-\delta)}{(\alpha-1)(1-\delta)}\right)\right]$$
with initial value $\phi(1-\delta) = \log\left(1 + \frac{1}{\alpha-1}\right)$. Solving this initial value problem with Mathematica, we obtain that $\phi$ takes the value
\begin{align*}
    \phi(t) = &\,\,e^{-\frac{2(1-\delta)-t}{\alpha(1-\delta)}}\left[e\cdot\Ei\left(\frac{1}{\alpha} - 1\right) - e\cdot\Ei\left(\frac{(2-\alpha)(1-\delta)-t}{\alpha(1-\delta)}\right) + e^{\frac{1}{\alpha}}\log\left(\frac{\alpha}{\alpha-1}\right)\right]  \\
    &\,\,+ \log\left(1+\frac{t-(1-\delta)}{(\alpha-1)(1-\delta)}\right),
\end{align*}
on the interval $[1-\delta, 1]$, where $\Ei(x) = -\int_{-x}^\infty \frac{e^{-t}}{t}\,\der t$ is the exponential integral.

\subsection{Proof of Theorem~\ref{theorem:cont_time_ski_lower_bound} \label{appendix:theorem:cont_time_ski_lower_bound}}

\begin{proof}[noname]
    This proof is essentially a tightening of the phase transition bound in the discrete setting (Theorem~\ref{theorem:phase_transition}(ii); see Appendix~\ref{appendix:theorem:phase_transition}), enabled by the continuity of the decision space $[0, 1]$. The lower bound of $\frac{e}{e-1}$ holds trivially, since increasing $\delta$ increases the $\dcr$; thus we focus on the second element in the $\max$. Let $\mu$ be an arbitrary distribution supported in $[0, 1]$.

    Let $n = \big\lfloor \frac{1}{1-\delta} \big\rfloor - 1$, and define the indices $i_0 = 0$, $i_k = 1-\frac{1}{2^k}$ for each $k \in [n]$, and $i_{n+1} = 1$. We partition the unit interval into $n+1$ different sets $I_k$, with
    $$I_k = [i_{k-1}, i_k]$$
    for each $k \in [n+1]$. Because $n + 1 = \big\lfloor \frac{1}{1-\delta} \big\rfloor$, by the pigeonhole principle there must be at least one interval $I_k$ with $\mu(I_k) \geq 1-\delta$. Let the adversary's choice of ski season duration be $i_k$, and suppose $k \leq n$. Because $\mu(I_k) \geq 1-\delta$, there is at least a $(1-\delta)$ fraction of outcomes in which the algorithm's decision $x$ lies in $I_k$, and thus we can lower bound the $\dcr$ as
    \begin{align*}
        \alpha_\delta^\mu(i_k) &= \frac{\dcvar[i_k \cdot \indic_{X > i_k} + (X+1)\cdot\indic_{X \leq i_k}]}{i_k} \\
        &\geq \frac{i_{k-1} + 1}{i_k} \\
        &= \frac{2 - \frac{1}{2^{k-1}}}{1-\frac{1}{2^k}} \\
        &= 2.
    \end{align*}
    
    Alternatively, suppose $k = n+1$. Again, there is at least a $(1-\delta)$ fraction of outcomes in which the algorithm's decision $x$ lies in $I_{n+1}$, so we can lower bound the $\dcr$ as
    \begin{align*}
        \alpha_\delta^\mu(i_{n+1}) &= \frac{\dcvar[i_{n+1} \cdot \indic_{X > i_{n+1}} + (X+1)\cdot\indic_{X \leq i_{n+1}}]}{i_{n+1}} \\
        &\geq \frac{i_n + 1}{i_{n+1}} \\
        &= 2-\frac{1}{2^n}
    \end{align*}
    
    Thus, regardless of which set $I_k$ contains the at least $(1-\delta)$ fraction of mass, the $\dcr$ will be at least $2-\frac{1}{2^n}$. Substituting in the definition of $n$, we obtain
    $$\alpha_\delta^\mu \geq 2-\frac{1}{2^{\lfloor \frac{1}{1-\delta} \rfloor - 1}}$$
    for any algorithm $\mu$.
    
\end{proof}

\section{Proofs for Section~\ref{section:discrete_ski_rental}}

\subsection{Proof of Theorem~\ref{theorem:phase_transition} \label{appendix:theorem:phase_transition}}

We will prove parts (i) and (ii) of the theorem separately.

\begin{proof}[Proof of Theorem~\ref{theorem:phase_transition}, part (i)]
    Recall that $\alpha_\delta^{\mathrm{DSR}(B), *} \leq \alpha_\delta^{\mathrm{CSR}, *}$ for all $B \in \N$, since the discrete-time version of the ski rental amounts to restricting the continuous adversary's power. The bound in Theorem~\ref{theorem:cvar_skirental_firstalg} thus implies that $\alpha_\delta^{B, *} \leq 2 - \frac{1}{e^{\frac{c}{1-\delta}}-1}$, so a sufficient condition for $\alpha_\delta^{B, *}$ to strictly improve on the deterministic worst-case competitive ratio of $2 - \frac{1}{B}$ is to have $2 - \frac{1}{e^{\frac{c}{1-\delta}}-1} < 2 - \frac{1}{B}$. Rearranging this equation to isolate $\delta$ yields $\delta < 1 - \frac{c}{\log(B+1)}$, as claimed.
\end{proof}

\begin{proof}[Proof of Theorem~\ref{theorem:phase_transition}, part (ii)]
    Let $n = \log_2 B$; note that $2^{\lfloor n \rfloor} \leq B \leq 2^{\lceil n \rceil}$, and hence $1 \leq \frac{B}{2^{\lfloor n \rfloor}} < 2$ -- the second inequality is strict because if it held with equality, this would imply $B = 2^{\lfloor n \rfloor + 1}$, or $\log_2 B = n = \lfloor n \rfloor + 1$. Define the (not necessarily integer) indices $i_k = \frac{(2^k-1)B}{2^k}$ for every $k \in \{0, \ldots, \lfloor n \rfloor\}$. We construct the following sets: for each $k \in [\lfloor n \rfloor]$, define:
    \begin{align*}
        I_k &= \{\lceil i_{k-1}\rceil + 1, \ldots, \lfloor i_k \rfloor \}, \\
        J_k &= \{\lceil i_k \rceil\},
    \end{align*}
    and define $J_{\lfloor n \rfloor + 1} = \{B\}$. It is straightforward to observe that these sets form a cover of the action set $[B]$:
    $$\left(\bigcup_{k=1}^{\lfloor n \rfloor} I_k\right) \cup \left(\bigcup_{k=1}^{\lfloor n \rfloor + 1} J_k\right) = [B].$$
    To see this, simply note that $\lceil i_0\rceil + 1 = 0 + 1 = 1$ and
    $$\lceil i_{\lfloor n \rfloor} \rceil = \Big\lceil B - \frac{B}{2^{\lfloor n \rfloor}} \Big \rceil = B - 1,$$
    since $\frac{B}{2^{\lfloor n \rfloor}} < 2$.

    Now suppose that $\frac{1}{1-\delta} \geq 2\lfloor n \rfloor + 1$; then for any strategy $\pvec \in \Delta_B$, the pigeonhole principle assures us that $\pvec$ must assign probability at least $(1-\delta)$ to one of the action sets $I_k$ or $J_k$ in the cover. However, for each of these action sets, there is an adversary decision that forces each action in the set to have competitive ratio at least $2 - \frac{1}{B}$. If the set in question is $I_k$ for some $k \in [\lfloor n \rfloor]$, then if the adversary chooses the true number of skiing days to be $\lfloor i_k \rfloor$, the action in $I_k$ with the smallest competitive ratio is $\lceil i_{k-1} \rceil + 1$, which has competitive ratio lower bounded as
    \begin{align*}
        \frac{B + (\lceil i_{k-1} \rceil + 1) - 1}{\lfloor i_k \rfloor} &\geq \frac{B + i_{k-1}}{i_k} \\
        &= \frac{1 + \frac{2^{k-1}-1}{2^{k-1}}}{\frac{2^k-1}{2^k}} \\
        &= 2.
    \end{align*}
    On the other hand, if the set in question is one of the singleton sets $J_k = \{x\}$, then if the adversary chooses the true number of skiing days as $x$, the competitive ratio of this action is lower bounded as
    \begin{align*}
        \frac{B + x - 1}{x} &= 1 + \frac{B-1}{x} \\
        &= 2 - \frac{1}{B} + \left(\frac{B-1}{x} - \frac{B-1}{B}\right) \\
        &\geq 2 - \frac{1}{B}
    \end{align*}
    since $x \in [B]$, and in particular $x \leq B$. Thus, for each set in this cover, there is an adversary decision forcing every action in the set to have competitive ratio at least $2 - \frac{1}{B}$. However, since one of these sets $S$ must have probability at least $1-\delta$ assigned to it by $\pvec$, the ``bad'' adversary decision corresponding to $S$ will yield a competitive ratio of at least $2-\frac{1}{B}$ with probability at least $1-\delta$. It immediately follows that the adversary can force a $\dcr$ of at least $2-\frac{1}{B}$ in this setting, which implies that the optimal strategy is to buy deterministically at time $B$, which has a $\dcr$ of exactly $2 - \frac{1}{B}$.

    Rearranging the condition $\frac{1}{1-\delta} \geq 2\lfloor n \rfloor + 1$ to isolate $\delta$, we obtain $\delta \geq 1 - \frac{1}{2\lfloor n \rfloor + 1} = 1 - \frac{1}{2\lfloor \log_2 B \rfloor + 1}$, as claimed.
\end{proof}

Note that the lower bound on $\delta$ obtained in the above proof can be improved to $1-\frac{1}{\log_2 B + 1}$ when $B$ is a power of $2$, as in this case all of the sets $J_k$ for $k < n + 1$ are redundant, so eliminating these, the resulting cover has only $n+1$ sets. This is essentially equivalent to the argument used in the continuous-time lower bound (cf. Appendix~\ref{appendix:theorem:cont_time_ski_lower_bound}).

\subsection{Proof of Theorem~\ref{theorem:discrete_analytic_result} \label{appendix:theorem:discrete_analytic_result}}

Before proving the theorem, we will first prove an structural lemma analogous to the tightness results Lemmas~\ref{lemma:s_1_tight} and \ref{lemma:any_s_tight} in the continuous-time setting, which will establish that so long as $\delta$ is not too large, the optimal algorithm $\pvec^{B, \delta, *}$ satisfies the property that $\alpha_\delta^{B, \pvec^{B, \delta, *}}(i) = \alpha_\delta^{B, *}$ for all $i \in [B]$. In other words, under the algorithm with optimal $\dcr$ for discrete-time ski rental, the adversary is indifferent to the ski season duration that it chooses. First, we pose an optimization-based formulation of the $\dcr$ of an arbitrary algorithm $\pvec$ that will facilitate the analysis.

\begin{lemma} \label{lemma:maximization_form_cvar_discrete}
    Let $\delta \in [0, 1)$, and let $\pvec \in \Delta_B$ be an algorithm for discrete-time ski rental with buying cost $B$. Then $\alpha_\delta^{B, \pvec}(i)$, the $\dcr$  when the adversary chooses the true number of skiing days as $i \in [B]$, can be expressed as
    \begin{align*}
        \alpha_\delta^{B, \pvec}(i) = \max_{\qvec \in \R^B}&\quad \frac{1}{1-\delta}\left(\sum_{j=1}^i \frac{B + j - 1}{i}q_j + \sum_{j=i+1}^B q_j\right) \tageq\label{eq:max_form_cvar_discrete}\\
        \mathrm{s.t.} &\quad \bv{0} \leq \qvec \leq \pvec \\
        &\quad \bv{1}^\top \qvec = 1 - \delta.
    \end{align*}
\end{lemma}
\begin{proof}
    This is an immediate consequence of the minimization formulation of $\dcvar$ presented in \eqref{eq:cvar_definition}; in fact, this is the Lagrange dual of that formulation, applied to the definition of $\dcr$ for continuous-time ski rental. This is also the particular case of the supremum form of $\dcvar$ in \eqref{eq:cvar_definition} which computes the expected cost on the worst ``$(1-\delta)$-sized subpopulation of the distribution'' $\pvec$ when the loss random variable takes discrete outcomes; the optimal solution $\qvec$ is obtained by ``filling in'' $\pvec$ starting from the indices with highest cost, i.e., starting with $i$, then $i - 1$, through $1$, and then starting again with the indices $i+1$ through $B$, until the probability budget $1-\delta$ has been depleted. This structure of the optimal solution $\qvec$ can also be obtained from the characterization of $\dcvar$ for discrete random variables provided in \cite[Proposition 8]{rockafellarConditionalValueatriskGeneral2002}.
\end{proof}

We also prove another technical lemma asserting that any if $p_i = 0$ for some index $i$, this must introduce slack in $\alpha_\delta^{B, \pvec}(i)$.

\begin{lemma} \label{lemma:discrete_zero_prob_implies_slack}
    Let $\pvec \in \Delta_B$ be an algorithm for discrete-time ski rental with buying cost $B$ that has $\dcr$ $\alpha_\delta^{B, \pvec}$. If $p_i = 0$, then $\alpha_\delta^{B, \pvec}(i) < \alpha_\delta^{B, \pvec}$.
\end{lemma}
\begin{proof}
    Suppose $p_1 = 0$; then we may eliminate the variable $q_1$ and its coefficient in the objective of the maximization form of the $\dcr$ in \eqref{eq:max_form_cvar_discrete}, since $q_1$ must be zero; then it is clear
    \begin{align*}
        \alpha_\delta^{B, \pvec}(1) = \max_{\qvec \in \R^B}&\quad \frac{1}{1-\delta}\left(\sum_{j=2}^B q_j\right) = 1\\
        \mathrm{s.t.} &\quad \bv{0} \leq \qvec \leq \pvec \\
        &\quad \bv{1}^\top \qvec = 1 - \delta,
    \end{align*}
    which is strictly less than $\alpha_\delta^{B, \pvec}$, since we always have the ordering $\alpha_\delta^{B, \pvec} \geq \alpha_0^{B, *} = \frac{1}{1-(1-1/B)^B} > 1$.

    Alternatively, suppose that $i > 1$ and $p_i = 0$. If it is also the case that $p_j = 0$ for all $j < i$, then the prior argument holds and we once again have $\alpha_\delta^{B, \pvec}(i) = 1$. Otherwise, at least one $p_j > 0$ for $j < i$. Since $p_i = 0$, we may eliminate the variable $q_i$ and its coefficient in the objective of the maximization form of the $\dcr$ in \eqref{eq:max_form_cvar_discrete}, since $q_i$ must be zero; thus
    \begin{align*}
        \alpha_\delta^{B, \pvec}(i) = \max_{\qvec \in \R^B}&\quad \frac{1}{1-\delta}\left(\sum_{j=1}^{i-1} \frac{B + j - 1}{i}q_j + \sum_{j=i+1}^B q_j\right) \tageq\label{eq:slack_result_obj_1}\\
        \mathrm{s.t.} &\quad \bv{0} \leq \qvec \leq \pvec \\
        &\quad \bv{1}^\top \qvec = 1 - \delta \\
        <\max_{\qvec \in \R^B}&\quad \frac{1}{1-\delta}\left(\sum_{j=1}^{i-1} \frac{B + j - 1}{i-1}q_j + \sum_{j=i+1}^B q_j\right) \tageq\label{eq:slack_result_obj_2}\\
        \mathrm{s.t.} &\quad \bv{0} \leq \qvec \leq \pvec \\
        &\quad \bv{1}^\top \qvec = 1 - \delta \\
        = \alpha_\delta^{B, \pvec}&(i-1)
    \end{align*}
    where the strict inequality holds due to the fact that, since there is some $p_j > 0$ for $j < i$, the optimal solution of both maximization problems will have some $q_j > 0$, and on this domain the objective of \eqref{eq:slack_result_obj_1} is strictly less than that of \eqref{eq:slack_result_obj_2}. Thus $\alpha_\delta^{B, \pvec}(i) < \alpha_\delta^{B, \pvec}(i-1) \leq \alpha_\delta^{B, \pvec}$.
\end{proof}

Using these lemmas, we can now prove the structural result establishing that the optimal algorithm, as long as it has competitive ratio strictly better than $2 - \frac{1}{B}$, has $\dcr$ independent of the adversary's choice of ski season duration. This is similar in spirit to the ``principle of equality'' in the expected cost setting (see, e.g., \cite{mathieuOnlineAlgorithmsSki}) and our tightness results in continuous time (Lemmas~\ref{lemma:s_1_tight} and \ref{lemma:any_s_tight}).

\begin{lemma} \label{lemma:equality_principle_discrete}
    Let $\delta \in [0, 1)$ be such that the optimal $\dcr$ for ski rental with buying cost $B$ strictly improves on the deterministic optimal, i.e, $\alpha_\delta^{B, *} < 2 - \frac{1}{B}$, and let $\pvec^{B, \delta, *}$ be an algorithm obtaining this optimal $\dcr$. Then
    $$\alpha_\delta^{B, \pvec^{B, \delta, *}}(i) = \alpha_\delta^{B, *}$$
    for all $i \in [B]$.
\end{lemma}
\begin{proof}
    We will abbreviate $\pvec^{B, \delta, *}$ simply as $\pvec$. By assumption that $\alpha_\delta^{B, *} < 2 - \frac{1}{B}$, it must be the case that $p_i < 1-\delta$ for all $i \in [B]$.

    Suppose for the sake of contradiction that there is some slack in the $\dcr$ for some adversary decision $i \in [B]$, so that
    $$\alpha_\delta^{B, \pvec}(i) < \alpha_\delta^{B, *}.$$
    Similar to the structure of the proof for tightness in the continuous-time setting, we will prove this result in two parts: in part (a), we will show that there exists some other distribution $\hat{\pvec} \in \Delta_B$ with at least as good $\dcr$ as $\pvec$ that has slack at time $i = B$, i.e., $\alpha_\delta^{B, \hat{\pvec}} \leq \alpha_\delta^{B, \pvec}$ and $\alpha_\delta^{B, \hat{\pvec}}(B) < \alpha_\delta^{B, *}$. Then, in part (b) we will show that we can redistribute this slack to every other time, i.e., we can construct some other $\tilde{\pvec} \in \Delta_B$ with $\alpha_\delta^{B, \tilde{\pvec}} \leq \alpha_\delta^{B, \hat{\pvec}}$ and $\alpha_\delta^{B, \tilde{\pvec}}(i) < \alpha_\delta^{B, *}$ for all $i \in [B]$, which implies that $\pvec$ is not optimal.

    \paragraph{(a)} Let $i$ be the largest element in $[B]$ with the slack property $\alpha_\delta^{B, \pvec}(i) < \alpha_\delta^{B, *}$; if $i = B$, we may define $\hat{\pvec} = \pvec$ and move to part (b). Otherwise, we have $i < B$. Note that since there is no slack for adversary decisions $i+1, \ldots, B$, Lemma~\ref{lemma:discrete_zero_prob_implies_slack} implies that $p_j > 0$ for all $j = i+1, \ldots, B$. Inspecting the maximization formulation of $\alpha_\delta^{B, \pvec}(i)$ in \eqref{eq:max_form_cvar_discrete}, it is clear that by adding a small constant $\epsilon > 0$ to $p_i$ and subtracting $\epsilon$ from $p_{i+1}$, one can slightly increase $\alpha_\delta^{B, \pvec}(i)$ while strictly decreasing $\alpha_\delta^{B, \pvec}(i+1)$, thus introducing slack at $i+1$. This is because increasing $p_i$ (which, recall, must be strictly less than $1-\delta$) by $\epsilon \leq 1 - \delta - p_i$ will increase the optimal $q_i$ in the maximization form of $\alpha_\delta^{B, \pvec}(i)$ by $\epsilon$, since $q_i$ is associated with the largest coefficient in the objective. However, the budget constraint $\bv{1}^\top \qvec = 1-\delta$ means that some other $q_j$ (or the sum of several $q_j$) must then decrease by $\epsilon$ in the optimal solution, leading to an increase of the optimal value by at most $\epsilon\left(\frac{B+i-1}{i} - 1\right)$. On the other hand, decreasing $p_{i+1}$ by $\epsilon \leq p_{i+1}$ will lead to a corresponding decrease by $\epsilon$ of the optimal $q_{i+1}$ in the maximization form of $\alpha_\delta^{B, \pvec}(i+1)$, since $q_{i+1}$ is associated with the largest coefficient in the objective, and the fact $p_{i+1} < 1-\delta$ means $q_{i+1} = p_{i+1}$ for the optimal $\qvec$. However, since $p_i$ is increased by $\epsilon$, this decrease in $q_{i+1}$ will be absorbed by $q_i$ (or a combination of multiple $q_j$, $j \neq i+1$), which is associated with the second largest coefficient in the objective. Altogether, the optimal value of the problem will decrease by at least $\epsilon\left(\frac{B+i}{i+1} - \frac{B+i-1}{i+1}\right)$, meaning that $\alpha_\delta^{B, \pvec}(i+1)$ has decreased accordingly. Thus if we choose $\epsilon > 0$ satisfying
    $$\epsilon \leq \min\{1 - \delta - p_i, p_{i+1}\} \qquad\text{and}\qquad \alpha_\delta^{B, \pvec}(i) + \epsilon\left(\frac{B+i-1}{i} - 1\right) < \alpha_\delta^{B, *},$$
    the modified distribution obtained from increasing $p_i$ by $\epsilon$ and decreasing $p_{i+1}$ by $\epsilon$ will have slack at both $i$ and $i+1$. By similar reasoning, $\alpha_\delta^{B, \pvec}(j)$ will not be impacted for $j < i$ and will not increase (but might decrease) for $j > i+1$. Repeating this process at $i+1$, $i+2$, and so on, we eventually obtain a distribution $\hat{\pvec}$ with the properties that $\alpha_\delta^{B, \hat{\pvec}} \leq \alpha_\delta^{B, *}$ and $\alpha_\delta^{B, \hat{\pvec}}(B) < \alpha_\delta^{B, *}$.

    \paragraph{(b)} Suppose $\hat{\pvec}$ is a distribution satisfying the properties $\alpha_\delta^{B, \hat{\pvec}} \leq \alpha_\delta^{B, *}$ and $\alpha_\delta^{B, \hat{\pvec}}(B) < \alpha_\delta^{B, *}$, and define a new distribution $\tilde{\pvec}$ by shifting a small $\epsilon$ fraction of the mass on all actions $i < B$ to $B$:
    $$\tilde{p}_i = \begin{cases} (1-\epsilon)\hat{p}_i & \text{for $i < B$} \\ \hat{p}_B + \epsilon\sum_{i=1}^{B-1} \hat{p}_i &\text{for $i = B$.} \end{cases}$$
    Note that this will always result in a different distribution when $\epsilon > 0$, since otherwise $\hat{\pvec}$ must place all probability on the action $B$, violating the assumption of $\alpha_\delta^{B, *} < 2-\frac{1}{B}$.
    Similar to the previous case, we evaluate the impact of this change on the $\dcr$ through inspection of the maximization formula \eqref{eq:max_form_cvar_discrete}. Since we add mass at most $\epsilon$ to the action $B$, the optimal $q_B$, which is associated with the largest coefficient in the objective of $\alpha_\delta^{B, \hat{\pvec}}(B)$, will increase by at most $\epsilon$, compensated by a decrease in the sum of $q_1, \ldots, q_{B-1}$, each of which has a coefficient at least 1. Thus we will have $\alpha_\delta^{B, \tilde{\pvec}}(B) \leq \alpha_\delta^{B, \hat{\pvec}}(B) + \epsilon\left(1 - \frac{1}{B}\right)$. On the other hand, consider $i < B$; if $\hat{p}_i = 0$, Lemma~\ref{lemma:discrete_zero_prob_implies_slack} ensures that there will be slack in $\alpha_\delta^{B, \tilde{\pvec}}(i)$. If instead $\hat{p}_i > 0$, then by a similar argument to the previous part, it holds that $\alpha_\delta^{B, \tilde{\pvec}}(i) < \alpha_\delta^{B, \hat{\pvec}}(i)$, since decreasing all of the nonzero $\hat{p}_i$ entries by a multiplicative factor of $(1-\epsilon)$ will in particular decrease the optimal $q_i$ (and possibly other $q_j$ with $j < i$) by the same factor, while increasing $q_j$ associated with smaller coefficients in the objective of \eqref{eq:max_form_cvar_discrete}. Together, these bounds imply that $\epsilon$ can be chosen such that $\alpha_\delta^{B, \tilde{\pvec}}(i) < \alpha_\delta^{B, \hat{\pvec}}(i) \leq \alpha_\delta^{B, *}$ for all $i \in [B]$, contradicting the assumed optimality of $\pvec$.
\end{proof}
Note that while we cannot validate the $\dcr$ condition $\alpha_\delta^{B, *} < 2 - \frac{1}{B}$ of the above lemma \emph{a priori} without knowledge of the $\dcr$, by part (ii) of the theorem, we can in general use the sufficient condition $\delta < 1 - \frac{c}{\log(B + 1)}$ as a heuristic.

With these lemmas proved, we are now prepared to prove the theorem.

\begin{proof}[Proof of Theorem~\ref{theorem:discrete_analytic_result}]
    We will abbreviate $\pvec^{B, \delta, *}$ simply as $\pvec$. We will prove this result in two parts: first, we show that the proposed $\pvec$, defined as:
    $$p_i = \frac{C}{B}\left(1-\frac{1}{B}\right)^{B-i}$$
    for all $i \in [B]$, obtains the claimed competitive bound of $\frac{C-\delta}{1-\delta}$ for the assumed range of $\delta$, giving an upper bound on the $\dcr$ in this region. After proving this upper bound, we will then prove that $\pvec$ is, in fact, optimal for this region of $\delta$.

    Observe that $p_i$ is increasing in $i$; thus, $p_1 = \frac{C}{B}\left(1 - \frac{1}{B}\right)^{B - 1}$ is the smallest probability assigned by $\pvec$ to any action, and moreover, by the assumption on the range of $\delta$ in the theorem statement, we have $\delta \leq p_1 \leq p_i$ for any $i \in [B]$. Inspecting the maximization form of $\dcvar$ in \eqref{eq:max_form_cvar_discrete}, one can easily observe that an optimal $\qvec$ for $\alpha_\delta^{B, \pvec}(i)$ when $i < B$ will be $q_j = p_j$ for $j < B$ and $q_B = p_B - \delta$, since when $i < B$, $q_B$ is associated with the (smallest) coefficient $1$ in the objective, while if $i = B$, the optimal solution will be $q_1 = p_1 - \delta$ and $q_j = p_j$ for $j > 1$, since $q_1$ is associated with the (smallest) coefficient $1$ in the objective. As a result, the $\dcvar$ for this range of $\delta$ amounts to subtracting $\delta$ from the original competitive ratio and normalizing by $1-\delta$, i.e.,
    $$\alpha_\delta^{B, \pvec}(i) = \frac{1}{1-\delta}\left(\alpha_0^{B, \pvec}(i) - \delta\right).$$
    Since $\alpha_0^{B, \pvec}(i) = C = \frac{1}{1-(1-1/B)^B}$ for all $i \in [B]$, we obtain the claimed $\dcr$: $\alpha_\delta^{B, \pvec} = \frac{C - \delta}{1-\delta}$.

    Now, we turn to proving that this is the optimal $\dcr$ for $\delta \leq \frac{C}{B}\left(1 - \frac{1}{B}\right)^{B - 1}$; henceforth, $\pvec^\delta$ will refer to the algorithm with optimal $\dcr$, which is presumed to be unknown. It must be the case that $\alpha_\delta^{B, \pvec^\delta}$ is continuous in $\delta$; this is because $\alpha_{\delta+\epsilon}^{B, \pvec^\delta} - \alpha_\delta^{B, \pvec^\delta}$ is bounded by a linear function of $\epsilon$, since there exist optimal $\qvec^{\delta+\epsilon}, \qvec^{\delta}$ in the corresponding optimization formulations \eqref{eq:max_form_cvar_discrete} satisfying $\|\qvec^{\delta+\epsilon} - \qvec^{\delta}\|_1 \leq \epsilon$ due to the structure of the optimal solution (i.e., its ``filling in'' $p_i$ associated with larger costs first, and indifference between $p_j$ associated with the coefficient $1$). This in turn implies that $\pvec^\delta$ ought to be continuous as a function of $\delta$ when $\alpha_\delta^{B, \pvec^\delta} < 2 - \frac{1}{B}$; if this were not the case, i.e., if $\pvec^\delta$ were discontinuous at some $\delta$, then letting $i$ be the smallest index such that $p_i^\delta$ is discontinuous at $\delta$, it is clear from inspection of the maximization form \eqref{eq:max_form_cvar_discrete} that this would introduce a discontinuity in $\alpha_\delta^{B, \pvec^\delta}$, yielding a contradiction whenever the $\dcr$ equality (Lemma~\ref{lemma:equality_principle_discrete}) is supposed to hold. Note that when $\delta \leq \frac{C}{B}\left(1 - \frac{1}{B}\right)^{B - 1}$, $\alpha_\delta^{B, *} \leq \frac{C-\delta}{1-\delta} < 2$, since $\frac{C-\delta}{1-\delta}$ is increasing in $\delta$ and
    \begin{align*}
        \frac{C-\delta}{1-\delta}\Bigg|_{\delta = \frac{C}{B}\left(1 - \frac{1}{B}\right)^{B - 1}} &= \left(\frac{(B-1)^2}{1+\left(1-\frac{1}{B}\right)^B - B} + B\right)^{-1},
    \end{align*}
    which is increasing in $B$ and has limit $\frac{e}{e-1}$ as $B \to \infty$. Thus by Lemma~\ref{lemma:equality_principle_discrete}, the tightness property must hold on the specified domain of $\delta$: $\alpha_\delta^{B, \pvec^\delta}(i) = \alpha_\delta^{B, \pvec^\delta}$ for all $i \in [B]$.

    $\pvec^\delta$ continuous in particular implies that, for any $\epsilon > 0$, there is some non-degenerate interval $[0, \hat{\delta}]$ with the property that, for all $\delta \in [0, \hat{\delta}]$, $p_i^\delta \geq p_0^0 - \epsilon = \frac{C}{B}\left(1 - \frac{1}{B}\right)^{B - 1} - \epsilon$. Thus, if we pick $\delta$ sufficiently small, the optimal solution of the maximization form of $\alpha_\delta^{B, \pvec^\delta}(i)$ in \eqref{eq:max_form_cvar_discrete} will be $\qvec = \pvec^\delta - \delta \evec_j$, where $j$ is an index associated with the cost $1$ in the objective, since this is the lowest cost. Defining the cost matrix $\mat{M}$ whose $i$th row contains the cost coefficients when the true ski season duration is $i$:
    $$M_{ij} = \begin{cases} \frac{B + j - 1}{i} & \text{if $j \leq i$} \\ 1 & \text{otherwise,} \end{cases}$$
    we thus have the equation
    $$\frac{1}{1-\delta}\left(\mat{M}\pvec^\delta - \delta\cdot\bv{1}\right) = \alpha_\delta^{B, \pvec^\delta} \cdot \bv{1}$$
    when $\delta$ is sufficiently small. Rearranging, we have
    \begin{equation} \label{eq:discrete_cvar_opt_eqn}
        \mat{M}\pvec^\delta = \left((1-\delta)\alpha_\delta^{B, \pvec^\delta} + \delta\right)\cdot\bv{1}
    \end{equation}
    Now, notice that \eqref{eq:discrete_cvar_opt_eqn} is of the form $\mat{M}\pvec = c \cdot \bv{1}$; this is exactly the equation that arises in the classical discrete-time ski-rental setting when we seek to obtain the optimal algorithm for the (expected cost) competitive ratio (see, e.g., \cite{mathieuOnlineAlgorithmsSki}). Since $\mat{M}$ is invertible, the unique solution is $\pvec = c\cdot \mat{M}^{-1}\bv{1}$, and $c$ must be chosen as $C = \frac{1}{1-(1-\frac{1}{B})^B}$ to ensure normalization of the resulting probability distribution. However, in our setting, this same reasoning implies that $\pvec^\delta = \left((1-\delta)\alpha_\delta^{B, \pvec^\delta} + \delta\right)\cdot \mat{M}^{-1}\bv{1}$, and $(1-\delta)\alpha_\delta^{B, \pvec^\delta} + \delta = C$ to ensure $\pvec^\delta$ is a valid probability distribution. As a result, we obtain $\alpha_\delta^{B, \pvec^\delta} = \frac{C-\delta}{1-\delta}$ and $\pvec^\delta = \pvec^0$, as claimed. Finally, note that the preceding argument was made for $\delta$ small enough that the optimal solution of the maximization form of $\alpha_\delta^{B, \pvec^\delta}(i)$ in \eqref{eq:max_form_cvar_discrete} is of the form $\qvec = \pvec^\delta - \delta \evec_j$, where $j$ is an index associated with the cost $1$ in the objective. Because $\pvec^\delta$ is constant in $\delta$ while this condition holds, this condition is seen to be equivalent to $\delta \leq p_0^0 = \frac{C}{B}\left(1 - \frac{1}{B}\right)^{B - 1}$.
\end{proof}

\section{Proofs and Additional Results for Section~\ref{section:one_max_search} \label{appendix:section:one_max_search}}

\subsection{Proof of Lemma~\ref{lemma:one_max_search_cost_integral_representation} \label{appendix:lemma:one_max_search_cost_integral_representation}}

\begin{proof}[noname]

We begin by calculating an expression for the inverse CDF of the profit random variable $L \cdot \indic_{X > v} + X \cdot \indic_{X \leq v}$ in terms of the inverse CDF of the threshold $X$. As shorthand, we define the cost function $C(X, v) = L \cdot \indic_{X > v} + X \cdot \indic_{X \leq v}$. Then clearly
\begin{align*}
    F_{C(X, v)}(y) = \Prob(C(X, v) \leq y) &= \begin{cases} \Prob(X > v) + \Prob(X \leq y) &\text{if $y \leq v$} \\ 1 & \text{otherwise} \end{cases} \\
    &= \begin{cases} 1 - F_X(v) + F_X(y) &\text{if $y \leq v$} \\ 1 & \text{otherwise} \end{cases}
\end{align*}
and hence, for $t \in [0, 1]$,
\begin{align*}
    F_{C(X, v)}^{-1}(t) &= \inf\{y \in [L, U] : F_{C(X, v)}(y) \geq t\} \\
    &= \begin{cases} L & \text{if $t \leq 1 - F_X(v)$} \\ \inf\{y \in [L, U] : 1 - F_X(v) + F_X(y) \geq t\} &\text{otherwise}\end{cases} \\
    &= \begin{cases} L & \text{if $t \leq 1 - F_X(v)$} \\ F_X^{-1}(t - 1 + F_X(v)) &\text{otherwise.}\end{cases}
\end{align*}
Using the definition of $\dcvar$ as an integral of the inverse CDF \eqref{eq:integral_form_cvar_reward}, we have
\begin{align*}
    \dcvar[C(X, v)] &= \frac{1}{1-\delta} \int_0^{1-\delta} F_X^{-1}(t)\,\der t \\
    &= \begin{cases} L & \text{if $1-\delta \leq 1 - F_X(v)$} \\ \frac{1}{1-\delta}\left[(1-F_X(v))L + \int_{1-F_X(v)}^{1-\delta} F_X^{-1}(t - 1 + F_X(v))\,\der t\right] &\text{otherwise}\end{cases} \\
    &= \begin{cases} L & \text{if $F_X(v) \leq \delta$} \\ \frac{1}{1-\delta}\left[(1-F_X(v))L + \int_{0}^{F_X(v) - \delta} F_X^{-1}(t)\,\der t\right] &\text{otherwise,}\end{cases}
\end{align*}
as claimed.
\end{proof}

\subsection{Proof of Theorem~\ref{theorem:one_max_search_upper_bound} \label{appendix:theorem:one_max_search_upper_bound}}

\begin{proof}[noname]

In this proof, we will suppress the sub- and superscript and simply write $\alpha \coloneqq \alpha_\delta^\theta$. First, note that when $\delta = 1$, the initial condition $\phi(t) = \alpha L = \sqrt{LU}$ holds over the entire interval $[0, 1]$. This is the inverse CDF of the deterministic optimal strategy that always plays the threshold $\sqrt{LU}$, and is hence $\sqrt{\theta}$-competitive; this is easily seen to match the solution to \eqref{theorem:one_max_dcr_ub_root}. In the following, we will restrict to $\delta < 1$.

Let $\phi(t)$ be the solution to the delay differential equation posed in the theorem statement; If we solve it on the region $[\delta, 1]$ by integration, we have
\begin{align*}
    \phi(t) &= \begin{cases}\alpha L & \text{for $t \in [0, \delta]$} \\ \alpha L + \frac{\alpha_\delta^{\theta}}{1-\delta}\int_\delta^t \phi(s-\delta) - L \,\der s &\text{for $t \in (\delta, 1]$}\end{cases} \\
    &=\begin{cases}\alpha L & \text{for $t \in [0, \delta]$} \\ \alpha L -\frac{\alpha_\delta^{\theta}(t-\delta)L}{1-\delta} + \frac{\alpha_\delta^{\theta}}{1-\delta}\int_0^{t-\delta} \phi(s)\,\der s &\text{for $t \in (\delta, 1]$}\end{cases} \\
    &= \begin{cases}\alpha L & \text{for $t \in [0, \delta]$} \\ \frac{\alpha}{1-\delta}\left[(1-t)L + \int_0^{t-\delta} \phi(s)\,\der s\right] &\text{for $t \in (\delta, 1]$.}\end{cases} \tageq\label{eq:one_max_search_phi_integral_form}
\end{align*}
$\phi$ is clearly continuous on $[0, 1]$, and by construction, $\phi(t)$ is also strictly increasing on $[\delta, 1]$, since in this region $\phi'(t) = \frac{\alpha_\delta^{\theta}}{1-\delta}\left[\phi(t-\delta) - L\right] = \frac{\alpha L}{1-\delta}(\alpha - 1) > 0$ since the $\dcr$ cannot be 1 unless the problem is trivial ($\theta = 1$). Thus, assuming $\alpha$ is chosen such that $\phi(1) = U$, we have that $\phi$ is one-to-one on the interval $[\delta, 1]$ and $\phi([\delta, 1]) = \phi([0, 1]) = [\alpha L, U]$.

Now, assume that an algorithm uses $\phi$ as the inverse CDF of its random threshold $X$, and suppose the adversary chooses a sequence with maximal price $\vmax < \alpha L$. In this case, the algorithm will not accept a price during the sequence, since $\phi([0, 1]) = [\alpha L, U]$ implies that $X$ only takes values within the interval $[\alpha L, U]$. In this case, the algorithm earns (deterministic) profit $L$ during the compulsory trade, so the algorithm's $\dcr$ is simply $\frac{\vmax}{L} < \frac{\alpha L}{L} = \alpha$.

On the other hand, suppose the adversary chooses a sequence with maximal price $\vmax \geq \alpha L$. Because $\phi$ is one-to-one on $[\delta, 1]$, $\phi^{-1}$ exactly coincides with the CDF of $X$ on the domain $[\alpha L, U]$, and for $\phi(\phi^{-1}(\vmax)) = \vmax$. Defining $t = \phi^{-1}(\vmax)$, noting that $t \geq \delta$, and applying Lemma~\ref{lemma:one_max_search_cost_integral_representation}, we may compute the $\dcr$:
\begin{align*}
\frac{\vmax}{\dcvar[L \cdot \indic_{X > v} + X \cdot \indic_{X \leq v}]} &= \frac{\phi(t)}{\frac{1}{1-\delta}\left[(1-\phi^{-1}(\vmax))L + \int_{0}^{\phi^{-1}(\vmax) - \delta} \phi(s)\,\der s\right]} \\
&= \frac{\phi(t)}{\frac{1}{1-\delta}\left[(1-t)L + \int_{0}^{t - \delta} \phi(s)\,\der s\right]} \\
&= \alpha,
\end{align*}
where the final equality follows from the integral form of $\phi(t)$ in \eqref{eq:one_max_search_phi_integral_form}. Thus $\phi$, when used as an inverse CDF for the random threshold, yields an algorithm with $\dcr$ $\alpha$.

Now, let us establish an analytic characterization of $\phi$ for $\delta \in [0, 1]$. When $\delta = 0$, the delay differential equation simplifies to an ordinary differential equation $\phi'(t) = \alpha(\phi(t) - L)$ with initial condition $\phi(0) = \alpha L$. Solving this initial value problem yields the solution $\phi(t) = L + (\alpha - 1)L e^{\alpha t}$, which is easily seen to coincide with \eqref{theorem:one_max_analytic_phi} in the $\delta \to 0$ limit, as the sum in \eqref{theorem:one_max_analytic_phi} becomes the Taylor series of the exponential $e^{\alpha t}$.

On the other hand, if $\delta \in (0, 1)$, then we may solve the delay differential equation by integrating step-by-step. That is, suppose we know $\phi(t-\delta)$ exactly on the interval $[k\delta, (k+1)\delta]$ for some $k \in \N$, either by the initial condition or because we have solved for $\phi(t)$ on the previous interval $[(k-1)\delta, k\delta]$. Then we may treat the delay differential equation as an ordinary differential equation on $[k\delta, (k+1)\delta]$ and solve for $\phi(t)$ accordingly.

We claim that on the interval $[k\delta, (k+1)\delta]$, $\phi$ takes the form
\begin{equation} \label{eq:inductive_form_phi_oms}
    \phi(t) = L + (\alpha - 1)L\sum_{j=0}^k \frac{\alpha^j (t-j\delta)^j}{(1-\delta)^j j!}.
\end{equation}
Note that this inductive form of $\phi$ for any nonnegative integer $k$ is equivalent to the form of $\phi$ posited in the theorem statement in \eqref{theorem:one_max_analytic_phi}, as $t \in [k\delta, (k+1)\delta]$ causes all terms with $j \geq k+1$ in \eqref{theorem:one_max_analytic_phi} to vanish.

We establish the validity of this form by induction on $k$, which is exactly the number of step-by-step integrations that must be performed to obtain the solution $\phi$ on the interval $[k\delta, (k+1)\delta]$. As the base case, when $k = 0$, \eqref{eq:inductive_form_phi_oms} simply reduces to the initial condition $\phi(t) = \alpha L$ on the interval $[0, \delta]$. Now suppose that the formula holds for a certain $k$; expressing $\phi(t)$ as an integral of the delay differential equation starting from $(k+1)\delta$, we have that, for $t \in [(k+1)\delta, (k+2)\delta]$,
\begin{align*}
    \phi(t) &= \phi((k+1)\delta) + \frac{\alpha}{1-\delta}\int_{(k+1)\delta}^t \phi(s - \delta) - L\,\der s \\
    &= L + (\alpha - 1)L\sum_{j=0}^k \frac{\alpha^j ((k+1)\delta  -j\delta)^j}{(1-\delta)^j j!} + \frac{\alpha}{1-\delta}\left[-(t - (k+1)\delta)L + \int_{k\delta}^{t-\delta}\phi(s)\,\der s\right] \tageq\label{eq:oms_ub_ind_step_1}\\
    &= L + (\alpha - 1)L\sum_{j=0}^k \frac{\alpha^j ((k+1-j)\delta)^j}{(1-\delta)^j j!} \\
    &\quad + \frac{\alpha}{1-\delta}\left[-(t - (k+1)\delta)L + \int_{k\delta}^{t-\delta} L + (\alpha - 1)L\sum_{j=0}^k \frac{\alpha^j (s-j\delta)^j}{(1-\delta)^j j!}\,\der s\right] \tageq\label{eq:oms_ub_ind_step_2}\\
    &= L + (\alpha - 1)L\sum_{j=0}^k \frac{\alpha^j ((k+1-j)\delta)^j}{(1-\delta)^j j!} + \frac{\alpha}{1-\delta}\left[(\alpha - 1)L\sum_{j=0}^k \frac{\alpha^j (s-j\delta)^{j+1}}{(1-\delta)^j (j+1)!}\right]_{s=k\delta}^{t-\delta} \\
    &= L + (\alpha - 1)L\sum_{j=0}^k \frac{\alpha^j ((k+1-j)\delta)^j}{(1-\delta)^j j!} \\
    &\quad + (\alpha-1)L\left[\sum_{j=0}^k \frac{\alpha^{j+1} (t-(j+1)\delta)^{j+1}}{(1-\delta)^{j+1} (j+1)!} - \sum_{j=0}^k \frac{\alpha^{j+1} ((k-j)\delta)^{j+1}}{(1-\delta)^{j+1} (j+1)!}\right] \\
    &= L + (\alpha - 1)L\sum_{j=0}^k \frac{\alpha^j ((k+1-j)\delta)^j}{(1-\delta)^j j!} \\
    &\quad + (\alpha-1)L\left[\sum_{j=1}^{k+1} \frac{\alpha^{j} (t-j\delta)^{j}}{(1-\delta)^{j} j!} - \sum_{j=1}^{k+1} \frac{\alpha^{j} ((k-(j-1))\delta)^{j}}{(1-\delta)^{j} j!}\right] \\
    &= L + (\alpha - 1)L + (\alpha-1)L\sum_{j=1}^{k+1} \frac{\alpha^{j} (t-j\delta)^{j}}{(1-\delta)^{j} j!} \\
    &= L + (\alpha-1)L\sum_{j=0}^{k+1} \frac{\alpha^{j} (t-j\delta)^{j}}{(1-\delta)^{j} j!}.
\end{align*}
where \eqref{eq:oms_ub_ind_step_1} and \eqref{eq:oms_ub_ind_step_2} follow by the induction hypothesis.
Thus, by induction, we have established that \eqref{theorem:one_max_analytic_phi} is the solution to the delay differential equation \eqref{eq:one_max_search_delay_diffeq}.

Let us now turn to analyzing the competitive ratio $\alpha$. First, note that when $\delta \geq \frac{1}{2}$, we have
\begin{equation} \label{eq:oms_phi_linear_form}
    \phi(t) = \alpha L + \frac{\alpha L}{1-\delta}(\alpha - 1)[t-\delta]^+,
\end{equation}
as all terms in \eqref{theorem:one_max_analytic_phi} with $j > 1$ vanish for $t \leq 1$. As $\alpha$ must be chosen so that $\phi(1) = U$, solving this equation for $\alpha$ yields
$$\alpha L + \alpha (\alpha - 1)L = U \implies \alpha = \sqrt{\theta}.$$
We may also obtain an identical upper bound on $\alpha$ for all $\delta \in [0, 1)$ (and in particular, $\delta > \frac{1}{5}$) by simply lower bounding $\phi$ by \eqref{eq:oms_phi_linear_form}, since every term in the sum in \eqref{theorem:one_max_analytic_phi} is nonnegative:
\begin{align*}
    U = \phi(1) &\geq \alpha L + \frac{\alpha L}{1-\delta}(\alpha - 1)(1-\delta) \\
    \implies \alpha &\leq \sqrt{\theta}.
\end{align*}
This establishes the $\delta > \frac{1}{5}$ case in the analytic bound \eqref{theorem:one_max_dcr_ub}; moreover, this case matches the implicit bound $\overline{r}(\delta)$ defined by \eqref{theorem:one_max_dcr_ub_root} since $\delta > \frac{1}{5}$ implies $\overline{n}(\delta) = \max\left\{1, \left\lfloor \left(\lfloor \delta^{-1}\rfloor - 1\right)/2 \right\rfloor\right\} = 1$, in which case $\overline{r}(\delta)$ is the positive root of the equation $(\overline{r}(\delta) - 1)(\overline{r}(\delta) + 1) = \theta - 1$, i.e., $\sqrt{\theta}$.

On the other hand, suppose that $\delta \leq \frac{1}{5}$. We have that
\begin{align*}
    \theta = \frac{\phi(1)}{L} &= 1 + (\alpha - 1)\sum_{j=0}^\infty \frac{\alpha^j ([1-j\delta]^+)^j}{(1-\delta)^j j!} \\
    &= 1 + (\alpha - 1)\sum_{j=0}^k \frac{\alpha^j (1-j\delta)^j}{(1-\delta)^j j!} &\text{where $k \coloneqq \left\lfloor \delta^{-1} \right\rfloor$} \tageq\label{eq:oms_higher_terms_vanish} \\
    &\geq 1 + (\alpha - 1)\sum_{j=0}^k \frac{\alpha^j (1-\frac{j}{k})^j}{(1-\frac{1}{k})^j j!} \tageq\label{eq:oms_delta_k_ineq} \\
    &= 1 + (\alpha - 1)\sum_{j=0}^k \left(\frac{k-j}{k-1}\right)^j \frac{\alpha^j}{j!} \\
    &\geq 1 + (\alpha - 1)\sum_{j=0}^{\lfloor (k-1)/2\rfloor} \left(\frac{k-j}{k-1}\right)^j \frac{\alpha^j}{j!} \tageq\label{ineq:oms_soln_truncate_half_of_terms}
\end{align*}
where \eqref{eq:oms_higher_terms_vanish} follows from the fact that $j\delta \geq 1$ for all integers $j \geq  \left\lfloor \delta^{-1} \right\rfloor + 1$ and $j\delta \leq 1$ for $j \leq \left\lfloor \delta^{-1} \right\rfloor$, and \eqref{eq:oms_delta_k_ineq} follows due to the fact that $\frac{1-j\delta}{1-\delta}$ is decreasing in $\delta$ for $\delta < 1$ and $\delta \leq \left(\left\lfloor \delta^{-1} \right\rfloor\right)^{-1} = k^{-1}$. Now, we claim that for all $j \in \{0, \ldots, \lfloor (k-1)/2\rfloor\}$, the following inequality holds:

\begin{equation} \label{eq:oms_coeff_ineq}
    \left(\frac{k-j}{k-1}\right)^j \frac{1}{j!} \geq \binom{\lfloor(k-1)/2\rfloor}{j}\frac{1}{\lfloor(k-1)/2\rfloor^j}.
\end{equation}
Note that, by our assumption that $\delta \leq \frac{1}{5}$, we have $k \geq 5$, so $\lfloor(k-1)/2\rfloor \geq 1$ and the right-hand side of \eqref{eq:oms_coeff_ineq} is well-defined. To see that this inequality holds, we will bound the ratio between the right-hand side of \eqref{eq:oms_coeff_ineq} and the left-hand side above by 1. Calling this ratio $R(j)$, observe that
\begin{align*}
    R(j) = \frac{\binom{\lfloor(k-1)/2\rfloor}{j}\frac{1}{\lfloor(k-1)/2\rfloor^j}}{\left(\frac{k-j}{k-1}\right)^j \frac{1}{j!}} &= \frac{\frac{\lfloor(k-1)/2\rfloor!}{(\lfloor(k-1)/2\rfloor-j)!\lfloor(k-1)/2\rfloor^j}}{\left(\frac{k-j}{k-1}\right)^j} \\
    &= \frac{(k-1)^j \prod_{i=0}^{j-1} (\lfloor(k-1)/2\rfloor - i)}{(k-j)^j\lfloor(k-1)/2\rfloor^j},
\end{align*}
from which it is clear that $R(0) = 1$, $R(1) = 1$, and
$$R(j) = R(j-1) \cdot \frac{(k-j+1)^{j-1}(k-1)(\lfloor(k-1)/2\rfloor - j + 1)}{(k-j)^j\lfloor(k-1)/2\rfloor}$$
for $j \geq 2$. Thus, if we can prove that $\frac{(k-j+1)^{j-1}(k-1)(\lfloor(k-1)/2\rfloor - j + 1)}{(k-j)^j\lfloor(k-1)/2\rfloor} \leq 1$ for each $j \in \{2, \ldots, \lfloor (k-1)/2\rfloor\}$, induction will yield the desired property that $R(j) \leq 1$ for all such $j$. Thus we compute:
\begin{align*}
    \frac{(k-j+1)^{j-1}(k-1)(\lfloor(k-1)/2\rfloor - j + 1)}{(k-j)^j\lfloor(k-1)/2\rfloor} &= \left(1 + \frac{1}{k-j}\right)^{j-1}\frac{k-1}{k-j} \cdot \frac{\lfloor(k-1)/2\rfloor - j + 1}{\lfloor(k-1)/2\rfloor} \\
    &\leq \frac{1}{1 - \frac{j-1}{k-j}} \cdot \frac{k-1}{k-j} \cdot \frac{\lfloor(k-1)/2\rfloor - j + 1}{\lfloor(k-1)/2\rfloor} \tageq\label{eq:oms_analysis_bernoulli_ineq} \\
    &= \frac{k - 1}{k - 2j + 1} \cdot \frac{\lfloor(k-1)/2\rfloor - j + 1}{\lfloor(k-1)/2\rfloor} \\
    &\leq \frac{k - 1}{k - 2j + 1} \cdot \frac{(k-1)/2 - j + 1}{(k-1)/2} \tageq\label{eq:oms_analysis_remove_floor} \\
    &= 1
\end{align*}
where \eqref{eq:oms_analysis_bernoulli_ineq} follows by applying a version of Bernoulli's inequality (e.g., \citep[Chapter 3, (1.2)]{mitrinovicClassicalNewInequalities1993}) to $\left(1 + \frac{1}{k-j}\right)^{-(j-1)}$ and \eqref{eq:oms_analysis_remove_floor} follows from the fact that $1-\frac{1}{x}$ is increasing in $x$ for $x > 0$. Thus, we have established that $R(j) \leq 1$ for all $j \in \{0, \ldots, \lfloor (k-1)/2\rfloor\}$, and hence \eqref{eq:oms_coeff_ineq} holds for all such $j$. Applying this inequality to \eqref{ineq:oms_soln_truncate_half_of_terms}, we obtain
\begin{align*}
    \theta &\geq 1 + (\alpha - 1)\sum_{j=0}^{\lfloor (k-1)/2\rfloor} \left(\frac{k-j}{k-1}\right)^j \frac{\alpha^j}{j!} \\
    &\geq 1 + (\alpha - 1)\sum_{j=0}^{\lfloor (k-1)/2\rfloor} \binom{\lfloor(k-1)/2\rfloor}{j}\frac{\alpha^j}{\lfloor(k-1)/2\rfloor^j} \\
    &= 1 + (\alpha - 1)\left(1 + \frac{\alpha}{\lfloor(k-1)/2\rfloor}\right)^{\lfloor(k-1)/2\rfloor}, \tageq\label{eq:oms_ub_binom_theorem}
\end{align*}
where \eqref{eq:oms_ub_binom_theorem} follows from the binomial theorem. Note that, since we have assumed $\delta \leq \frac{1}{5}$ in this case, $\overline{n}(\delta)$ as defined in the theorem is equal to $\lfloor(k-1)/2\rfloor$, which is always at least $2$. Thus, observing that \eqref{eq:oms_ub_binom_theorem} is increasing in $\alpha$ when $\alpha > 0$, we immediately obtain that the unique positive solution $\overline{r}(\delta)$ to the equality \eqref{theorem:one_max_dcr_ub_root} is an upper bound on $\alpha$.

All that remains to be shown is the $\delta \todown 0$ case in the analytic bound \eqref{theorem:one_max_dcr_ub}; to facilitate this case, we prove the following lemma characterizing the asymptotic behavior of solutions to equations of the general form \eqref{theorem:one_max_dcr_ub_root} as $\delta$ becomes small.

\begin{lemma} \label{lemma:asymptotics_root_equation}
    Let $r(\delta)$ be the unique positive solution to the equation
    \begin{equation} \label{eq:lemma_root_equation}
        (r(\delta) - 1)\left(1 + \frac{r(\delta)}{n(\delta)}\right)^{n(\delta)} = \theta - 1,
    \end{equation}
    where $n(\delta)$ is a function satisfying
    $$c_1\delta \leq \frac{1}{n(\delta)} \leq c_2\delta$$
    for all $\delta \in (0, 1]$, given some $c_2 \geq c_1 > 0$. Then $r(\delta) = 1 + W_0\left(\frac{\theta - 1}{e}\right) + \Theta(\delta)$, where the asymptotic notation reflects the $\delta \todown 0$ regime and omits dependence on $\theta$.
\end{lemma}

The proof of Lemma~\ref{lemma:asymptotics_root_equation} relies on the bounds in the following lemma.

\begin{lemma} \label{lemma:exponential_inequalities}
    Let $C \geq 1$ and and $n \in \R_{++}$. Then
    \begin{equation} \nonumber
        e^{\frac{\log(1+C)}{C}x} \leq \left(1+\frac{x}{n}\right)^n \leq e^{x - \frac{C - \log(1+C)}{C^2}\frac{x^2}{n}}
    \end{equation}
    for all $x \in [0, nC]$.
\end{lemma}
\begin{proof}
    We begin by showing that, for $x \in [0, C]$, the inequalities
    \begin{equation} \label{ineq:log_bounds}
        \frac{\log(1+C)}{C}x \leq \log(1+x) \leq x - \frac{C - \log(1+C)}{C^2}x^2
    \end{equation}
    hold. First, note that the first inequality in \eqref{ineq:log_bounds} holds due to the fact that $\log(1+x)$ is concave and agrees in value with $\frac{\log(1+C)}{C}x$ at both $x = 0$ and $C$. For the second inequality, first observe that $\log(1+x)$ and $x - \frac{C - \log(1+C)}{C^2}x^2$ agree in value at $x = 0$ and $C$. Consider their difference
    $$f(x)= x - \frac{C - \log(1+C)}{C^2}x^2 - \log(1+x);$$
    we will show that $f(x)$ is nonnegative on $[0, C]$. Computing its derivative, we have
    $$f'(x) = x\left(-\frac{2}{C} + \frac{1}{1+x} + \frac{2\log(1+C)}{C^2}\right)$$
    which has roots
    $$x_1 = 0, \quad x_2 = \frac{C^2 + 2\log(1+C) - 2C}{2(C - \log(1+C))},$$
    where we can observe $x_2 > 0$, since $C \geq 1$ implies the well-known bounds $C - \frac{1}{2}C^2 < \log(1+C) < C$. Thus, if we can show that $f'(x) > 0$ for all $x \in (0, x_2)$, this will establish the desired property that $f(x) \geq 0$ for all $x \in [0, C]$. Given that $f'$ is continuous on the nonnegative reals and only has roots at $0$ and $x_2$, it suffices to show that $f'$ is positive for some small $\epsilon > 0$. Noting that $C - \frac{1}{2}C^2 < \log(1+C)$ implies $C - \frac{1}{2}C^2 + \delta = \log(1+C)$ for some $\delta > 0$, it follows that
    \begin{align*}
        f'(x) &= x\left(-\frac{2}{C} + \frac{1}{1+x} + \frac{2\log(1+C)}{C^2}\right) \\
        &= x\left(\frac{2\log(1+C) - 2C}{C^2} + \frac{1}{1+x}\right) \\
        &= x\left(\frac{-C^2 + 2\delta}{C^2} + \frac{1}{1+x}\right) \\
        &= x\left(-1 + \frac{2\delta}{C^2} + \frac{1}{1+x}\right),
    \end{align*}
    which can be made strictly positive by choosing $x > 0$ sufficiently small, thus establishing the bound.

    Multiplying \eqref{ineq:log_bounds} by $n$, exponentiating, and making the substitution $x \leftarrow \frac{y}{n}$, we obtain the bounds
    \begin{equation} \nonumber
        e^{\frac{\log(1+C)}{C}y} \leq \left(1+\frac{y}{n}\right)^n \leq e^{y - \frac{C - \log(1+C)}{C^2}\frac{y^2}{n}}
    \end{equation}
    for $\frac{y}{n} \in [0, C]$, i.e., $y \in [0, nC]$.
\end{proof}

With Lemma~\ref{lemma:exponential_inequalities} proved, we may now proceed with the proof of Lemma~\ref{lemma:asymptotics_root_equation}.

\begin{proof}[Proof of Lemma~\ref{lemma:asymptotics_root_equation}]
We begin by establishing coarse lower and upper bounds on the solution $r(\delta)$ to \eqref{eq:lemma_root_equation}. First, note that $\theta = 1$ (the trivial case when all prices are identical) yields the unique positive solution $r(\delta) = 1 = \sqrt{\theta} = 1 + W_0\left(\frac{\theta - 1}{e}\right)$, regardless of the value of $n(\delta)$. On the other hand, suppose $\theta > 1$. It is clear that $n(\delta) = 1$ implies $r(\delta) = \sqrt{\theta}$ and the $\delta \todown 0$ limit (equivalently $n(\delta) \to \infty$, by the assumed bounds on $n(\delta)$) yields $r(0) = 1 + W_0\left(\frac{\theta - 1}{e}\right)$.
In addition, observe that the left-hand side of \eqref{eq:lemma_root_equation} is continuous and strictly increasing in both $r(\delta)$ and $n(\delta)$ when $r(\delta) > 1$ and $n(\delta) > 0$. Thus, increasing $n(\delta)$ from $1$ must yield a decrease in $r(\delta)$. As a result, we must have $r(\delta) \in \left[1 + W_0\left(\frac{\theta - 1}{e}\right), \sqrt{\theta}\right]$ for all $\delta \in [0, 1]$.

Now, we continue on to prove the main result. We break the proof into two parts: the upper bound and the lower bound. In the following, we omit the dependence of $r(\delta)$ and $n(\delta)$ on $\delta$, simply writing $r$ and $n$, respectively.

\paragraph{Upper bound.} Since $r \in \left[1 + W_0\left(\frac{\theta - 1}{e}\right), \sqrt{\theta}\right]$, we may apply the lower bound in Lemma~\ref{lemma:exponential_inequalities} with $C = c_2\delta\sqrt{\theta}$ to \eqref{eq:lemma_root_equation} to obtain
\begin{align*}
    \theta - 1 &= \left(r - 1\right)\left(1 + \frac{r}{n}\right)^{n} \\
    &\geq \left(r - 1\right)\left(1 + \frac{r}{\frac{1}{c_2\delta}}\right)^{\frac{1}{c_2\delta}} \tageq\label{ineq:ub_root_n_lb}\\
    &\geq (r-1)e^{\frac{\log(1+c_2\delta\sqrt{\theta})}{c_2\delta\sqrt{\theta}}r}, \tageq\label{ineq:ub_root_pre_lambert}
\end{align*}
where \eqref{ineq:ub_root_n_lb} follows by the monotonicity of the left-hand side of \eqref{eq:lemma_root_equation} in $n$. Monotonicity of \eqref{ineq:ub_root_pre_lambert} in $r > 1$ and the definition of the Lambert $W$ function yields the bound
$$\frac{\log(1+c_2\delta\sqrt{\theta})}{c_2\delta\sqrt{\theta}}\left(r-1\right) \leq W_0\left(\frac{\log(1+c_2\delta\sqrt{\theta})}{c_2\delta\sqrt{\theta}}\cdot\frac{\theta - 1}{\exp\left(\frac{\log(1+c_2\delta\sqrt{\theta})}{c_2\delta\sqrt{\theta}}\right)}\right),$$
and hence
$$r \leq 1 + \frac{c_2\delta\sqrt{\theta}}{\log(1+c_2\delta\sqrt{\theta})}\cdot W_0\left(\frac{\log(1+c_2\delta\sqrt{\theta})}{c_2\delta\sqrt{\theta}}\cdot\frac{\theta - 1}{\exp\left(\frac{\log(1+c_2\delta\sqrt{\theta})}{c_2\delta\sqrt{\theta}}\right)}\right).$$
Taylor expanding about $\delta = 0$ gives:
$$r \leq 1 + W_0\left(\frac{\theta - 1}{e}\right) +\frac{c_2\sqrt{\theta}\cdot W_0\left(\frac{\theta - 1}{e}\right)}{2}\delta + \calO(\delta^2),$$
i.e., $r = 1 + W_0\left(\frac{\theta - 1}{e}\right) + \calO(\delta)$ as $\delta \todown 0$.

\paragraph{Lower bound.} Since $r \in \left[1 + W_0\left(\frac{\theta - 1}{e}\right), \sqrt{\theta}\right]$, we may apply the upper bound in Lemma~\ref{lemma:exponential_inequalities} with $C = c_1\delta\sqrt{\theta}$ to \eqref{eq:lemma_root_equation} to obtain
\begin{align*}
    \theta - 1 &= (r-1)\left(1 + \frac{r}{n}\right)^n \\
    &\leq (r-1)\left(1 + \frac{r}{\frac{1}{c_1\delta}}\right)^{\frac{1}{c_1\delta}} \tageq\label{ineq:lb_root_n_ub} \\
    &\leq (r-1)e^{r - \frac{c_1\delta\sqrt{\theta} - \log\left(1+c_1\delta\sqrt{\theta}\right)}{(c_1\delta\sqrt{\theta})^2}c_1\delta r^2} \\
    &\leq (r-1)e^{r - \frac{c_1\delta\sqrt{\theta} - \log\left(1+c_1\delta\sqrt{\theta}\right)}{(c_1\delta\sqrt{\theta})^2}c_1\delta\left(1 + W_0\left(\frac{\theta - 1}{e}\right)\right)^2} \tageq\label{ineq:lb_root_pre_lambert}
\end{align*}
where \eqref{ineq:lb_root_n_ub} follows by the monotonicity of the left-hand side of \eqref{eq:lemma_root_equation} in $n$ and \eqref{ineq:lb_root_pre_lambert} results from $r \geq 1 + W_0\left(\frac{\theta - 1}{e}\right)$ and the fact that $C - \log(1+C) \geq 0$ for $C \geq 0$. Following the same approach as employed in the upper bound, monotonicity of \eqref{ineq:lb_root_pre_lambert} in $r > 1$ and the definition of the Lambert $W$ function yields the lower bound
$$r \geq 1 + W_0\left((\theta - 1)\exp\left[\frac{c_1\delta\sqrt{\theta} - \log\left(1+c_1\delta\sqrt{\theta}\right)}{(c_1\delta\sqrt{\theta})^2}c_1\delta\left(1 + W_0\left(\frac{\theta - 1}{e}\right)\right)^2 - 1\right]\right),$$
and Taylor expanding about $\delta = 0$ gives
$$r \geq 1 + W_0\left(\frac{\theta - 1}{e}\right) + \frac{c_1 \cdot W_0\left(\frac{\theta - 1}{e}\right)\left(1 + W_0\left(\frac{\theta - 1}{e}\right)\right)}{2}\delta + \Omega(\delta^2),$$
i.e., $r = 1 + W_0\left(\frac{\theta - 1}{e}\right) + \Omega(\delta)$ as $\delta \todown 0$.
\end{proof}

Having proved Lemma~\ref{lemma:asymptotics_root_equation}, the $\delta \todown 0$ case in the analytic bound \eqref{theorem:one_max_dcr_ub} follows as an immediate consequence of the fact that, when $\delta \in (0, 1]$, $\overline{n}(\delta) = \max\left\{1, \left\lfloor \left(\lfloor \delta^{-1}\rfloor - 1\right)/2 \right\rfloor\right\}$ can be upper bounded as
\begin{align*}
    \max\left\{1, \left\lfloor \left(\lfloor \delta^{-1}\rfloor - 1\right)/2 \right\rfloor\right\} &\leq \max\left\{1, \delta^{-1}\right\} \\
    &= \delta^{-1}
\end{align*}
and lower bounded as
\begin{align*}
    \max\left\{1, \left\lfloor \left(\lfloor \delta^{-1}\rfloor - 1\right)/2 \right\rfloor\right\} &= \begin{cases} \left\lfloor \left(\lfloor \delta^{-1}\rfloor - 1\right)/2 \right\rfloor & \text{if $\delta \leq \frac{1}{5}$} \\ 1 & \text{otherwise}\end{cases} \\
    &\geq \begin{cases} \left\lfloor 2\lfloor \delta^{-1}\rfloor/5 \right\rfloor & \text{if $\delta \leq \frac{1}{5}$} \\ 1 & \text{otherwise}\end{cases} \tageq\label{ineq:floor_lb_1} \\
    &\geq \begin{cases} \left\lfloor \frac{8}{25} \delta^{-1}\right\rfloor & \text{if $\delta \leq \frac{1}{5}$} \\ 1 & \text{otherwise}\end{cases} \tageq\label{ineq:floor_lb_2} \\
    &\geq \begin{cases} \frac{3}{25}\delta^{-1} & \text{if $\delta \leq \frac{1}{5}$} \\ 1 & \text{otherwise}\end{cases} \tageq\label{ineq:floor_lb_3} \\
    &\geq \frac{3}{25}\delta^{-1},
\end{align*}
where \eqref{ineq:floor_lb_1}, \eqref{ineq:floor_lb_2}, and \eqref{ineq:floor_lb_3} hold since $\delta \leq \frac{1}{5}$ implies $-1 \geq -\frac{1}{5}\lfloor \delta^{-1}\rfloor \geq -\frac{1}{5}\delta^{-1}$, which in turn implies $\lfloor \delta^{-1}\rfloor \geq \delta^{-1} - 1 \geq \frac{4}{5}\delta^{-1}$ and $\left\lfloor \frac{8}{25} \delta^{-1}\right\rfloor \geq \frac{8}{25} \delta^{-1} - 1 \geq \frac{3}{25}\delta^{-1}$. This concludes the proof.
\end{proof}

\subsection{Proof of Theorem~\ref{theorem:oms_lower_bound} \label{appendix:theorem:oms_lower_bound}}

\begin{proof}[noname]

The proof of this lower bound follows by establishing a connection between one-max search with the $\dcr$ metric and the problem of deterministic \emph{$k$-max search}, which is a modified form of one-max search in which an agent seeks to sell $k$ units of an item, rather than a single one. \cite{lorenzOptimalAlgorithmsKSearch2009} gives a threshold-based algorithm for $k$-max search which uses $k$ distinct, increasing price thresholds, with the agent selling its $i$th item at the first price surpassing the $i$th threshold. We construct our lower bound by comparing the quantiles of an arbitrary randomized algorithm for one-max search to the price thresholds of \cite{lorenzOptimalAlgorithmsKSearch2009} for $\lfloor \delta^{-1} \rfloor$-max search. We will assume without loss of generality that $L = 1$ and $U = \theta$. In the following, we define $k \coloneqq \underline{n}(\delta) = \max\left\{1, \lceil \delta^{-1} \rceil - 1\right\}$ and write the solution to \eqref{eq:one_max_dcr_root_lb} as $r \coloneqq \underline{r}(\delta)$ for clarity. That is, $r$ is the unique positive solution to the equation
\begin{equation} \label{eq:oms_lower_bound_fixed_pt_eqn}
    (r - 1)\left(1 + \frac{r}{k}\right)^{k} = \theta - 1.
\end{equation}
Note that \eqref{eq:oms_lower_bound_fixed_pt_eqn} has a unique positive solution $r$ since the left-hand side is strictly increasing in $r$ when $r > 0$.
First, in the $\delta = 0$ case, \eqref{eq:oms_lower_bound_fixed_pt_eqn} becomes
$$(r-1)e^r = \theta - 1,$$
whose solution is exactly $r = 1 + W_0\left(\frac{\theta - 1}{e}\right)$. Moreover, note that the $\delta = 0$ case is exactly the standard case of randomized one-max search with expected cost, in which case \cite{el-yanivOptimalSearchOneWay2001} has shown that the optimal competitive ratio is exactly the unique positive solution to $(r-1)e^r = \theta - 1$, thus establishing the validity of our lower bound for $\delta = 0$.\footnote{\cite{el-yanivOptimalSearchOneWay2001} specifically shows that the solution $r$ to the equation $(r-1)e^r = \theta - 1$ is the optimal competitive ratio for a fractional version of one-max search known as \emph{one-way-trading}, and that randomized one-max search (with expected cost) is equivalent to this fractional version, in the sense that any algorithm for one can be transformed into an algorithm for the other with identical competitive ratio.}

On the other hand, suppose $\delta = 1$; in this case, $k = 1$, and \eqref{eq:oms_lower_bound_fixed_pt_eqn} becomes
$$(r-1)(1+r) = \theta - 1,$$
yielding the positive solution $r = \sqrt{\theta}$, which matches the optimal deterministic strategy for one-max search \citep{el-yanivOptimalSearchOneWay2001}. Since the optimal $\dcr$ coincides with the optimal deterministic $\comprat$ when $\delta = 1$, this establishes the validity of our bound for this case.

Now, consider an arbitrary $\delta \in (0, 1)$; note that, in this case, $k$ simplifies to $k = \lceil \delta^{-1} \rceil - 1$. Define $k$ price thresholds $p_1, \ldots, p_k$ following \cite[Lemma 1]{lorenzOptimalAlgorithmsKSearch2009}:
$$p_i = 1 + (r - 1)\left(1 + \frac{r}{k}\right)^{i-1}$$
for $i \in [k]$. Let $X \sim \mu$ be any random threshold algorithm for one-max search supported on $[1, \theta]$ (recall from Section~\ref{section:problems_studied} that the restriction to such random threshold algorithms is made without loss of generality). For each $i \in [k]$, define $q_i \in [1, \theta]$ as the $i$th $(k+1)$-quantile of $X$:
$$q_i = F_X^{-1}\left(\frac{i}{k+1}\right).$$
By definition of the inverse CDF, for each $i$ we have $\mu[1, q_i] \geq \frac{i}{k+1}$ and $\mu[q_i, \theta] \geq 1 - \frac{i}{k+1}$.

Suppose that $q_i > p_i$ for some $i \in [k]$, and let $i^*$ be the smallest index for which this strict inequality holds. If $i^* = 1$, then we have
$$\mu(p_1, \theta] \geq \mu[q_1, \theta] \geq 1 - \frac{1}{k+1} = 1 - \frac{1}{\lceil \delta^{-1} \rceil} \geq 1 - \delta,$$
so the algorithm assigns a probability mass of at least $1-\delta$ to thresholds strictly greater than $p_1$. Thus its $\dcr$ is lower bounded as
$$\alpha_\delta^{\theta, \mu}(p_1) = \frac{p_1}{\dcvar[1 \cdot \indic_{X > p_1} + X \cdot \indic_{X \leq p_1}]} \geq \frac{p_1}{1} = r.$$

Otherwise, if $i^* > 1$, then we have $q_{i^*} > p_{i^*}$ and $q_{j} \leq p_{j}$ for all $j \in [i^* - 1]$. We define a modified version of the inverse CDF of $X$ as
$$\hat{F}_X^{-1}(t) = \begin{cases} 1 & \text{if $t = 0$} \\ q_j & \text{if $t \in (\frac{j - 1}{k+1}, \frac{j}{k+1}]$ for $j \in [i^* - 1]$} \\ F_X^{-1}(t) & \text{otherwise,} \end{cases}$$
which is the inverse CDF of the modified random variable $\hat{X}$ obtained by moving all the probability mass between $q_{j-1}$ and $q_j$ to $q_j$ for each $j \in [i^* - 1]$, leaving the rest of the distribution alone. Clearly $F_X^{-1}(t) \leq \hat{F}_X^{-1}(t)$ for all $t \in [0, 1]$, since inverse CDFs are increasing and we define $\hat{F}_X^{-1}$ by increasing the value of $F_X^{-1}$ on $(\frac{j - 1}{k+1}, \frac{j}{k+1}]$ to $q_j = F_X^{-1}\left(\frac{j}{k+1}\right)$ for $j \in [i^* - 1]$.

Now, suppose the adversary chooses the maximum price as $p_{i^*}$; then the $\dcvar$ of the algorithm's profit is upper bounded as:
\begin{align*}
    \dcvar[1 \cdot \indic_{X > p_{i^*}} + X \cdot \indic_{X \leq p_{i^*}}] &\leq \cvar_{\frac{1}{k+1}}[1 \cdot \indic_{X > p_{i^*}} + X \cdot \indic_{X \leq p_{i^*}}] \tageq\label{ineq:cvar_monotone} \\
    &= \frac{1}{1-\frac{1}{k+1}}\left[1 - F_X(p_{i^*}) + \int_0^{F_X(p_{i^*}) - \frac{1}{k+1}} F_X^{-1}(t)\,\der t\right] \tageq\label{eq:oms_lb_inv_cdf_rep} \\
    &\leq \frac{k+1}{k}\left[1 - F_X(p_{i^*}) + \int_0^{F_X(p_{i^*}) - \frac{1}{k+1}} \hat{F}_X^{-1}(t)\,\der t\right] \tageq\label{ineq:modified_inv_cdf} \\
    &= \frac{k+1}{k}\left[1 - F_X(p_{i^*}) + \frac{1}{k+1}\sum_{j=1}^{i^*-2} q_j + \left(F_X(p_{i^*}) - \frac{i^* - 1}{k+1}\right)q_{i^* - 1}\right] \tageq\label{eq:sum_form_cvar_oms_lb}\\
    &\leq \frac{k+1}{k}\left[1 - \frac{i^*}{k+1} + \frac{1}{k+1}\sum_{j=1}^{i^*-1} q_j\right] \tageq\label{eq:oms_lb_2nd_to_last_bound} \\
    &\leq \frac{1}{k}\left[k + 1 - i^* + \sum_{j=1}^{i^*-1} p_j \right] \tageq\label{eq:oms_lb_final_bound},
\end{align*}
where \eqref{ineq:cvar_monotone} holds due to the property that $\dcvar$ (in the maximization setting) is decreasing in $\delta$ \citep[Proposition 3.4]{acerbiCoherenceExpectedShortfall2002} and $\delta \geq (\lceil \delta^{-1} \rceil)^{-1} = (k + 1)^{-1}$; \eqref{eq:oms_lb_inv_cdf_rep} follows from Lemma~\ref{lemma:one_max_search_cost_integral_representation} and the fact that $i^* > 1$ implies $F(p_{i^*}) \geq F(q_{i^* - 1}) \geq \frac{1}{k+1}$; \eqref{ineq:modified_inv_cdf} follows from $F_X^{-1} \leq \hat{F}_X^{-1}$; \eqref{eq:sum_form_cvar_oms_lb} follows by the definition of $\hat{F}_X^{-1}$ and the fact that $p_{i^*} \in [q_{i^* - 1}, q_{i^*})$ implies $F_X(p_{i^*}) - \frac{1}{k+1} \in \Big[\frac{i^* - 2}{k+1}, \frac{i^* - 1}{k+1}\Big)$; \eqref{eq:oms_lb_2nd_to_last_bound} is a result of $F_X(p_{i^*}) \leq \frac{i^*}{k+1}$ and $q_{i^* - 1} \geq 1$; and \eqref{eq:oms_lb_final_bound} follows by the assumption that, in this case, $q_j \leq p_j$ for all $j \in [i^* - 1]$. Substituting the definition of $p_j$ into \eqref{eq:oms_lb_final_bound} and simplifying the sum, we thus obtain the lower bound
\begin{align*}
    \alpha_\delta^{\theta, \mu}(p_{i^*}) &= \frac{p_{i^*}}{\dcvar[1 \cdot \indic_{X > p_{i^*}} + X \cdot \indic_{X \leq p_{i^*}}]} \\
    &\geq \frac{k\left(1 + (r - 1)\left(1 + \frac{r}{k}\right)^{i^*-1}\right)}{k + 1 - i^* + \sum_{j=1}^{i^* - 1}1 + (r - 1)\left(1 + \frac{r}{k}\right)^{j-1}} \\
    &= r.
\end{align*}

Finally, consider the case that $q_i \leq p_i$ for all $i \in [k]$. Defining
$$\hat{F}_X^{-1}(t) = \begin{cases} 1 & \text{if $t = 0$} \\ q_j & \text{if $t \in (\frac{j - 1}{k+1}, \frac{j}{k+1}]$ for $j \in [k]$} \\ F_X^{-1}(t) & \text{otherwise,} \end{cases}$$
it is straightforward to see that an argument identical to the previous case gives the following upper bound when the adversary chooses a maximum price of $\theta$:
\begin{align*}
    \dcvar[1 \cdot \indic_{X > \theta} + X \cdot \indic_{X \leq \theta}] &\leq \frac{1}{k} \sum_{j = 1}^k p_j.
\end{align*}
As a result, the $\dcr$ in this case is lower bounded as
\begin{align*}
    \alpha_\delta^{\theta, \mu}(\theta) &= \frac{\theta}{\dcvar[1 \cdot \indic_{X > \theta} + X \cdot \indic_{X \leq \theta}]} \\
    &\geq \frac{k\theta}{\sum_{j=1}^{k}1 + (r - 1)\left(1 + \frac{r}{k}\right)^{j-1}} \\
    &= \frac{r\theta}{1 + (r-1)\left(1 + \frac{r}{k}\right)^{k}} \\
    &= r & \text{by \eqref{eq:oms_lower_bound_fixed_pt_eqn}}.
\end{align*}

Thus, we have established that any random threshold algorithm (and thus \emph{any} algorithm) for one-max search has $\dcr$ at least $r$. For the analytic bounds in \eqref{eq:one_max_dcr_lb}, there are two cases: first, the $\delta \geq \frac{1}{2}$ case follows from the fact that $\delta \geq \frac{1}{2}$ implies $k = \max\left\{1, \lceil \delta^{-1} \rceil - 1\right\} = 1$, in which case $r = \sqrt{\theta}$. Second, the $\delta \todown 0$ case follows as an immediate consequence of Lemma~\ref{lemma:asymptotics_root_equation} upon noting that, for $\delta \in (0, 1]$, $\underline{n}(\delta) = \max\left\{1, \lceil \delta^{-1} \rceil - 1\right\}$ can be  upper bounded as
\begin{align*}
    \max\left\{1, \lceil \delta^{-1} \rceil - 1\right\} &\leq \max\left\{1, \delta^{-1}\right\} = \delta^{-1}
\end{align*}
and lower bounded as
\begin{align*}
    \max\left\{1, \lceil \delta^{-1} \rceil - 1\right\} &= \begin{cases} \lceil \delta^{-1} \rceil - 1 & \text{if $\delta < \frac{1}{2}$} \\ 1 & \text{otherwise} \end{cases} \\
    &\geq \begin{cases} \frac{2}{3}\lceil \delta^{-1} \rceil & \text{if $\delta < \frac{1}{2}$} \\ 1 & \text{otherwise} \end{cases} \tageq\label{ineq:ceil_lb} \\
    &\geq \begin{cases} \frac{1}{2}\delta^{-1} & \text{if $\delta < \frac{1}{2}$} \\ 1 & \text{otherwise} \end{cases} \\
    &\geq \frac{1}{2}\delta^{-1},
\end{align*}
where \eqref{ineq:ceil_lb} follows from the fact that $\delta < \frac{1}{2}$ implies $\lceil \delta^{-1} \rceil \geq 3$, so $-1 \geq -\frac{1}{3}\lceil \delta^{-1} \rceil$. This concludes the proof.
\end{proof}

\end{document}

%% file: preamble.tex
\usepackage{amsmath, amssymb}
\usepackage{amsfonts}
\usepackage{upgreek}
\usepackage{bm}
\usepackage[margin=1in]{geometry}
\usepackage{enumerate}
\usepackage[shortlabels]{enumitem}
\usepackage{etoolbox}
\usepackage{mathtools}
\usepackage{graphicx}
\usepackage{titlesec}
\usepackage{braket}
\usepackage{float}
\usepackage{nicefrac}
\usepackage[multiple]{footmisc}
\usepackage{float}
\usepackage{bbm}
\usepackage{dsfont}
\usepackage{mathrsfs}
\usepackage{adjustbox}
\usepackage{multirow}
\usepackage{makecell}
\usepackage[math]{cellspace}
\usepackage{hyperref}
\usepackage{xcolor}
\hypersetup{
    colorlinks=true,
    linkcolor = {blue},
    citecolor = {teal},
    urlcolor = {blue}
}

\newcommand{\R}{{\mathbb R}}

\newcommand{\Rpp}{{\mathbb{R}_{++}}}

\newcommand{\N}{{\mathbb N}}

\newcommand{\E}{{\mathbb E}}

\newcommand{\Prob}{{\mathbb P}}
\newcommand{\der}{{\mathrm{d}}}

\DeclareMathOperator{\Ei}{Ei}

\DeclareMathOperator{\var}{\mathrm{VaR}}
\DeclareMathOperator{\cvar}{\mathrm{CVaR}}
\DeclareMathOperator{\cvard}{\mathrm{CVaR}_\delta}
\DeclareMathOperator{\comprat}{\mathrm{CR}}
\DeclareMathOperator{\compratbold}{\mathbf{CR}}
\DeclareMathOperator{\dcvar}{\mathrm{CVaR}_\delta}
\DeclareMathOperator{\dcvarbold}{\mathbf{CVaR}_{\bm{\delta}}}
\DeclareMathOperator{\dcr}{\delta\text{-}\mathrm{CR}}
\DeclareMathOperator{\dcrbold}{\bm{\delta}\mathbf{-CR}}

\DeclareMathOperator*{\esssup}{ess\,sup}
\DeclareMathOperator*{\essinf}{ess\,inf}

\newcommand{\vmax}{v_{\max}}

\renewcommand{\vec}[1]{\mathbf{#1}}

\newcommand{\bv}[1]{\mathbf{#1}}

\newcommand{\indic}{\mathds{1}}
\newcommand{\Borel}{\mathscr{B}}

\newcommand{\calI}{{\mathcal{I}}}

\newcommand{\calO}{{\mathcal{O}}}

\newcommand{\opt}{{\textsc{Opt}}}

\newcommand{\alg}{{\textsc{Alg}}}

\newcommand{\pvec}{{\vec{p}}}
\newcommand{\qvec}{{\vec{q}}}

\newcommand{\vvec}{{\vec{v}}}

\newcommand{\xvec}{{\vec{x}}}

\newcommand{\evec}{{\vec{e}}}

\newcommand{\mat}[1]{\mathbf{#1}}

\newcommand{\toup}{%
  \mathrel{\nonscript\mkern-1.2mu\mkern1.2mu{\uparrow}}%
}
\newcommand{\todown}{%
  \mathrel{\nonscript\mkern-1.2mu\mkern1.2mu{\downarrow}}%
}
\newcommand*\tageq{\refstepcounter{equation}\tag{\theequation}}

\AtBeginEnvironment{theorem}{\setlist{before=\leavevmode, nosep}}